\documentclass[11pt]{article}
\usepackage{setspace}
\usepackage{bm}
\usepackage{bbm}
\usepackage{amssymb}
\usepackage{amsmath}
\usepackage{amsthm}
\usepackage{mathtools}
\usepackage{graphicx}
\usepackage[shortlabels]{enumitem}
\usepackage[margin=1in]{geometry}
\usepackage{verbatim}
\usepackage{subcaption}
\usepackage{float}
\usepackage{makecell}
\usepackage[dvipsnames]{xcolor}
\usepackage{threeparttable}
\usepackage{threeparttablex}
\usepackage{amsfonts}%
\usepackage{parskip}
\usepackage{tikz}
\usepackage{multirow}
\definecolor{navy}{RGB}{0,0,128}
\usepackage[round]{natbib}
\usetikzlibrary{decorations.pathreplacing, positioning, matrix}
\usepackage[
  colorlinks=true,
  citecolor=navy,
  linkcolor=navy,
  urlcolor=navy
]{hyperref}
\let\oldcitep\citep
\renewcommand{\citep}[1]{\textcolor{navy}{\textnormal{\oldcitep{#1}}}}

\let\oldcitet\citet
\renewcommand{\citet}[1]{\textcolor{navy}{\textnormal{\oldcitet{#1}}}}

\newcommand{\bb}{\mathbb}

\newcommand{\R}{\bb R}

\newcommand{\Z}{\bb Z}
\newcommand{\G}{\bb G}

\newcommand{\E}{{\bb E}}
\newcommand{\conv}{\mathrm{conv}}
\newcommand{\Cov}{{\mathrm{Cov}}}
\newcommand{\Var}{{\mathrm{Var}}}
\newcommand{\diag}{{\mathrm{diag}}}
\newcommand{\iid}{{\mathrm{iid}}}
\newcommand{\interior}{{\mathrm{int}}}

\newcommand{\proj}{{\mathrm{proj}}}

\definecolor{myLightBlue}{RGB}{50,110,160}

\newcommand{\A}{\mathcal{A}}

\newcommand{\sgn}{\mathrm{sgn}}
\definecolor{ceruleanblue}{rgb}{0.16, 0.32, 0.75}
\newtheoremstyle{spacestyle}
  {15pt}
  {15pt}
  {\itshape}
  {}
  {\bfseries}
  {.}
  { }
  {}
  \newtheoremstyle{remarkstyle}
  {15pt}
  {15pt}
  {\normalfont}
  {}
  {\bfseries}
  {.}
  { }
  {}
  
\theoremstyle{spacestyle}
\newtheorem{theorem}{Theorem}
\newtheorem{definition}{Definition}
\newtheorem{lemma}{Lemma}

\newtheorem{proposition}{Proposition}
\newtheorem{assumption}{Assumption}
\newcounter{subthm}

\renewcommand{\thesubthm}{\mainthmnum-\arabic{subthm}} 

\usepackage{xparse}

\newcounter{subdef}
\newcommand{\mainthmnum}{} 

\renewcommand{\thesubdef}{\mainthmnum-\arabic{subdef}}

\NewDocumentEnvironment{subdefinition}{ o }{%
  \refstepcounter{subdef}%
  \par\noindent
  \textbf{Definition \thesubdef}%
  \IfNoValueF{#1}{\ \textbf{(#1)}}%
  .\\
}{\par}

\theoremstyle{remarkstyle}
\newtheorem{example}{Example}
\newtheorem{remark}{Remark}

\begin{document}

\clearpage
\setcounter{tocdepth}{2} \thispagestyle{empty} \clearpage
\setcounter{page}{1}

\title{\textbf{Is Complexity the Problem?}\\
\textbf{Testing Random Choice with Heterogeneity}\thanks{I am indebted to Alessandra Casella, Michael Woodford, Mark Dean, and Haoge Chang for their guidance and encouragement. I thank Olivier Compte, Francesco Fabbri, Junnan He, Navin Kartik, Paulo Natenzon, Serena Ng, Jacopo Perego, Pietro Tebaldi, Sevgi Yuksel, and Yangfan Zhou for many helpful suggestions. I also thank participants at the Cognition and Decision Lab, CELSS Experimental Lunch, and Micro Theory Colloquium at Columbia, the Graduate Women in Economic Theory conference at Yale, and seminars at Bocconi University and Caltech. Special thanks to John Clithero, Joshua Peterson, and Ekta Prashnani for sharing their experimental data.}}
\author{Shuhua Si\thanks{Department of Economics, Columbia University. Email: \href{mailto:ss5580@columbia.edu}{ss5580@columbia.edu}.}}
\date{\today}
\maketitle

\vspace{-0.5cm}

\begin{abstract}
Economic choices are often stochastic: the same person may make a different choice when facing the same alternatives repeatedly. Standard models assume that the degree of randomness reflects the size of utility differences, but choice inconsistencies could also reflect difficulty comparing alternatives. Recent studies estimate such \textit{comparison difficulty} (or ``complexity'') by fitting choice models on data collected from different subjects under a representative agent assumption. However, pooling such data while assuming homogeneous preference could violate standard models of random choice simply because of heterogeneity across subjects, even in the absence of variation in comparison difficulty. This paper develops a revealed preference framework, \textit{collective rationalizability}, that tests for variation in comparison difficulty while incorporating heterogeneity. The framework characterizes whether violations of standard models can be explained by comparison difficulty alone, heterogeneity alone, or require both. I then provide a statistical test for collective rationalizability and apply the method to two existing experiments. In both cases, heterogeneity alone explains observed failures of stochastic transitivity well, demonstrating that comparison difficulty can be not only theoretically but also empirically confused with heterogeneity in aggregate data.
\end{abstract}
\newpage

\section{Introduction}
Economic choices are often stochastic: the same individual may make different choices when facing the same problem repeatedly. The degree of randomness is conventionally explained by the location of alternatives on a utility scale: each option is assigned a position on a one-dimensional utility scale, and the probability of choosing one option over another depends only on their relative locations on this scale. This postulate, satisfied by widely-used discrete choice models such as logit and probit, is known as \textit{Simple Scalability} \citep{krantz1964scaling}.
Yet not all random choice reflects small utility gaps.
Consider choosing between health insurance plans with different monthly costs, provider networks, and covered services, or comparing job offers that differ in salary, location, and work-life balance. Even when one option is actually better than the other, the comparison itself might be inherently difficult: the stochastic choice need not reflect that the alternatives have similar values but rather that they are difficult to compare. This distinction matters for welfare, as the loss from choice mistakes could be much larger than Simple Scalability implies.  As a recent body of empirical work shows convincingly, accounting for various cognitive and information frictions has important welfare implications in markets featuring complex choices \citep{barahona2023equilibrium, brown2024endogenous, vatter2025quality}. Therefore, a central question is: How can we separately detect the difficulty of comparing alternatives from the alternatives' utility differences?

Recent studies estimate comparison difficulty by specifying functional forms and fitting them to the choice data under a representative agent assumption. The choice data consist of the proportion of total choices favoring option $A$ over $B$ collected from possibly heterogeneous subjects.\footnote{See, for example, \citet{natenzon2019random, he2023random, shubatt2024tradeoffs}.} Researchers rely on such data because it is rarely feasible to obtain rich individual-level choice frequencies, where each person is observed making a large number of repeated decisions over all pairs. The problem, however, is that aggregate data might mask \textit{individual heterogeneity}. A 60\% choice share for option $A$ over $B$ could mean that \textit{every} person chooses $A$ with probability 0.6, or that 60\% of individuals always choose $A$ while the remaining always choose $B$, with no individual randomness, just different tastes. When aggregate data violate Simple Scalability, is it because some options are harder to compare, or because people have heterogeneous preferences?

The main contribution of this paper is a revealed preference framework, \textit{collective rationalizability}, for testing random choice models while allowing for heterogeneity. I derive the collectively rationalizable version of Simple Scalability and two other models incorporating the difficulty of comparing alternatives  (\textit{comparison difficulty}). The characterizations identify whether violations of Simple Scalability can be explained by heterogeneity alone, comparison difficulty alone, or the two together. This framework does not assume any specific functional forms for individual preference or comparison difficulty.
Building on the theoretical foundation, I then present a statistical test for collective rationalizability. Applications to data from two existing experiments \citep{prashnani2018pieapp, clithero2018improving} demonstrate the practical applicability of the method. In both cases, heterogeneity alone explains the observed deviations from Simple Scalability well.

This paper focuses on three nested classes of models proposed in the literature: Simple Scalability \citep{krantz1964scaling}, Moderate Utility \citep{halff1976choice}, and Weak Utility \citep{luce1965preference}. While Simple Scalability assumes all variations in choice probabilities reflect utility differences, Moderate and Weak Utility models introduce variable comparison difficulty, with Weak Utility imposing fewer restrictions on how it varies across pairs.  
To disentangle preference differences from particular difficulty of comparison, existing revealed preference approaches check for \textit{stochastic transitivity} \citep{luce1965preference, tversky1969substitutability, he2024moderate}: if $A$ is chosen more frequently than $B$ (denoted $A\succcurlyeq B$), and $B\succcurlyeq C$, then specific requirements must hold for choices between $A$ and $C$. Under Simple Scalability, a larger utility difference translates into a more consistent choice, and in this example since $A$ and $C$ must be further apart than any intermediate pair, the choice frequency for $A$ over $C$ must exceed that of either of the other two pairs. But with comparison difficulty, the frequency of choice also depends on how hard $A$ and $C$ are to compare, and thus only weaker forms of stochastic transitivity must hold.

Existing revealed preference approaches yield theoretical tests for the presence of comparison difficulty. However, implementing them in practice is challenging,\footnote{See, for example, \citet{coombs1961some, griswold1962choices}, and a review by \citet{rieskamp2006extending}.} because they require an \textit{ideal dataset}: exact choice probabilities for each individual over all binary pairs. With the data typically available to researchers, these approaches face two obstacles. First, violations of Simple Scalability after aggregation might reflect \textit{heterogeneity in preferences} rather than comparison difficulty. Second, it is important to account for \textit{statistical uncertainty}, that whether deviations reflect genuine violations or sampling variation. This paper addresses both challenges by developing testable characterizations of collective rationalizability that explicitly allow for heterogeneity and provide uniformly valid statistical inference.

To clarify these challenges, the first one is \textit{preference heterogeneity}: aggregate patterns may well deviate from Simple Scalability even in the absence of varying comparison difficulty across pairs. An example comes from binary choices of snack foods in \citet{clithero2018improving}'s experiment. The aggregate choice frequency of choosing KitKat over M\&M's is 0.84, of choosing M\&M's over Fruit Snacks is 0.58, and of choosing KitKat over Fruit Snacks is 0.61. These frequencies violate Simple Scalability. Since KitKat is preferred to M\&M's and M\&M's is preferred to Fruit Snacks, KitKat must be strongly preferred to Fruit Snacks: transitivity would require the frequency for KitKat over Fruit Snacks to exceed 0.84, yet we observe only 0.61. This could be taken as a sign of comparison difficulty, that comparing KitKat with Fruit Snacks is more difficult than comparing KitKat with M\&M's. However, such violations can also arise purely from preference heterogeneity without any comparison difficulty. For instance, imagine that around half of the population always choose chocolate snacks while the other half always choose the fruit one, and among the two chocolate candies, most prefer KitKat. Each individual is fully rational, chooses deterministically, and satisfies Simple Scalability, yet the aggregate frequencies appear to violate it.\footnote{Precisely, one such distribution of heterogeneous preferences that generates these aggregate frequencies is:

    \begin{tabular}{ll}
23\%: FruitSnacks $\succ$ KitKat $\succ$ M\&M's & 16\% FruitSnacks $\succ$ M\&M's $\succ$ KitKat \\
58\%: KitKat $\succ$ M\&M's $\succ$ FruitSnacks & 3\%: KitKat $\succ$ FruitSnacks $\succ$ M\&M's
\end{tabular}
}

To address this, I develop a novel characterization of \textit{collective rationalizability} for the three nested models of random choice described above (\autoref{thm:collective}). The key insight is that aggregate choice patterns can be collectively rationalized by a model if and only if they can be expressed as a weighted average of individual choice rules that satisfy the model. The characterizations reveal when violations of Simple Scalability  can be accommodated by heterogeneity and/or comparison difficulty.

To characterize collective rationalizability, one should consider all possible population compositions: the number of types, the preferences of each type, and their population fractions. The core insight is that collective rationalizability is about \textit{convex combinations}: a choice rule can be collectively rationalized if and only if it belongs to the convex hull of the individual rationalizable set. Nonparametric conditions for collective rationalizability follow from a two-step strategy. First, I analytically characterize the set of individually rationalizable choice probabilities. Second, I derive collective rationalizability conditions by taking the convex hull of this individual rationalizable set. This approach generalizes to any random choice model once its individual rationalizable set is characterized.\footnote{See, for example, Appendix \ref{sec:appdx-multinomial}. I discuss collectively rationalizable Simple Scalability under multinomial choice and how it can rationalize behavioral regularities such as choice reversals.}

I offer two equivalent characterizations of collectively rationalizable versions of random choice models. First, building on techniques for modeling disjunctions \citep{balas1971intersection}, the model can be expressed as a system of linear inequalities where membership can be readily checked using convex programming. This technique applies broadly: given any individually rationalizable set, the same method produces testable linear conditions for its collectively rationalizable version. Second, I provide an alternative characterization by identifying a finite list of candidates for extreme points, with the convex hull comprising all convex combinations of these points. The extreme points offer geometric insights and more transparent intuitions, and allow to check how different  models nest within each other. Both characterizations lend themselves readily to statistical testing.

In particular, the characterizations reveal a  difference between comparison difficulty and heterogeneity.  While some aggregate patterns can be rationalized by either explanation, the two are not equivalent: Weak Utility (with comparison difficulty but no heterogeneity) and \textit{Collective} Simple Scalability (with heterogeneity but no comparison difficulty) are  not nested within each other. Suppose, for example, that choice probabilities between $A$ over $B$, $B$ over $C$, and $A$ over $C$ are, in order, $(0.9, 0.9, 0.6)$. A single individual under Weak Utility can rationalize this with preferences $A\succ B\succ C$ but high difficulty comparing $A$ and $C$. However, Collective Simple Scalability cannot: the combination of strong preferences ($A$ vs $B$, $B$ vs $C$) with near-indifference ($A$ vs $C$) cannot arise from aggregating any heterogeneous population where individuals satisfy Simple Scalability and choice probabilities reflect only utility differences. As a second example, consider the choice probabilities $\left(\tfrac{2}{3},\tfrac{2}{3},\tfrac{1}{3}\right)$. This cycle ($A\succ B\succ C$ but $C\succ A$) violates transitivity and cannot be represented by Weak Utility. However, Collective Simple Scalability can rationalize it through heterogeneity: three equal-sized groups with preferences $A\succ B\succ C$, $B\succ C\succ A$, and $C\succ A\succ B$, each choosing deterministically, generate exactly this pattern.\footnote{This example comes from the Condorcet paradox of social choice theory \citep{condorcet1785}, which shows that transitive individual preferences can produce cycles in aggregate preference over binary pairs.}

Beyond heterogeneity, the second challenge is accounting for \textit{statistical variation}: observed aggregate patterns include sampling variation, making it unclear whether deviations from collective rationalizability reflect genuine departures or noise. 
Based on \autoref{thm:collective}'s characterization of which choice patterns are collectively rationalizable,  \autoref{thm:uniform-test} presents a uniformly valid test evaluating whether observed deviations can be attributed to sampling error, accommodating various data generating processes and random assignment schemes between subjects and choice problems.

One challenge for the statistical testing procedure is that the test statistic, defined as the distance from observed choice frequencies to the collectively rationalizable set, has a limit distribution that does not depend continuously on the null hypothesis. This occurs because the collectively rationalizable set is a polytope whose boundary includes edges and vertices where the test statistic exhibits discontinuities. Standard bootstrap fails in such cases \citep{andrews2000inconsistency}. As a solution, I adopt the numerical delta method from \citet{hong2018numerical}, which guarantees uniform validity, that the test maintains its nominal size regardless of where the true parameter lies, under a broad range of data generating processes.

When either heterogeneity or comparison difficulty can explain violations of Simple Scalability in aggregate data, does the model become so flexible that it loses meaningful empirical content? Following \citet{bronars1987power}, \citet{selten1991properties}, and \citet{beatty2011demanding}, I compute the relative size of the predicted region for each model up to $n=5$ options. I find that \textit{Collective} Simple Scalability (allowing only heterogeneity) is always more restrictive than \textit{Representative-agent} Weak Utility (allowing only comparison difficulty). This theoretical restrictiveness can be combined with each model's empirical fit to compare model performance along two dimensions: the frequency of successfully rationalizing \textit{real} data and the frequency of rejecting \textit{random} data. The following empirical applications demonstrate this idea.

Finally, I illustrate the practical applicability of  collective rationalizability using data from two existing experiments: \citet{clithero2018improving}, on binary choices among snack foods, and \citet{prashnani2018pieapp} on perceptual judgments of image similarity. In both applications, aggregate choices violate Simple Scalability. Without accounting for heterogeneity, researchers might interpret these violations as evidence of comparison difficulty. However, the collective rationalizability framework reveals that heterogeneity alone explains the observed patterns better than comparison difficulty does in both ways: it rationalizes more of the actual experimental data while rejecting more randomly generated data. These applications demonstrate that heterogeneity can be confused with comparison difficulty, both theoretically and empirically, highlighting the importance of the collective rationalizability framework.

This paper provides additional results on identifying the ranking of comparison difficulty: if comparison difficulty indeed varies across pairs, can we go beyond and characterize \textit{how} it varies? I provide necessary and sufficient conditions for rationalizability of individual choice probabilities under a given \textit{ordering} of comparison difficulty. By applying the same technique from \citet{balas1971intersection}, I derive the collectively rationalizable version, characterizing whether observed aggregate patterns can be rationalized by a heterogeneous population where each individual shares a given ranking of comparison difficulty across pairs with $n=3$ options. This characterization can help identify if everyone finds some specific two alternatives are hard to compare, even though people might disagree on which one they prefer. These difficult comparisons are precisely where people are most likely to make mistakes and choose options that don't reflect their true preferences.
\vspace{0.3cm}
\paragraph{Related Literature.}
This paper contributes to the growing literature on cognitive foundations of economic decision-making, which examines what makes decisions difficult and how people respond to such difficulty (see \citet{oprea2024review} for a review). Recent work has identified various sources of difficulty \citep{iyengar2000choice, dean2022better, de2024rational, shubatt2024tradeoffs, puri2025simplicity} and studied behavioral responses \citep{woodford2020modeling, arrieta2024procedural, gabaix2025theory}. A key challenge in this literature is separating cognitive difficulty from preferences. Existing approaches include directly eliciting cognitive uncertainty in the lab \citep{enke2023cognitive}, anchoring preferences through experimental design \citep{oprea2024decisions, enke2023quantifying}, and manipulating decision environments in the lab \citep{bernheim2020empirical, kirchler2017effect, gerhardt2016cognitive, porcelli2009acute}.

This paper focuses specifically on comparison difficulty: the cognitive challenge of comparing two alternatives. Previous researchers have used various terms, such as ``similarity'' \citep{tversky1969substitutability}, ``comparability'' \citep{halff1976choice}, and ``complexity'' \citep{shubatt2024tradeoffs} to discuss the same underlying phenomenon. I deliberately use ``comparison difficulty'' to distinguish the concept from other uses of similarity or complexity in the behavioral economics literature, and emphasize the specific type of cognitive difficulty arising from pairwise comparisons.

Within the comparison difficulty literature, one approach relies on objective benchmarks: if preferences are known, then the frequency of choosing the lower-value option indicates comparison difficulty \citep{tversky1969substitutability}. \citet{enke2023quantifying} apply the idea to lottery choices by asking which lottery has higher expected value rather than which lottery subjects choose. Controlling for the expected value differences, the error rate is a sign of difficult comparisons. This paper contributes a revealed preference approach that embody comparison difficulty while incorporating heterogeneity. The framework remains agnostic about individuals' preferences and imposes no functional form assumptions on either preferences or comparison difficulty. This generality enables analysis in domains such as snack foods or images, where natural metrics are absent and experimental designs requiring objective benchmarks are difficult to implement.

Existing revealed preference approaches characterize comparison difficulty by strong, moderate, and weak stochastic transitivity \citep{luce1965preference, tversky1969substitutability, he2024moderate}.  Recent works in random choice \citep{natenzon2019random, he2023random, shubatt2024tradeoffs} model such difficulty of comparing alternatives with specific functional forms and distance-based complexity metrics.
My approach addresses a key challenge faced by researchers testing stochastic transitivity or estimating models functional forms: the lack of exact choice probabilities for each individual and each pair. I modify the existing revealed preference framework to a collective rationalizability version for aggregate data. A critical advantage of this modification is that it explicitly allows for heterogeneity. Although researchers recognize that heterogeneity across subjects can theoretically confound comparison difficulty across options, this paper characterizes which aggregate patterns violating strong stochastic transitivity can be explained by comparison difficulty alone, by heterogeneity alone, or require both. I demonstrate the empirical importance of this distinction using two existing experiments.

More broadly, this paper belongs to the literature on revealed preferences in stochastic choice \citep{block1959random, falmagne1978representation,
mcfadden1981econometric, mcfadden2005revealed,
mcfadden1990stochastic, fishburn1992induced, fishburn1998stochastic}. They ask the question: are the distributions of choices observed from a population of individuals consistent with rational choice theory, i.e., \textit{random utility maximization} (RUM). I ask the question: are the distributions of choices observed from a population of individuals consistent with different random choice models, with or without varying comparison difficulty across pairs.
In this paper, response distributions may arise from taste heterogeneity, stochastic elements in individual choices, or both. 
\citet{kitamura2018nonparametric} develops a statistical test of random utility models for rational demand system (i.e. SARP/GARP). My paper develops a similar test using the recent econometrics results by \citet{fang2019inference, hong2018numerical}.


The remainder of this paper is organized as follows. Section \ref{sec: model} describes the three nested classes of model and develops characterizations for collective rationalizability. Section \ref{sec: test} develops the statistical testing for collective rationalizability,  including the statistical setup, a testing procedure that is uniformly valid, and a Monte Carlo investigation. Section \ref{sec: application} applies the framework to two existing experimental datasets. Section \ref{sec: conclusion} concludes. Details and some additional results are reported in the Appendix.

\newpage
\section{Random Choice and Collective Rationalizability}\label{sec: model}
This section develops a revealed preference framework for testing random choice models while allowing for heterogeneity. The analysis focuses on three nested classes of random choice models that progressively relax restrictions on how choice probabilities vary across different pairs of alternatives. The benchmark model, Simple Scalability, assumes that all variation in choice probabilities reflects only the relative locations of the two alternatives on a one-dimensional utility scale. Moderate Utility and Weak Utility progressively relax this assumption by allowing comparison difficulty to vary across different pairs.
I characterize \textit{collective rationalizability} for the three model classes: choice patterns are collectively rationalizable if they can be expressed as a weighted average of individual choices satisfying the corresponding model.

The section proceeds as follows. Section \ref{sec: models} introduces the model setup and existing revealed preference framework, stochastic transitivity. Section \ref{sec: collective-rational} characterizes collective rationalizability for the three model classes. Section \ref{sec: to-rum} connects collective rationalizability to random utility models. Section \ref{sec: cplx-ranking} discusses identification of the \textit{ranking} of comparison difficulty across pairs.

\subsection{Three Nested Classes of Models and Stochastic Transitivity}\label{sec: models}
Let $Z$ be a \textit{finite} set of choice options. A \textit{binary menu} is a subset of $Z$ with cardinality of 2. Let $\A$  denote the set of binary menus for which choice probabilities have been observed. A \textit{binary stochastic choice rule} $\rho:\A\to[0,1]$ assigns to each menu $(x,y)$ the probability $\rho(x,y)$ that $x$ is chosen over $y$, with $\rho(x,y)+\rho(y,x)=1$ for all $(x,y)\in \A.$ In the main analysis, I assume $\A=Z^2$.

\begin{definition}[\textbf{Random Choice Models}]
\

\begin{enumerate}[(i).]
    \item A \textit{Simple Scalability (SS) representation} has the form \citep{krantz1964scaling}
    \begin{equation}
        \rho(x,y) = F(u(x),u(y)),\quad \forall x,y\in Z
    \end{equation}
    for some utility function $u:Z\to \R$ and a function $F: \R\times \R\to (0,1)$ that is strictly increasing in the first argument and strictly decreasing in the second, satisfying $F(a,b) = 1-F(b,a),\forall a,b\in \R.$\footnote{In SS representation, it is assumed that all choice probabilities are neither 0 or 1. As discussed in \cite{mcfadden1974conditional} and \cite{fudenberg2015stochastic}, the positivity assumption is empirically indistinguishable from allowing for the boundary values.}
    \item A \textit{Moderate Utility (MU) representation} has the form \citep{halff1976choice}
    \begin{equation}
        \rho(x,y) = F\left(\frac{u(x)-u(y)}{c(x,y)}\right),\quad \forall x,y\in Z
    \end{equation}
    for some utility function $u:Z\to \R$, strictly increasing function $F:\R\to [0,1]$ satisfying $F(x)=1-F(-x)$, and distance metric\footnote{A distance metric is a function $c: Z^2\to\R^+ $ satisfying (i) $c(x,y)=0$ if and only if $x=y$; and (ii) symmetry $c(x,y)=c(y,x)$; and (iii) triangle inequality $c(x,z)\le c(x,y)+c(y,z),\forall (x,y), (y,z), (x,z)\in\A$.} $c: Z^2\to \R^+$.
    \item A \textit{Weak Utility (WU) representation} has the form \citep{luce1965preference}
    \begin{equation}
        \rho(x,y) = F\left(\frac{u(x)-u(y)}{c(x,y)}\right),\quad \forall x,y\in Z
    \end{equation}
    for some utility function $u:Z\to \R$, strictly increasing function $F:\R\to [0,1]$, satisfying $F(x)=1-F(-x)$, and semimetric\footnote{A semimetric is a function $c: Z^2\to\R^+ $ satisfying (i) $c(x,y)=0$ if and only if $x=y$; and (ii) symmetry $c(x,y)=c(y,x)$.} $c: Z^2\to \R^+$.
\end{enumerate}
\end{definition}

\begin{remark}\label{remark: binary-relation}
    Under Simple Scalability, Moderate Utility, or Weak Utility, the condition $\rho(x,y)\ge 1/2$ holds if and only if $u(x)\ge u(y)$. Define the preference relation $x\succcurlyeq y$ if and only if $\rho(x,y)\ge 1/2$.
\end{remark}

The key postulate of Simple Scalability is that for a given individual, each alternative is assigned a location on a one-dimensional utility scale, and the degree to which choice probabilities vary between any two alternatives depends only on their relative locations. The term ``Simple Scalability'' reflects this reliance on scalar utility representations.

Many familiar stochastic choice rules are special cases of Simple Scalability with additional structures. For example, the logit model $\rho(x,y) = \frac{e^{u(x)/\lambda}}{e^{u(x)/\lambda}+e^{u(y)/\lambda}}$, and the probit model $\rho(x,y)=\Phi_\sigma(u(x)-u(y))$.\footnote{These models have the additional structure that choice probabilities depend only on utility \textit{differences}, corresponding to the Fechner model \citep{fechner1860elemente} $\rho(x,y)=F(u(x)-u(y))$ for some utility function $u:Z\to \R$ and a function $F:\R\to(0,1)$ that is strictly increasing. For logit, $F$ is the cumulative distribution function (CDF) of a Gumbel distribution with variance $1/\lambda$; for probit, $F$ is the CDF of normal distribution with standard deviation $\sigma$. $\Phi_\sigma$ denotes the cumulative distribution function of the normal distribution normal distribution $N(0,\sigma)$.} More generally, additive random utility models with i.i.d. utility shocks
\begin{equation}
    \rho(x,y) = \Pr[u(x)+\varepsilon_x \ge u(y)+\varepsilon_y],\quad \text{ where }\varepsilon_x,\varepsilon_y \text{ are i.i.d. noise,}
\end{equation}
are examples of Simple Scalability, where stochasticity can stem from uncertainty or noisy reading regarding the true utilities. Alternatively, following \citet{fudenberg2015stochastic}, stochasticity in Simple Scalability can also reflect costly implementation of desired choices. Their additive perturbed utility representation, restricted here to binary menus,
\begin{equation}
    \rho(x,y) = \operatorname*{arg\,max}_{p\in(0,1)}\ \  pu(x)+(1-p)u(y)-c(p),
\end{equation}
where $c$ is a convex perturbation function capturing costs from implementing more precise decisions, also yields an example of Simple Scalability.\footnote{While \citet{fudenberg2015stochastic} develop their model for general menu sizes, I only present the binary case.}

The distinguishing feature of Moderate Utility and Weak Utility representations, compared to Simple Scalability, is the presence of $c(x,y)$ in the denominator. This term captures the difficulty of discriminating between the two options $x$ and $y$, which we define as the \textit{comparison difficulty} of the binary menu $\{x,y\}$.  Crucially, while Simple Scalability assumes that differences in choice randomness reflect only the relative locations of alternatives on a one-dimensional utility scale, the comparison difficulty $c(x,y)$ can depend on any attributes of the alternatives. This allows Moderate Utility and Weak Utility to capture how non-utility features affect  stochastic choice.

For instance, consider a newly-hired assistant professor buying a laptop. Even though evaluating the value of each laptop is not straightforward for her, choosing between two laptops that differ only in color should be relatively straightforward with low randomness in choice. However, choosing between laptops from different brands, with different operating systems, processing speeds, and storage capacities, becomes considerably more difficult, generating greater randomness in choice. This pair-specific noise in the comparison process constitutes what we call comparison difficulty, a dimension distinct from utility differences.

One interpretation of this stochasticity comes from \citet{natenzon2019random}, who models the choice process as Bayesian learning under imperfect information about alternatives. In his framework, the correlation between noise terms depends on the degree of ``similarity'' between options. When alternatives are easier to compare, the decision-maker can extract more precise signals about their true utility differences, leading to less random choices. Conversely, difficult comparisons, such as those requiring tradeoffs across multiple incommensurable attributes, produce noisier signal extraction and more randomness in choice.

The difference between Moderate Utility and Weak Utility lies in whether comparison difficulty  $c(x,y)$ must satisfy the triangle inequality. Under Moderate Utility, if comparisons of $\{x,y\}$, $\{y,z\}$ are both easy, then $x$ and $z$ cannot be arbitrarily difficult to compare. Their comparison difficulty is bounded by the sum of the other two. This creates a coherent structure where comparison difficulties relate to each other systematically,  like distances in physical space. Many existing models that derive comparison difficulty from the distance of attributes or correlation of the noise distribution belong to Moderate Utility \citep{mcfadden1978modeling, natenzon2019random, he2023random,shubatt2024tradeoffs}. Weak Utility relaxes this requirement, allowing comparison difficulties to violate the triangle inequality. It permits more complex patterns where the difficulty of comparing $x$ and $z$ need not be constrained by their respective difficulties with an intermediate option $y$.  Testing whether data satisfies MU versus WU reveals whether comparison difficulty follows a distance-based structure or exhibits more complex, non-distance metric patterns.

It has been well established in the literature that Simple Scalability, Moderate Utility, and Weak Utility representations are equivalent to strong, moderate, and weak stochastic transitivity, respectively, when all binary pairs are observed.

\begin{definition}[\textbf{Stochastic Transitivity}]\label{def:transitivity}
\

A choice rule $\rho$ satisfies \citep{luce1965preference, tversky1969substitutability, he2024moderate}
    \begin{enumerate}[(i).]
        \item \textit{Strong stochastic transitivity} if $\forall x,y,z\in Z$
        \begin{equation}
            \min\{\rho(x,y),\rho(y,z)\}\ge\!\!(>) 1/2 \quad \text{\textup{implies}} \quad \rho(x,z)\ge\!\!(>) \max\{\rho(x,y),\rho(y,z)\}.
        \end{equation}
        \item \textit{Moderate stochastic transitivity} if $\forall x,y,z\in Z$
        \begin{equation}
            \min\{\rho(x,y),\rho(y,z)\}\ge 1/2 \quad \text{\textup{implies}} \quad \left\{ \begin{matrix}
            \rho(x, z) > \min\{\rho(x,y),\rho(y,z)\}  \\
            \mathrm{or} \\
            \rho(x, z) = \rho(x,y)=\rho(y,z)
        \end{matrix} \right. .
        \end{equation}
        \item \textit{Weak stochastic transitivity} if $\forall x,y,z\in Z$
        \begin{equation}
            \min\{\rho(x,y),\rho(y,z)\}\ge 1/2 \quad \text{\textup{implies}} \quad \rho(x,z)\ge 1/2.
        \end{equation}
    \end{enumerate}
\end{definition}

The equivalence arises because stochastic transitivity conditions impose restrictions on how choice probabilities can vary across different triples of alternatives. Strong stochastic transitivity requires the most restrictive pattern: if $x\succcurlyeq y$ and $y\succcurlyeq z$, then $x$ must be chosen over $z$ with probability at least $\max\{\rho(x,y), \rho(y,z)\}$. This constraint is precisely what Simple Scalability implies when all variation in choice probabilities reflects only the relative locations on a one-dimensional scale: the distance between $x$ and $z$ must be at least as large as that of the intermediate pairs. Moderate and weak stochastic transitivity progressively relax this restriction by weakening the bound on $\rho(x,z)$. In particular, weak stochastic transitivity is exactly the transitivity of the preference relation: if $x\succcurlyeq y$ and $y\succcurlyeq z$, then $x\succcurlyeq z$.

\subsection{Collective Rationalizability}\label{sec: collective-rational}
The preceding analysis characterizes when choices are rationalizable by Simple Scalability, Moderate Utility, or Weak Utility \textit{without} heterogeneity. In practice, however, researchers often work with data collected from possibly heterogeneous subjects but have to assume representative agent. The classic Condorcet paradox \citep{condorcet1785} illustrates this issue. Consider three individuals with different preference orderings: $x\succ y\succ z$, $y\succ z\succ x$, and $z\succ x\succ y$. If each individual chooses deterministically according to their preferences, the representative-agent choice probabilities are $\rho(x,y)=\tfrac{2}{3},\rho(y,z)=\tfrac{2}{3}$, and $\rho(x,z)=\tfrac{1}{3}$, which violates even weak stochastic transitivity, despite each individual satisfying Simple Scalability.

To account for heterogeneity, I develop a framework of \textit{collective rationalizability} that characterizes when  stochastic choices can be represented as arising from a heterogeneous population of rational individuals. I begin with the definition of collective rationalizability, then provide the characterization in Section \ref{sec: coll-rational-charact}, and present an illustrated example with three options in Section \ref{sec: rational-exp}.

To formalize the idea, I distinguish between aggregate and individual choices. Let $\rho$ now denote the aggregate stochastic choice rule, representing the population-level probability that $x$ is chosen over $y$. The question of collective rationalizability asks: when can we find some $N$ individuals, each characterized by a stochastic choice rule $\rho_i$ satisfying a given model (Simple Scalability, Moderate Utility, or Weak Utility), such that the aggregate choice probabilities $\rho$ can be expressed as a weighted average of these individual choice rules?

\begin{definition}[\textbf{Collective Rationalizability}]\label{def_collective}
\

A \textit{Collective Simple Scalability} representation has the form
\begin{equation}
    \rho(x,y)=\sum_{i=1}^N \lambda_i  \rho_i(x,y),\quad \forall x,y\in Z,
\end{equation}
for some positive integer $N$, weights $\lambda_i> 0$ that sum to 1, and each $\rho_i$ admits Simple Scalability.

\textit{Collective Moderate Utility} and \textit{Collective Weak Utility} representations are defined analogously.
\end{definition}

Crucially, this definition allows us to determine whether observed violations of Simple Scalability in aggregate data can be explained by heterogeneity alone, comparison difficulty alone, or whether both are required. If aggregate choices cannot be \textit{collectively} rationalized by Simple Scalability, it implies that no population composition \textit{exists} that can generate the observed choice probabilities such that every individual finds all pairwise comparisons equally easy.

\begin{remark}
    In the definition above, choices are collectively rationalizable if there exists some positive integer $N$ representing the number of ``types'' of individuals, where each type admits the corresponding random choice model. I remain agnostic about the number of types. A natural question is whether the $N$ required for rationalization could exceed the actual sample size. As the definition shows, collective rationalizability is about convex combination of individual choice rules. Let $m=\binom{n}{2}$ denote the number of binary pairs for $n$ options. By Carathéodory's Theorem, as long as the sample size is at least $m+1$, the number of types needed will not exceed the sample size. For instance, with $n=6$ options, this requires a sample size of at least 16. This is typically satisfiable in practice.
\end{remark}

\begin{example}\label{eg_collective}
    Consider observed population choice over three alternatives $\{x_1,x_2,x_3\}$ is 
    $$(\rho(x_1,x_2),\rho(x_2,x_3),\rho(x_1,x_3)) = (0.65,0.65,0.1).$$
    If treated as a representative agent, it cannot even be rationalized by Weak Utility since it violates weak stochastic transitivity: $\rho(x_1,x_2)>0.5$, $\rho(x_2,x_3)>0.5$ but $\rho(x_1,x_3)<0.5$. However, it can be rationalized by Collective Weak Utility. Consider two subjects with individual stochastic choice rules as
    \begin{equation*}
        \begin{aligned}
            \text{Subject 1 -- }(\rho(x_1,x_2),\rho(x_2,x_3),\rho(x_1,x_3)) = (0.9,0.4, 0.1),\\
            \text{Subject 2 -- }(\rho(x_1,x_2),\rho(x_2,x_3),\rho(x_1,x_3)) = (0.4,0.9,0.1).
        \end{aligned}
    \end{equation*}
    It is easy to verify that aggregating subjects 1 and 2's choice rates with equal weights $\lambda_1=\lambda_2=0.5$ gives us the aggregated choice probabilities above. Both of them satisfy weak stochastic transitivity but with different tastes: subject 1 has preference ordering $x_3\succ  x_1\succ x_2$, while subject 2 has $x_2\succ x_3\succ x_1$.
\end{example}

Can the above example be rationalized by Collective Simple Scalability? Try constructing individual choice rules that satisfy Simple Scalability and aggregate to $(0.65,0.65,0.1)$, and you will find it impossible. To determine this systematically requires a characterization of collective rationalizability. More generally, Definition \ref{def_collective} raises a practical challenge: how can we test whether given aggregate data are collectively rationalizable without exhaustively searching over all possible combinations of individual types? The following section addresses this question.

\subsubsection{Characterization of Collective Rationalizability}\label{sec: coll-rational-charact}
Given a finite set of $n$ alternatives $Z=\{x_i\}_{i=1}^n$ and $m=\binom{n}{2}$ possible pairs, a stochastic choice rule can be identified with a vector $\bm\rho \in [0,1]^m$, where entries of this \textit{choice vector} are $\rho(x_i, x_j)$ with some arbitrary but henceforth fixed ordering.\footnote{Throughout this paper, I assume $\rho(x_i, x_j)$ for $1\le i < j \le n$, arranged in the lexicographical order. For example, with $n=4$ options, a stochastic choice rule becomes the vector $$\bm\rho = [\rho(x_1,x_2), \rho(x_1,x_3), \rho(x_1,x_4), \rho(x_2,x_3), \rho(x_2,x_4), \rho(x_3,x_4)]^T \in [0,1]^6.$$}
The characterization of collective rationalizability is thus finding a subset of the whole space $[0,1]^m$ that consists of all rationalizable choice vectors.

The core insight underlying the approach is that collective rationalizability concerns convex combinations: a choice rule is collectively rationalizable if and only if it can be expressed as a convex combination of individual choice rules satisfying the corresponding model. This suggests a two-step characterization strategy. First, I characterize the sets of choice vectors that are rationalizable by Simple Scalability, Moderate Utility, or Weak Utility, respectively. Second, I take the convex hull of each set to obtain the collective rationalizability characterization. This transforms the problem into checking membership in a well-defined convex set. Importantly, this approach generalizes beyond the three models studied here: as long as a characterization without heterogeneity is available, the techniques in Step 2 apply broadly to derive the collectively rationalizable version of the corresponding model.

\textbf{\textit{Step 1: Characterize rationalizable set without heterogeneity.}} Stochastic transitivity conditions characterize rationalizability without heterogeneity, but they are defined implicitly through constraints on all triplets. To apply the convex hull approach for collective rationalizability, I derive an explicit analytical characterization of the individual rationalizable set as a subset of $[0,1]^m$. More precisely, I characterize this set as a union of \textit{convex polytopes}, where each polytope corresponds to a specific preference ordering (or, for Simple Scalability and Moderate Utility, specific orderings of choice probabilities under a preference ordering), providing both \textit{half-space-} and \textit{vertex-}representations for each polytope.

Note that weak stochastic transitivity is the minimal constraint shared by all three models. To understand the geometry of the choice space, consider how we can partition the entire space $[0,1]^m$ based on the  preference relation $\succcurlyeq$ defined by $x\succcurlyeq y$ iff $\rho(x,y)\ge 1/2$. Weak stochastic transitivity is precisely the requirement that this binary relation is transitive. The entire choice space $[0,1]^m$ can be partitioned into $2^m$ small hypercubes based on whether each $\rho(\cdot,\cdot) \geq \frac{1}{2}$. Since there are $n!$ different linear orders, there are $n!$ out of $2^m$ hypercubes that can be rationalized by Weak Utility. Each of these hypercubes is a convex polytope with explicit half-space constraints defined with the help of the following definition.

\begin{definition}[\textbf{Deterministic Choice Rule}]\label{def:det_choice}
\

    A choice vector $\bm\rho^* \in [0,1]^m$ is a rational deterministic choice rule if all entries are in $\{0,1\}$ and correspond to a linear order. Specifically, there exists a permutation $\sigma$ over $\{1,2,\dots,n\}$ such that 
\begin{equation*}
    \rho^*(x_i,x_j) = \left\{ \begin{matrix}
        1, \quad\text{ if } \sigma(i) < \sigma(j) \\
        0, \quad\text{ if } \sigma(i) > \sigma(j) 
    \end{matrix}
    \right., \qquad \forall 1\le i<j\le n.
\end{equation*}
The permutation $\sigma$ induces a linear order where $x_i\succ x_j$ if and only if $\sigma(i)<\sigma(j)$.
\end{definition}

Define the set of all rational deterministic choice rules $\bm\rho^*$ as $R_{\mathrm{det}}$. Then $\bm\rho$ satisfies weak stochastic transitivity if and only if
\begin{equation}
    \bm\rho \in \mathrm{Cube}(\bm\rho^*) := \left\{\bm\rho: \left|\rho(x_i,x_j) - \rho^*(x_i,x_j)\right| \le \frac{1}{2},\ \forall x_i,x_j\in Z\right\} \quad \text{ for some } \bm\rho^* \in R_{\mathrm{det}}.
    \label{eqn:indiv-small-cube}
\end{equation}
The inequality constraints $\mathrm{Cube}(\bm\rho^*) $ ensures that $\bm\rho$ inherits the same transitive ordering as $\bm\rho^*$.\footnote{The union of the sets $\mathrm{Cube}(\bm\rho^*)$ characterizes the \textit{closure} of the rationalizable set of WU. The closure is necessitated by complications at the indifference boundary $\rho(x,y) = 1/2$. For instance, $\bm\rho=[0.49, 0.49, 0.7]^T$ satisfies WU while $\bm\rho=[0.5, 0.5, 0.7]^T$ does not. Since the convex hull of finite unions of closed convex sets equals the closure of the convex hull, boundary effects remain measure-zero for collective rationalizability. Hence, throughout the paper, I work with closures of all rationalizable sets (WU, MU, SS) without loss for empirical applications.
} Hence, Weak Utility rationalizability can be represented as the following subset:
\begin{equation}
R^{\mathrm{WU}} = \bigcup_{k=1}^{n!} \mathrm{Cube}(\bm\rho^*_k) := \bigcup_{k=1}^{n!} \left\{\bm\rho : \bm A^{\mathrm{WU}}_k \bm\rho \le \bm b^{\mathrm{WU}}_k\right\}.
\end{equation}

For Simple Scalability, strong stochastic transitivity imposes additional constraints within each of the $n!$ hypercubes satisfying weak transitivity. Specifically, if $x\succcurlyeq y$ and $y\succcurlyeq z$, then $\rho(x,z)\ge \rho(x,y)$ \textit{and} $\rho(x,z)\ge \rho(y,z)$.
These constraints, inequalities of the form $\rho(x_i, x_j) \ge \rho(x_{i'}, x_{j'})$, slice away portions of each hypercube and leave a smaller convex polytope, and Simple Scalability representation is the union of these $n!$ convex polytopes.

Moderate Utility presents the most intricate geometry. Within each of the $n!$ hypercubes satisfying weak transitivity, moderate stochastic transitivity requires: if $x\succcurlyeq y$ and $y\succcurlyeq z$, then $\rho(x,z)\ge \rho(x,y)$ \textit{or} $\rho(x,z)\ge \rho(y,z)$.
The disjunctive nature of the constraints creates sets that are generally non-convex. To handle this, I divide each hypercube $\mathrm{Cube}(\bm\rho^*)$ into simplexes based on how the choice probabilities rank against each other. Each hypercube splits into $m!$ simplexes, and moderate stochastic transitivity determines which survive. Let $L$ be the number of surviving simplexes per hypercube $\mathrm{Cube}(\bm\rho^*)$, which is equal for different hypercubes due to symmetry, and the Moderate Utility rationalizable set becomes a union of $n!\cdot L$ simplexes.

Therefore, the Simple Scalability and Moderate Utility representations can be written as
\begin{align}
    R^{\mathrm{SS}}= \bigcup_{k=1}^{n!} \left\{\bm\rho : \bm A^{\mathrm{SS}}_k \bm\rho \le \bm b^{\mathrm{SS}}_k\right\},\quad R^{\mathrm{MU}}= \bigcup_{k=1}^{n!}\bigcup_{l=1}^{L} \left\{\bm\rho : \bm A^{\mathrm{MU}}_{k,\ell} \bm\rho \le \bm b^{\mathrm{MU}}_{k,\ell}\right\},
\end{align}
where $\bm A$ and $\bm b$ represent inequalities encoding both the hypercube boundaries and the strong/moderate stochastic transitivity inequalities, with exact forms in Appendix \ref{appdx: prf-a-b}. Rationalizable set for Weak Utility, Simple Scalability, and Moderate Utility can be characterized as a union of \textit{finitely} many convex polytopes, where the visualization under $n=3$ is shown in Figure \ref{fig:wu}, \ref{fig:ss}, and \ref{fig:mu}.

In addition to the half-space-representation, these polytopes' vertices are characterized as follows.
\begin{lemma}[\textbf{Vertices of Individual Rationalizability}]
    \label{lem:extreme-point}
    \
    
    Every vertex of polytopes comprising $R^{\mathrm{SS}}$, $R^{\mathrm{MU}}$, or $R^{\mathrm{WU}}$ has entries taking values only in $\{0, \frac{1}{2}, 1\}$.
\end{lemma}
Therefore, vertices can be obtained by checking strong, moderate, or weak stochastic transitivity of $\bm\rho\in\{0,\tfrac{1}{2},1\}^m$. Let $P_\mathrm{SS}$, $P_\mathrm{MU}$, and $P_\mathrm{WU}$ denote the set of extreme points respectively. For an individual with a \textit{given} preference ordering, the extreme points $\{0, \frac{1}{2}, 1\}$ emerge naturally from how we think about choice difficulty. When comparisons are easy, choices become deterministic (0 or 1 depending on the preference). When comparisons are maximally difficult, choices become essentially random (approaching $\tfrac{1}{2}$). Any observed choice probability can be expressed as a combination of these extremes, and these values represent the fundamental boundaries of rationalizable choice under transitivity constraints with some given preference ordering.

\textbf{\textit{Step 2: Characterize collective rationalizability with heterogeneity.}}
Collective rationalizability is then obtained by taking the convex hull of the rationalizable set without heterogeneity.  

\begin{theorem}[\textbf{Characterization of Collective Rationalizability}]
\label{thm:collective}
\

    For a choice vector $\bm\rho \in [0,1]^m$, the following are equivalent:
    \begin{enumerate}[(i).]
        \item $\bm\rho$ is \textit{collectively rationalizable} by Simple Scalability.
        \item $\bm\rho \in \conv(P_\mathrm{SS})$.
        \item The following linear system over $(\bm\rho_1, \dots, \bm\rho_{n!}, \bm x)$ is feasible: 
        \vspace{-.3cm}
            \begin{align}
            &\bm\rho = \sum_{k=1}^{n!} \bm\rho_k, \quad \bm x \in \Delta^{n!-1}, \quad
                \bm 0 \le \bm \rho_k \le x_k \bm 1, \quad \bm A_k^\mathrm{SS}\bm\rho_k\le x_k\bm b_k^\mathrm{SS}, \quad \quad \forall k=1,\dots, n! 
            \end{align}
    \end{enumerate}
    Collective Moderate Utility and Collective Weak Utility are characterized analogously.
\end{theorem}
For the vertex-representation of collective rationalizability (i) $\iff$ (ii), I use the fact that:
\begin{equation}
\conv\left(\bigcup_{i=1}^k \mathrm{conv}(P_i)\right)=\;
\conv\left(\bigcup_{i=1}^k P_i\right),\quad \forall \text{ finite collection of sets }P_i\subseteq \mathbb{R}^m.
\end{equation}
This means the collective rationalizability set equals the convex hull of all extreme points obtained from the each of the  polytopes comprising $R^{\mathrm{SS}}$, $R^{\mathrm{MU}}$, or $R^{\mathrm{WU}}$.

For the half-space-representation of collective rationalizability (i) $\iff$ (iii), the convex hull of a finite union of convex polytopes can be readily obtained using techniques from modeling disjunctions \citep{balas1971intersection}, by first using indicator variables $\bm x$ to represent the finite union, then  relaxing the integrality constraints to obtain the convex hull. 

\autoref{thm:collective} provides both vertex and half-space representations of the collectively rationalizable sets. This dual characterization offers geometric insights into the structure of collectively rationalizable choice and the nesting relationships among Simple Scalability, Moderate Utility, and Weak Utility, with or without heterogeneity, as well as with other stochastic choice models. It also establishes the theoretical foundation for statistical testing, which will be developed in Section \ref{sec: test}.

\subsubsection{Example: Collective Rationalizability for Three Options}\label{sec: rational-exp}
To illustrate the theoretical results, consider the case of three choice options $Z=\{x_1,x_2,x_3\}$. The choice vector is 
$
    \bm\rho=[\rho(x_1,x_2),\ \rho(x_1,x_3),\ \rho(x_2,x_3)]^T\in[0,1]^3.
$

\textbf{Inequalities of individual rationalizability.} With three options, there are $3!=6$ possible linear orders. For each preference ordering, individual rationalizability constraints create distinct regions. For instance, under $x_1\succ x_2\succ x_3$, the corresponding deterministic choice vector is 
$$\bm\rho_k^*=[\rho^*(x_1,x_2),\rho^*(x_1,x_3),\rho^*(x_2,x_3)]=[1,1,1]^T.$$
\begin{enumerate}[(i).]
   \item Weak Utility: $\bm A_k^\mathrm{WU}\bm\rho\le \bm b_k^{\mathrm{WU}}$ means all choice rates $\rho(x_1,x_2),\rho(x_2,x_3),\rho(x_1,x_3)\in[\tfrac{1}{2},1]$, which are also included in the constraints for Simple Scalability and Moderate Utility.
    \item Simple Scalability: $\bm A_k^\mathrm{SS}\bm\rho\le \bm b_k^{\mathrm{SS}}$ \textit{additionally} requires $\rho(x_1,x_3)\ge \max\{\rho(x_1,x_2),\rho(x_2,x_3)\}$.
    \item Moderate Utility: The hypercube is further partitioned into 6 simplexes by the ranking of choice rates, where $L=4$ of them satisfy moderate stochastic transitivity.
    \begin{itemize}
        \item $\bm A_{k,1}^\mathrm{MU}\bm\rho\le \bm b_{k,1}^{MU}$ \textit{additionally}  requires $\rho(x_1,x_3)\ge \rho(x_1,x_2)\ge \rho(x_2,x_3)$;
        \item $\bm A_{k,2}^\mathrm{MU}\bm\rho\le \bm b_{k,2}^{\mathrm{MU}}$ \textit{additionally}  requires $\rho(x_1,x_3)\ge \rho(x_2,x_3)\ge \rho(x_1,x_2)$;
        \item $\bm A_{k,3}^\mathrm{MU}\bm\rho\le \bm b_{k,3}^{\mathrm{MU}}$ \textit{additionally}  requires $\rho(x_1,x_2)\ge \rho(x_1,x_3)\ge \rho(x_2,x_3)$;
        \item  $\bm A_{k,4}^\mathrm{MU}\bm\rho\le \bm b_{k,4}^{\mathrm{MU}}$ \textit{additionally}  requires $\rho(x_2,x_3)\ge \rho(x_1,x_3)\ge \rho(x_1,x_2)$.
    \end{itemize}
\end{enumerate}
Each individual rationalizable set is the union of 6 such regions, with each region representing one preference ordering illustrated by different colors in Figures \ref{fig:wu}, \ref{fig:mu}, and \ref{fig:ss}.

\textbf{Vertices of individual rationalizability.} Within each region, the extreme points have entries in $\{0,\tfrac{1}{2},1\}$. For example, for $x_1\succ x_2\succ x_3$, the extreme points are
\vspace{-0.1cm}
\begin{align*}
    P_{\mathrm{WU}} &= \{(1,1,1),(\tfrac{1}{2},1,1),(1,1,\tfrac{1}{2}),(\tfrac{1}{2},1,\tfrac{1}{2}),(\tfrac{1}{2},\tfrac{1}{2},\tfrac{1}{2}),(\tfrac{1}{2},\tfrac{1}{2},1), (1,\tfrac{1}{2},\tfrac{1}{2}), (1,\tfrac{1}{2},1)\},\\
    P_{\mathrm{MU}} &= \{(1,1,1),(\tfrac{1}{2},1,1),(1,1,\tfrac{1}{2}),(\tfrac{1}{2},1,\tfrac{1}{2}),(\tfrac{1}{2},\tfrac{1}{2},\tfrac{1}{2}),(\tfrac{1}{2},\tfrac{1}{2},1), (1,\tfrac{1}{2},\tfrac{1}{2})\},\\
    P_{\mathrm{SS}} &= \{(1,1,1),(\tfrac{1}{2},1,1),(1,1,\tfrac{1}{2}),(\tfrac{1}{2},1,\tfrac{1}{2}),(\tfrac{1}{2},\tfrac{1}{2},\tfrac{1}{2})\}.
\end{align*}
\textbf{Collective rationalizability.} The next step is to take the convex hull of the individual rationalizable sets, and Figures \ref{fig:collective-wu}, \ref{fig:collective-mu}, and \ref{fig:collective-ss} visualize the result.
For $n=3$, collective Simple Scalability and Moderate Utility coincide, and the facet-defining inequalities can be written as
\begin{align}
    \text{Collective SS and MU}&:\quad 1 \le \rho(x_1, x_2) + \rho(x_2, x_3) + \rho(x_3, x_1) \le 2,
    \label{eqn:facet-defining-triplet}\\
    \text{Collective WU}&:\quad \tfrac{1}{2} \le \rho(x_1, x_2) + \rho(x_2, x_3) + \rho(x_3, x_1) \le \tfrac{5}{2},
    \label{eqn:facet-defining-triplet-wu}
\end{align}
in addition to the trivial boundary inequalities $0 \le \rho(x_i, x_j) \le 1$, $\forall (i,j)$. 

\begin{figure}[H]
    \centering
    \begin{subfigure}{0.45\textwidth}
        \centering
        \includegraphics[width=\linewidth,trim={5cm 9cm 1cm 4cm},clip]{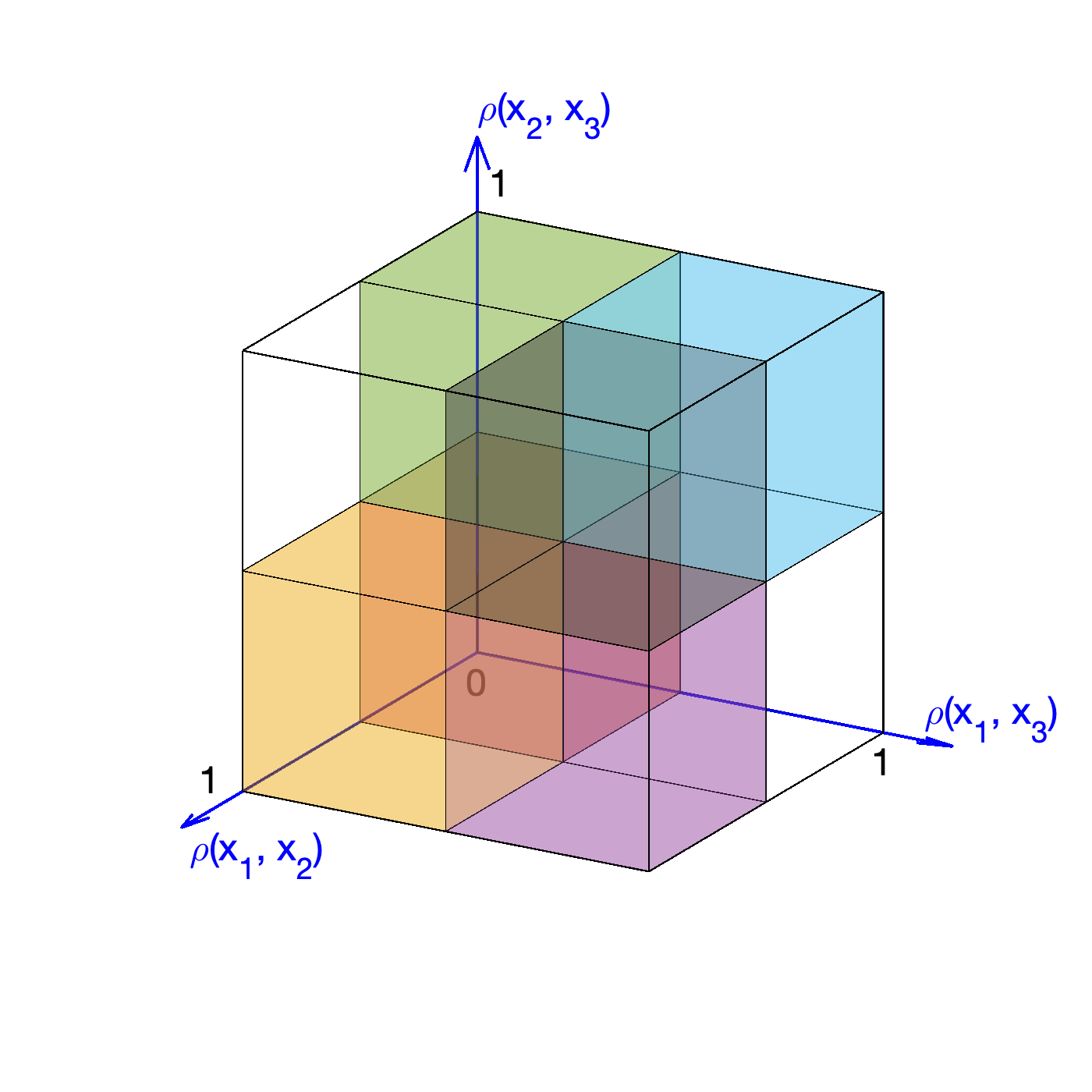}
        \caption{Rationalizable set for Individual WU}
        \label{fig:wu}
    \end{subfigure}
    \begin{subfigure}{0.45\textwidth}
        \centering
        \includegraphics[width=\linewidth,trim={5cm 9cm 1cm 4cm},clip]{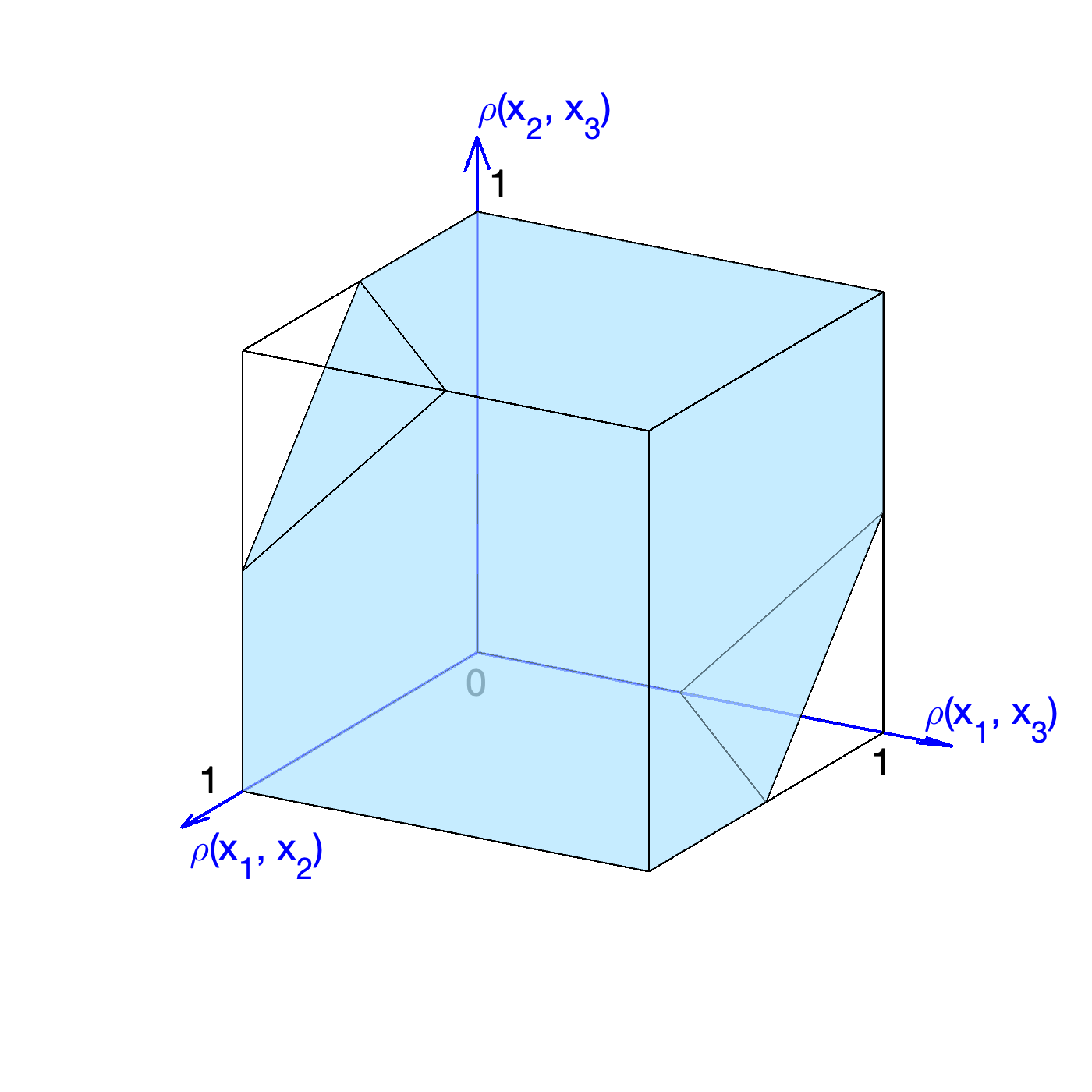}
        \caption{Rationalizable set for Collective WU}
        \label{fig:collective-wu}
    \end{subfigure}

    \begin{subfigure}{0.43\textwidth}
        \centering
        \includegraphics[width=\linewidth,trim={5cm 9cm 1cm 4cm},clip]{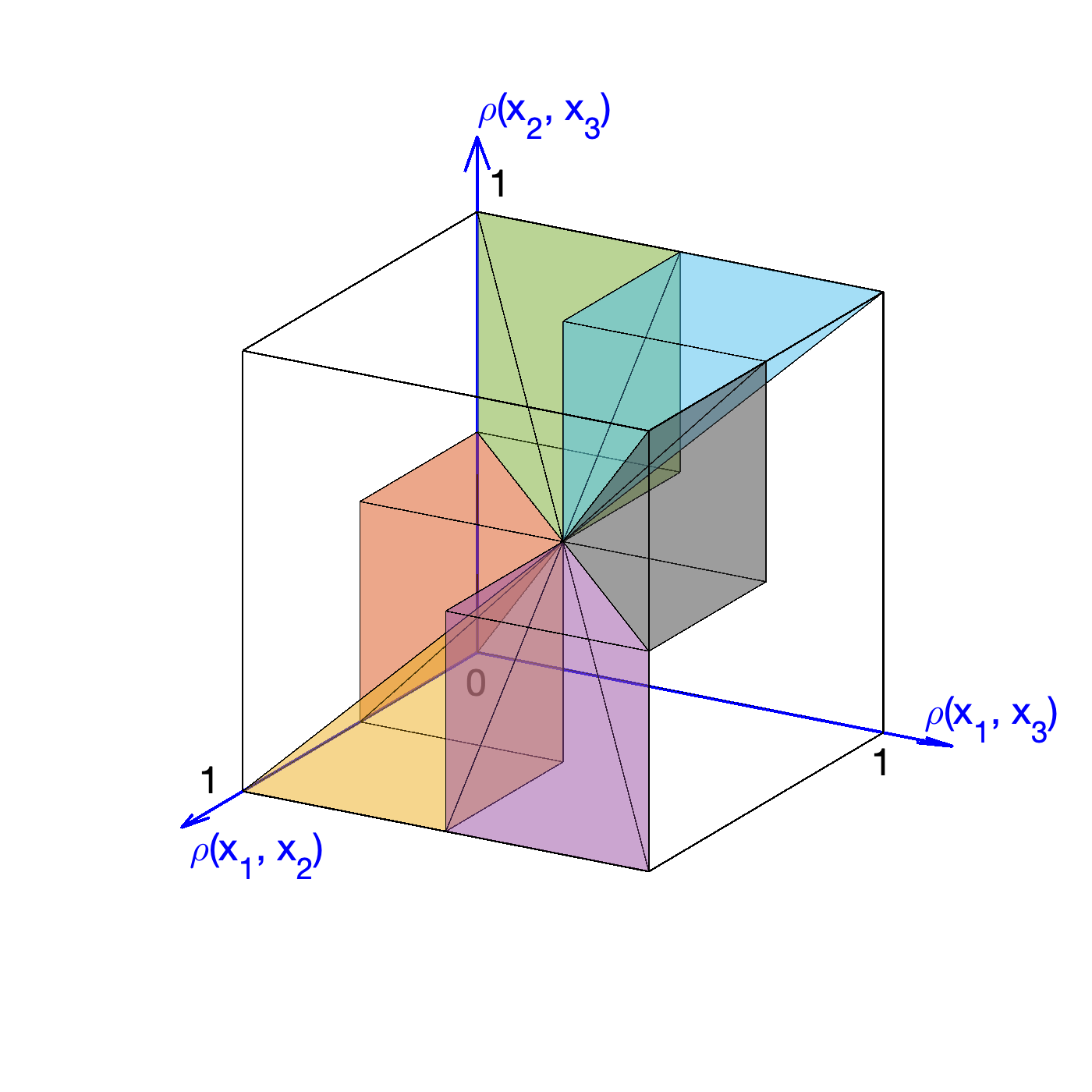}
        \caption{Rationalizable set for Individual SS}
        \label{fig:ss}
    \end{subfigure}
    \begin{subfigure}{0.43\textwidth}
        \centering
        \includegraphics[width=\linewidth,trim={5cm 9cm 1cm 4cm},clip]{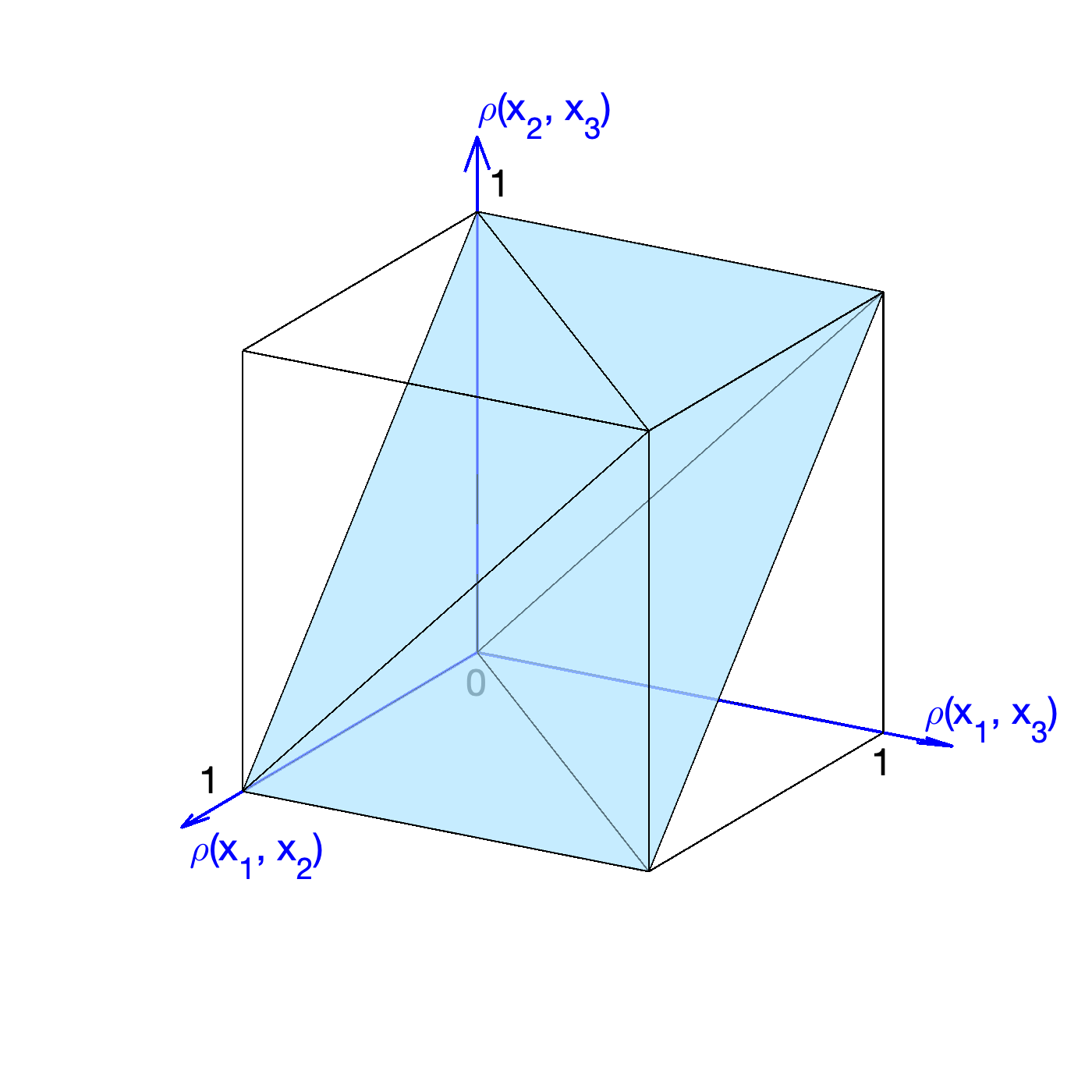}
        \caption{Rationalizable set for Collective SS}
        \label{fig:collective-ss}
    \end{subfigure}

    \begin{subfigure}{0.43\textwidth}
        \centering
        \includegraphics[width=\linewidth,trim={5cm 9cm 1cm 4cm},clip]{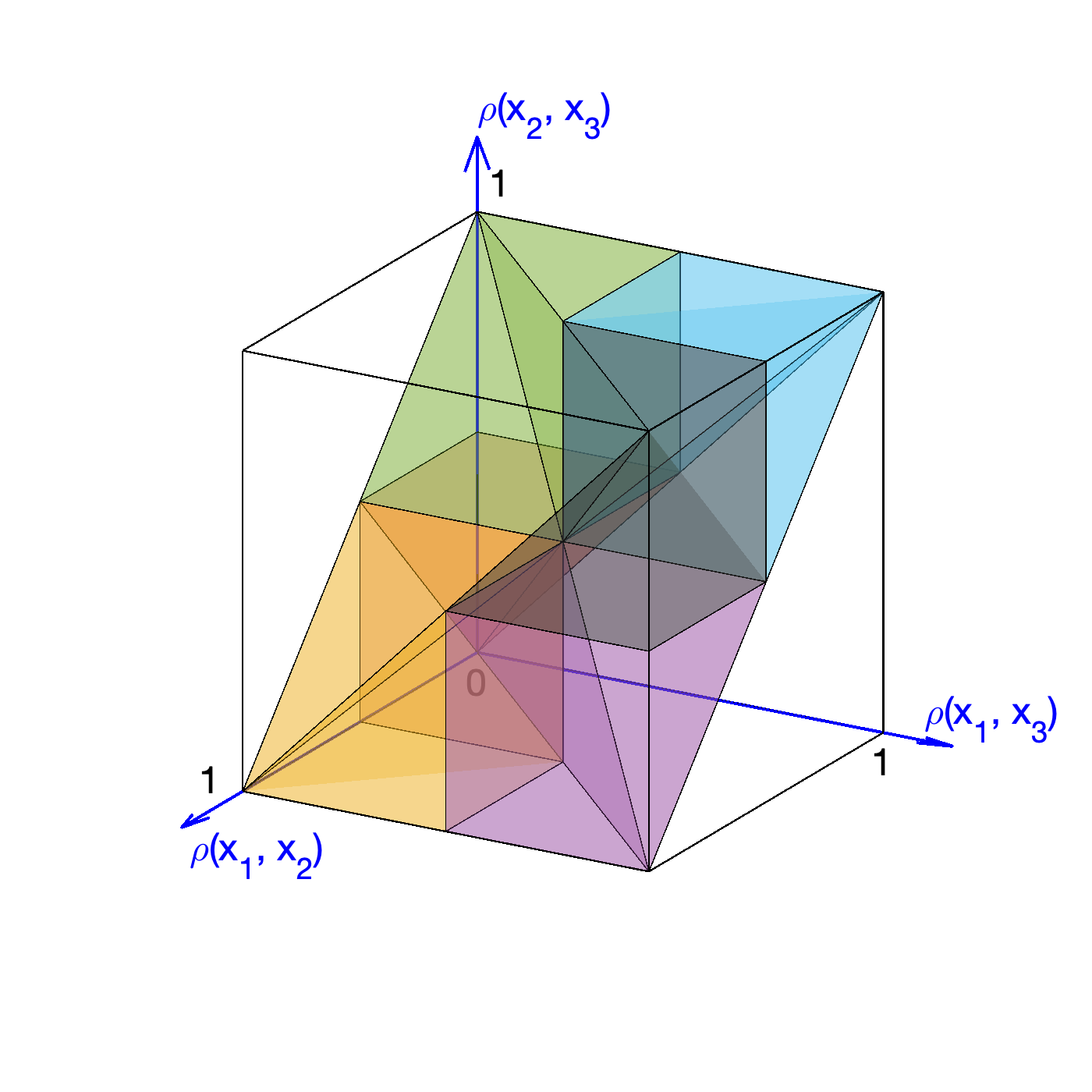}
        \caption{Rationalizable set for Individual MU}
        \label{fig:mu}
    \end{subfigure}
    \begin{subfigure}{0.43\textwidth}
        \centering
        \includegraphics[width=\linewidth,trim={5cm 9cm 1cm 4cm},clip]{figs/3d_viz_figs/collective_ssm.png}
        \caption{Rationalizable set for Collective MU}
        \label{fig:collective-mu}
    \end{subfigure}

    \caption{Rationalizable sets for Individual and Collective WU/MU/SS models for $n=3$. In panels (a), (c), (e), different colors represent the six regions corresponding to different preferences. MU and SS impose increasingly restrictive transitivity constraints within each region. Panels (b), (d), (f) show the corresponding collective rationalizability sets, which are the convex hulls of the individual sets. Taking the convex hull connects the previously regions into a single polytope, expanding the set of rationalizable choice patterns while allowing for heterogeneity.}
    \label{fig:3d-set-viz}
\end{figure}

Collective Rationalizability constraints capture that rational heterogeneous populations create a balance: while individuals are transitive, aggregation may produce some ``cycling'' or inconsistency, but not too much. The feasible intervals $[1,2]$ or $[\tfrac{1}{2},\tfrac{5}{2}]$ represent the natural amount of intransitivity that emerges from mixing different types, distinguishing patterns explainable by heterogeneity from those requiring additional behavioral assumptions.

Comparing Figures \ref{fig:wu} and \ref{fig:collective-ss} with Figure \ref{fig:ss} reveals two key insights. First, there are significant overlaps where violations of Simple Scalability can be explained by either heterogeneity (Collective Simple Scalability) or comparison difficulty (Weak Utility) alone. However, Collective Simple Scalability and Weak Utility are not nested. Some choice patterns can only be explained by heterogeneity, others can only be explained comparison difficulty, and still others require both to be rationalized. Interestingly under $n=3$, the volume of the Weak Utility rationalizable set is 0.75, while that of Collective Simple Scalability is 0.67, indicating that the heterogeneity model is actually less permissive than the comparison difficulty model.

Comparing Figures \ref{fig:collective-ss} and \ref{fig:collective-mu} shows that Collective Simple Scalability coincides with Collective Moderate Utility when $n=3$, as moderate stochastic transitivity provides no additional degrees of freedom beyond strong transitivity for three alternatives when heterogeneity is allowed.

Readers familiar with the literature will recognize that inequalities \eqref{eqn:facet-defining-triplet}  are known as the ``triangle inequalities'' in stochastic choice, which are essentially Block-Marschak polynomials for binary choices under $n=3$ and characterize the random utility models
\citep{block1959random}. This observation reveals an intriguing connection between our collective representation and random utility models. I explore the nesting relationships between different collective representations and also random utility representation in detail in the following section.

\subsection{Nesting Relationships and Connections to Random Utility}
\label{sec: to-rum}
I now compare the testable implications of random choice models, Simple Scalability, Moderate Utility, and Weak Utility, with or without heterogeneity, and relate them to those of random utility models. To formalize this comparison, let $Z$ denote the finite set of alternatives and let $\Pi$ denote the set of all strict preferences on $Z$, where each $\pi\in \Pi$ is a bijection $\pi:Z\to \{1,\dots,|Z|\}$, which is a specific utility representation of the preference.

\begin{definition}[\textbf{Random Utility}]\label{def:ru}
\
    A \textit{Random Utility (RU) representation} has the form
    \begin{equation}
        \rho(x,y) = \sum_{\pi\in\Pi}\nu(\pi)\mathbbm{1}_{\{\pi(x)\ge \pi(y)\}},\quad \forall x,y\in Z,
    \end{equation}
    where $\nu:\Pi\to\R^+$ is some probability distribution on $\Pi$ with $\sum_{\pi\in\Pi}\nu(\pi)=1.$
\end{definition}

It is well known that RU and random choice representations are not nested with each other (see discussion in, for example, \citet{rieskamp2006extending, strzalecki2025stochastic}). In fact, weak stochastic transitivity can be violated by RU,\footnote{Again, recall the example of Condorcet paradox.} while on the other side, one can construct choice vector that is rationalized by Simple Scalability but not RU, since Simple Scalability does not require the choice rates to be represented by some weighted average of strict orders.

However, once heterogeneous taste is taken into account, we can see that RU is nested within Collective Simple Scalability, the most restrictive of all three collective models.
From Definition \ref{def:ru}, a choice rule $\bm\rho$ is representable by RU if it corresponds to a probability measure $\nu$ over the strict preference orderings on $Z$, i.e., $\bm\rho$ is representable by RU iff $\bm\rho\in\conv\left(R_\mathrm{det}\right)$, where $R_\mathrm{det}$ is the set of all rational deterministic choices from Definition \ref{def:det_choice}. Since all rational deterministic choices satisfy Collective Simple Scalability, it follows that RU is nested within Collective Simple Scalability.

A more thorough analysis reveals the following relationship between the collective models and RU:

\begin{lemma}[\textbf{Nesting Relations}]
\label{lem:model-equivalence-small-n}
\
\begin{enumerate}[(i).]
    \item RU is nested within Collective SS, with equivalence for $n\le 6$ and strict nesting for $n\ge 7$.
    
    \item Collective SS is nested within Collective MU, with equivalence for $n\le 5$ and strict nesting for $n\ge 6$.
    \item Collective MU is strictly nested within Collective WU for all $n \ge 3$.
\end{enumerate}
\end{lemma}

Lemma \ref{lem:model-equivalence-small-n} follows from comparing the extreme points computed in Lemma \ref{lem:extreme-point}. The nesting relationships can be verified by checking whether each model's extreme points can be expressed as convex combinations of the alternative model's extreme points.

\textbf{Why equivalence for small $n$?} When the number of alternatives is small, the geometry of choice probabilities is relatively simple, leading to identical constraints across models. As shown in the example for $n=3$, Collective Simple Scalability/Moderate Utility and RU share the same facet-defining constraints: the triangle inequality \eqref{eqn:facet-defining-triplet} and the probability boundaries $[0,1]$. This pattern persists for $n=4$ and $n=5$, where numerical computation reveals that all facet-defining inequalities are ``triangle-type'', i.e., versions of \eqref{eqn:facet-defining-triplet} applied to different triples of alternatives.\footnote{Computed using the Quickhull algorithm \citep{barber1996quickhull} via \texttt{scipy.spatial.ConvexHull} in Python.} While these triangle inequalities remain necessary for all $n$, the models diverge for larger $n$ as additional structural constraints emerge.

\textbf{The complexity of larger $n$.} For RU, its rationalizable set is a well-studied combinatorial object often referred to as the linear ordering polytope \citep{grotschel1985facets, katthan2011linear}. Beyond triplet inequalities \eqref{eqn:facet-defining-triplet}, characterization for RU requires increasingly complex constraints, such as $k$-fence, Mobius ladder, and others \citep{grotschel1985facets}, where the number of facet-defining inequalities of the linear ordering polytope grows exponentially with $n$. The complete characterization remains an open problem even for RU. Similarly, while I have established the characterizations of collective rationalizability in \autoref{thm:collective}, their exact facet-defining inequalities remain to be fully provided (see Lemma \ref{lem: extra-collective} in the Appendix for partial results).

\textbf{Economic interpretation.} The equivalence between RU and Collective Simple Scalability for small $n$ provides the following insights in the binary choice problem: behavior interpreted as arising from heterogeneous deterministic agents (RU) is observationally equivalent to behavior from heterogeneous agents who themselves face uncertainty of the alternative's utility scales or the noisy readings of them (Collective Simple Scalability). 
Furthermore, the equivalence between Collective Simple Scalability and Collective Moderate Utility for $n\le 5$ implies a practical limitation: with few alternatives, we lack the power to determine whether a population's stochastic behavior arises from distance-based comparison difficulty or not.

\begin{example}[RU strictly nested within Collective SS when $n=7$]
\

Consider $n=7$ and choice rates $\rho$ defined as 
    $$
        \rho(x_i, x_j) = \left\{\begin{matrix}
            1-\varepsilon & \text{ if } j-i \ge 3 \quad\ \ \ \  \\
            \frac{1}{2}+\varepsilon & \text{ if } 1\le j-i \le 2
        \end{matrix} \right.,
    $$
    with $\varepsilon$ small enough. It is straightforward to verify that this $\bm\rho$ satisfies (individual) Simple Scalability and thus also collective SS. However, $\bm\rho$ cannot be rationalized by RU.
    To generate near-equal choice probabilities for nearby pairs, the population must be nearly evenly split between opposite rankings of these pairs. But the majority of the population needs to agree on the ranking of distant pairs, making it impossible to satisfy transitivity.

\end{example}

\subsection{Identify Ranking of Comparison Difficulty}\label{sec: cplx-ranking}
Before turning to statistical testing based on collective rationalizability, I present one additional theoretical result that goes beyond testing whether comparison difficulty varies across pairs: if comparison difficulty indeed varies across pairs, can we characterize \textit{how} it varies? Readers primarily interested in collective rationalizability may proceed directly to Section \ref{sec: test}.

In the absence of additional assumptions, comparison difficulty cannot be uniquely identified from binary choice data over finite sets in random choice models.
Regarding rankings of comparison difficulties, \citet{he2024moderate} establishes a sufficient condition under which certain observed choices can uniquely determine the ordinal ranking of difficulty under Moderate Utility and Weak Utility representations: $\rho(x_1,x_2)>\rho(x_1,x_3)>\rho(x_2,x_3)\ge 1/2 \text{\ \ implies that\ \ }c(x_1,x_2)<c(x_1,x_3)$.

In this section, I develop \textit{necessary and sufficient} conditions for rationalizing individual choice data with a WU representation and a \textit{given ranking of comparison difficulty}. Consider the problem: given observed choices $\bm\rho = (\rho(x_1,x_2),\rho(x_1,x_3),\rho(x_2,x_3))$, when can these be rationalized by WU with a specific difficulty ranking such as $c(x_1,x_2)>c(x_1,x_3)>c(x_2,x_3)$? I solve this for arbitrary numbers of alternatives in a \textit{loop}. Then, I generalize this notion from individual rationalizability to collective rationalizability with shared ranking of comparison difficulty across heterogeneous agents.

I focus on binary choices over loops, namely individual choice rates $\{\rho(x_j,x_{j+1})\}_{j=1}^m$ for a set of $m$ choice options $\{x_j\}_{j=1}^m$ (indices modulo $m$). Denote $\rho_j := \rho(x_j,x_{j+1})$ and $c_j := c(x_j,x_{j+1})$ for all $j=1,\dots,m$. 
To formalize the analysis, for each pair $j\in\{1,\dots,m\}$, I introduce:
\begin{align*}
    &\text{Choice rate magnitudes: } \mu_j := \left| \rho_j - \tfrac{1}{2}\right|;\\
    &\text{Sign of the preferences: }J^+(\bm\rho) := \left\{ j \mid \rho_{j} > \tfrac{1}{2}\right\}, \quad J^-(\bm\rho) := \left\{ j \mid \rho_{j} < \tfrac{1}{2}\right\}, \quad J^0(\bm\rho) := \left\{ j \mid \rho_{j} = \tfrac{1}{2}\right\}.
\end{align*}

The loop structure is crucial: utility differences must sum to zero throughout any cycle. This constraint interacts with the given ranking of comparison difficulty, because higher comparison difficulty and choice rate magnitudes imply larger absolute utility differences. The key insight is that a given ranking of $c_j$ is compatible with observed choices if it aligns appropriately with the ranking of $\mu_j$ across the sets with different preference directions $J^+$ and $J^-$.
To formalize this alignment of $c_j$ and $\mu_j$, I define the notions of \textit{ranking function} and \textit{ranking-dominance}. 

\begin{definition}[\textbf{Ranking Functions and Ranking-Dominance}]
\

Let $\#$ denote the cardinality of a finite set.
\begin{enumerate}
    \item A \textit{ranking function} of a sequence of values $\{x_1,\dots,x_m\}$ is a function $\pi: \left\{1,\dots,m\right\} \to \left\{1,\dots,m\right\}$
\begin{equation}
    \pi(j) = \#\{k\in\{1,\dots,m\}: x_j < x_k\} + 1.
\end{equation}
\item  Let $\sigma, \pi$ be two ranking functions  over $\{1,\dots,m\}$. Given two index subsets $A, B \subseteq \{1,\dots,m\}$, $A$ \textit{ranking-dominates} $B$ with respect to $(\sigma,\pi)$ if
    \begin{equation}
        \#\{j\in A\mid\sigma(j)\le s,\pi(j)\le t\}\  \ge\  \#\{j\in B\mid\sigma(j)\le s,\pi(j)\le t\}, \quad \forall s,t\in \{1,\dots,m\},
        \label{eqn:ranking-domination-def}
    \end{equation}
    with strict inequality for some $(s,t)$.
\end{enumerate} 
\end{definition}

Intuitively, an index $j$ with ranking function value $\pi(j)$ means that $x_j$ ranks the $\pi(j)$-th largest among values in the sequence. And ranking-dominance means that, cumulatively, the elements in $A$ are consistently concentrated at higher ranks than those in $B$ under both rankings, shifting their distribution upward and leftward relative to $B$. With these two definitions, the characterization in Proposition \ref{prop:loop-rank} shows that ranking-dominance provides the key obstruction for rationalizability under a given ranking of comparison difficulty.

\begin{proposition}[\textbf{Weak Utility with a Given Ranking of Comparison Difficulty}]
\label{prop:loop-rank}
\

Consider individual choice rates over a loop of binary pairs $\bm\rho=\{\rho_j\}_{j=1}^m\in[0,1]^m$. Let $\sigma$ be the ranking function of the choice rate magnitudes $\{\mu_j\}_{j=1}^m$, and let $\pi$ be a \textit{given} ranking function of comparison difficulty $\{c_j\}_{j=1}^m$.
Then, $\bm\rho$ can be represented by WU with the given ranking $\pi$ \textit{if and only if} neither $J^+(\bm\rho)$ nor $J^-(\bm\rho)$ ranking-dominates each other with respect to $(\sigma,\pi)$, with necessity additionally assuming that no choice rates $\rho_j$ equal $\frac{1}{2}$.
\end{proposition}

The intuition is that over a loop, if $J^+$ ranking-dominates $J^-$, then the positive preferences systematically have larger both choice magnitudes and comparison difficulties than negative preferences, and the required zero-sum condition on utility differences cannot be satisfied.

Given the key role of the ranking-dominance property, here we provide further intuition behind this concept, alongside Proposition \ref{prop:loop-rank}. Consider a setup of $m=5$ choice options forming a loop, and the following two different cases of ranking of comparison difficulty. We will discuss whether WU representation with the given ranking is possible by examining the conditions for ranking-domination.

\textbf{\textit{Example a.}} Consider the following ranking of the choice magnitudes $\mu_j$, and a given ranking of comparison difficulty $c_j$, where each $j$ is the index for a binary pair:
\begin{align*}
    \mu_1>\mu_2>\mu_3>\mu_4>\mu_5,\qquad c_1>c_3>c_5>c_2>c_4.
\end{align*}
Thus the ranking functions $\sigma$ and $\pi$ are respectively
\begin{align*}
    \sigma(1) &= 1, \ \sigma(2) = 2, \ \sigma(3) = 3, \ \sigma(4) = 4, \ \sigma(5) = 5, \\
    \pi(1) &= 1, \ \pi(2) = 4, \ \pi(3) = 2, \ \pi(4) = 5, \ \pi(5) = 3.
\end{align*}
Also assume the index subsets of the signs of the preferences are
\begin{align*}
    J^+(\bm\rho) = \left\{ j \mid \rho_{j} > \tfrac{1}{2}\right\} = \left\{1,4,5\right\},\ J^-(\bm\rho) =\left\{ j \mid \rho_{j} < \tfrac{1}{2}\right\}= \{2,3\}.
\end{align*}

\textbf{\textit{Example b.}} Consider the following ranking of $\mu_j$, and a given ranking of $c_j$ 
\begin{align*}
    \mu_1>\mu_2>\mu_3>\mu_4>\mu_5,\quad c_1>c_3>c_2>c_5>c_4.
\end{align*}
Thus the ranking functions $\sigma$ and $\pi$ are respectively
\begin{align*}
    \sigma(1) &= 1, \ \sigma(2) = 2, \ \sigma(3) = 3, \ \sigma(4) = 4, \ \sigma(5) = 5, \\
    \pi(1) &= 1, \ \pi(2) = 3, \ \pi(3) = 2, \ \pi(4) = 5, \ \pi(5) = 4.
\end{align*}
Also assume the index subsets are
\begin{align*}
    J^+(\bm\rho) = \left\{ j \mid \rho_{j} > \tfrac{1}{2}\right\} = \left\{1,3,4\right\},\ J^-(\bm\rho) =\left\{ j \mid \rho_{j} < \tfrac{1}{2}\right\}= \{2,5\}.
\end{align*}

For the two examples described above, We illustrate the ranking functions of \textit{Example a} in Figure \ref{fig: rank-dom-neg}, and \textit{Example b} in Figure \ref{fig: rank-dom-pos}. Each colored cell corresponds to one index $j$ in $J^+$ or $J^-$, with its coordinates on the grid being $(\sigma(j),\pi(j))$. Notice that the term $\#\{j\in J^+\mid\sigma(j)\le s,\pi(j)\le t\}$ can be interpreted as the number of blue (or red for $J^-$) cells on the top-left $s \times t$ submatrix. The inequality condition for ranking-dominance \eqref{eqn:ranking-domination-def} thus implies that for every top-left submatrix, there are at least as many blue cells as there are red cells.

\begin{figure}[H]
\centering
\begin{minipage}{0.49\textwidth}
\centering
\includegraphics[width=0.9\linewidth]{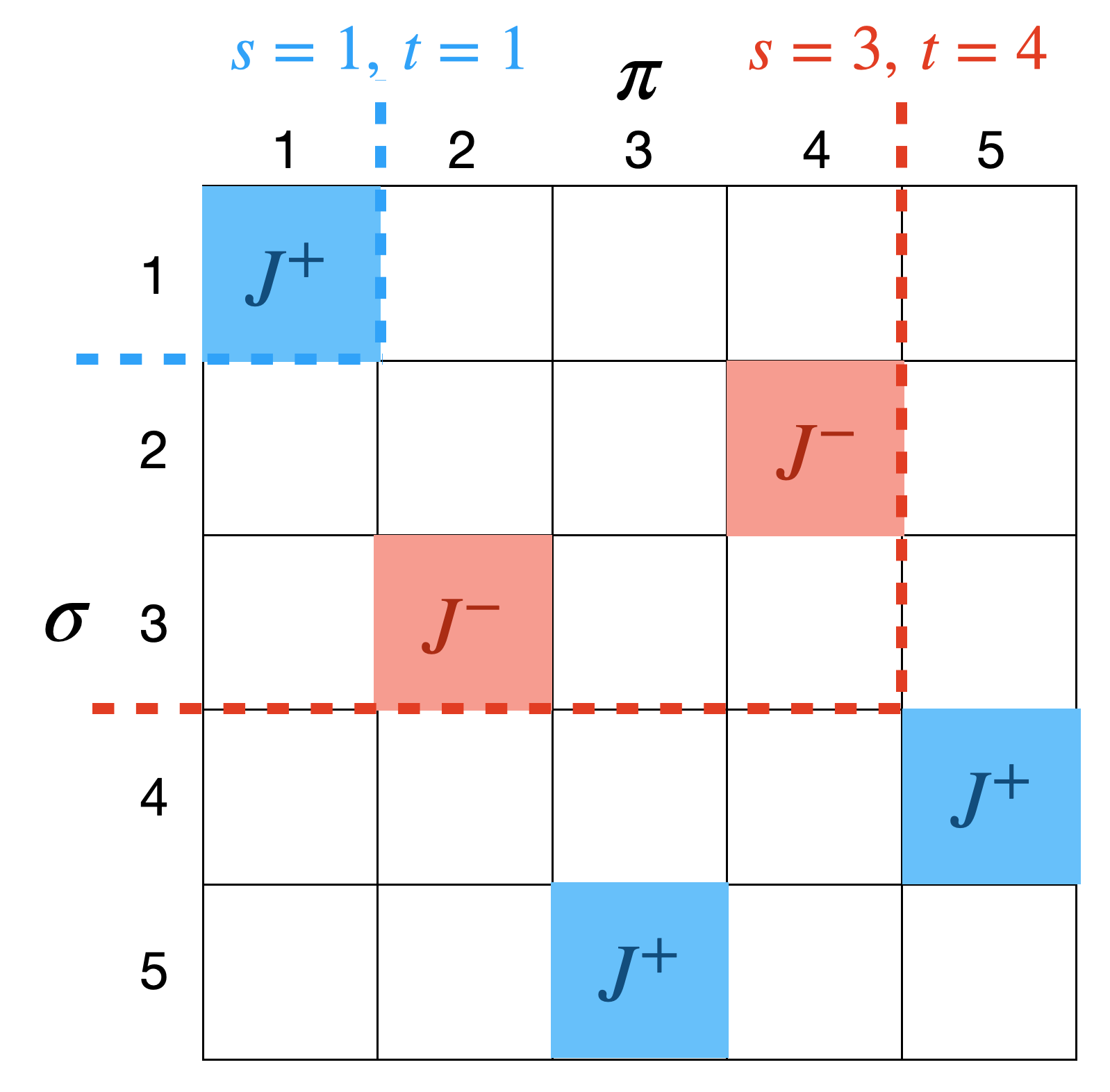}
\subcaption{Neither $J^+$ nor $J^-$ ranking-dominates the other}
\label{fig: rank-dom-neg}
\end{minipage}
\begin{minipage}{0.49\textwidth}
\centering
\includegraphics[width=0.9\linewidth]{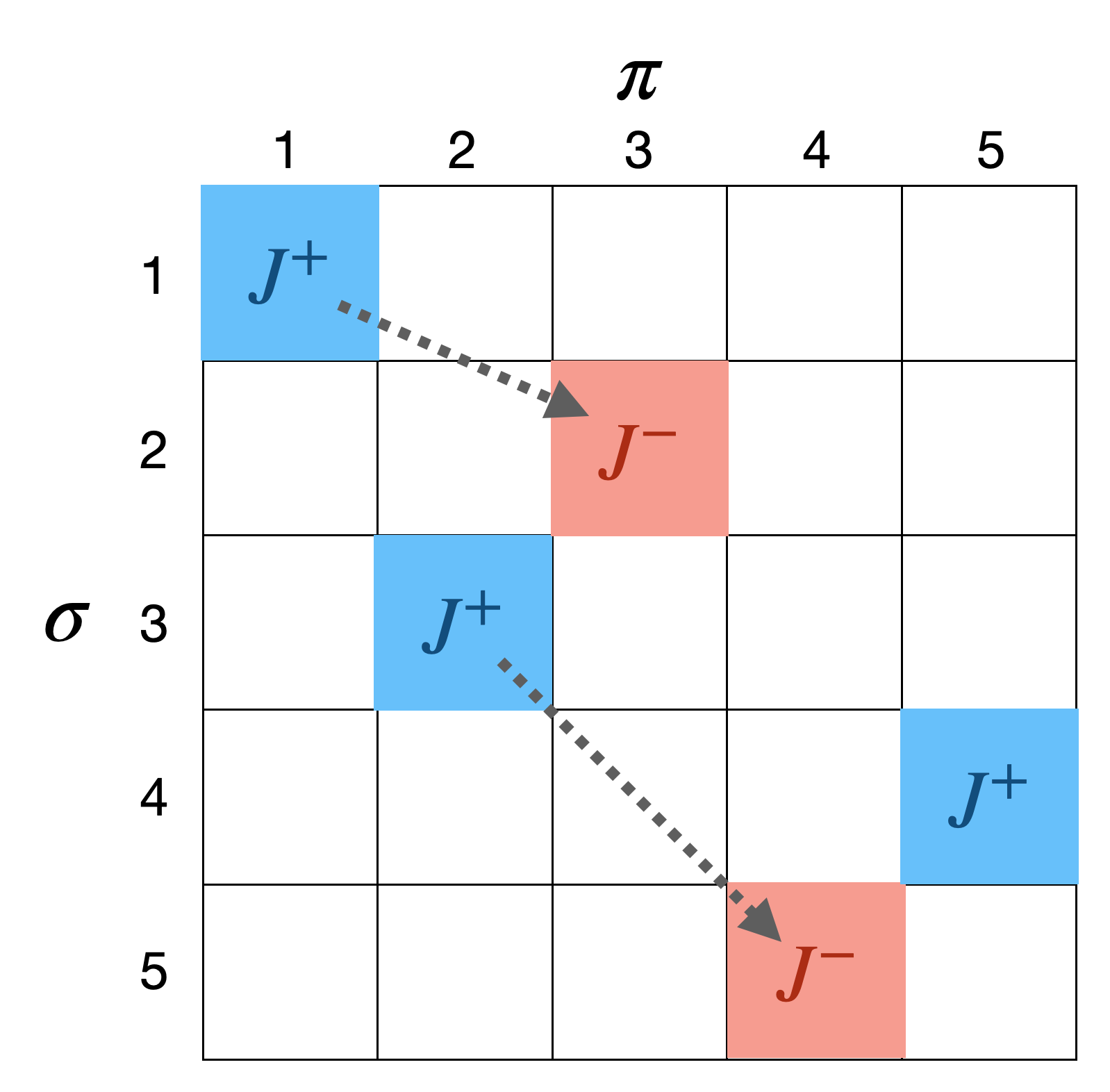}
\subcaption{$J^+$ ranking-dominates $J^-$}
\label{fig: rank-dom-pos}
\end{minipage}
\caption{Illustration of the ranking-domination condition.}
\end{figure}

Figure \ref{fig: rank-dom-neg} shows that in \textit{Example a}, neither $J^+$ nor $J^-$ ranking-dominates the other. Here the blue dashed lines correspond to a top-left $s\times t$ submatrix with $(s,t)=(1,1)$, which contains 1 blue and 0 red cells. Meanwhile, the red dashed lines correspond to a top-left $s\times t$ submatrix with $(s,t)=(3,4)$, which contains 1 blue and 2 red cells. This implies that choosing $(s,t)=(1,1)$ makes the left-hand-side of \eqref{eqn:ranking-domination-def} strictly greater, while $(s,t)=(3,4)$ makes the right-hand-side of \eqref{eqn:ranking-domination-def} strictly greater, which concludes that neither $A$ nor $B$ ranking-dominates the other. From Proposition \ref{prop:loop-rank}, this implies that the choice rates data \textit{can} be represented as WU with this specific ranking of comparison difficulty.

Figure \ref{fig: rank-dom-pos} on the other hand, shows that in \textit{Example b}, $J^+$ ranking-dominates $J^-$. One can verify by enumerating all possible pairs of $(s,t)$, to see that every top-left $s\times t$ submatrix has at least as many blue cells as red, with strict inequality holding in some cases (for example $(s,t)=(1,1)$). Another way to see ranking-dominance is from the fact that every cell in $J^-$ has a unique corresponding cell in $J^+$ on its top-left direction, as indicated by the dashed arrows. From Proposition \ref{prop:loop-rank}, this implies that the choice rates data \textit{cannot} be represented as WU with this specific ranking of comparison difficulty.

\subsubsection{Example: Ranking of Comparison Difficulty for Three Options}\label{sec: cplx-ranking-exp}

I now illustrate the idea in the context of three options $\{x_1,x_2,x_3\}$. Note that while our result above is limited to choice rates on a loop instead of full pairing, this distinction vanishes for the case of three options, affording us the complete analysis.\footnote{For a general number of choice options and full pairing, we can obtain necessary conditions for iterating over all possible loops.}

\textbf{Ranking of comparison difficulty without heterogeneity.} Recall the analysis in Proposition \ref{prop:loop-rank}, I need to check the ranking-dominance between the ranking functions, $\pi$ for comparison difficulty $c_j$ and $\sigma$ for choice magnitude $\mu_j$ over the index subsets with different signs of preferences $J^+$ and $J^-$. Without loss of generality, consider the ranking of comparison difficulty $c(x_1,x_2)>c(x_1,x_3)>c(x_2,x_3)$. Using our simplified notation $c_j := c(x_j,x_{j+1})$, this is $c_1 > c_3 > c_2$, and thus corresponds to the ranking function $\pi(1)=1, \pi(2)=3, \pi(3)=2$. The strategy is to examine each possible preference ordering (determining $J^+$ and $J^-$), then derive what constraints the magnitude ranking $\sigma$ must satisfy to avoid ranking-dominance and thus ensure the compatibility with $\pi$.

As an example, consider the case of $J^+ = \{1,2\}, \ J^- = \{3\}$, which means $\rho(x_1,x_2) > \frac{1}{2}$, $\rho(x_2,x_3) > \frac{1}{2}$, and $\rho(x_3,x_1) < \frac{1}{2}$. Since $\#J^+=2>\#J^-=1$, $J^-$ cannot ranking-dominate $J^+$, which can be seen by taking $s=t=3$ in inequality \eqref{eqn:ranking-domination-def}. Hence it remains to choose rankings $\sigma$ such that $J^+$ does not ranking-dominate $J^-$, which requires
\begin{equation}
    \#\{j\in J^+\mid\sigma(j)\le s,\pi(j)\le t\}\  <\  \#\{j\in J^-\mid\sigma(j)\le s,\pi(j)\le t\}
    \label{eqn: eg-ranking}
\end{equation}
holds for some $s,t \in \{1,2,3\}$. Since $\#J^- = 1$, \eqref{eqn: eg-ranking} holds when the left side equals 0 while the right side equals 1. Given $J^-=\{3\}$ and $\pi(3) = 2$ and $\pi(1)=1$, this forces either:
\begin{itemize}
    \item $\sigma(3) = 1\quad \implies\quad \mu_3 > \max(\mu_1, \mu_2)$, or
    \item $\sigma(3) = 2$, $\sigma(1) = 3$, $\sigma(2) = 1\quad \implies\quad \mu_2 > \mu_3 > \mu_1$.\footnote{For simplicity, I ignore cases where $\bm\rho$ contains duplicate values, hence $\sigma$ is a permutation. In the full analysis in Proposition \ref{prop:loop-rank}, duplicate values of $\rho$ are allowed.}
\end{itemize}
Combined together, this simply requires $\mu_3 > \mu_1$, namely $\rho(x_1, x_3) > \rho(x_1, x_2)$. 

Intuitively, the case $J^+ = \{1,2\}, \ J^- = \{3\}$ corresponds to the preference ordering  $x_1 \succ x_2 \succ x_3$, which means $u(x_1) - u(x_3)>u(x_1) - u(x_2)$. Combined with the given requirement $c(x_1, x_3)<c(x_1, x_2)$, rationalization requires
$$
\frac{u(x_1)-u(x_3)}{c(x_1, x_3)} > \frac{u(x_1)-u(x_2)}{c(x_1, x_2)}\implies \rho(x_1,x_3)>\rho(x_1,x_2).
$$

Applying this analysis to all six preference orderings yields the complete characterization:
\begin{itemize}
    \item $x_1\succ x_2\succ x_3$ or $x_3\succ x_2\succ x_1$.\\ The additional rationalizability constraint is $\left|\rho(x_1,x_3)-\tfrac{1}{2}\right|>\left|\rho(x_1,x_2)-\tfrac{1}{2}\right|$. This yields a rationalizable region comprising half of the small cube corresponding to the preference ordering, depicted in green in Figure \ref{fig:cplx-ranking-indiv}.
    \item $x_1 \succ x_3 \succ x_2$ or $x_2 \succ x_3 \succ x_1$.\\ No additional constraints are imposed beyond the hypercube boundaries. The entire cube is rationalizable under these orderings, as the pair with maximal utility difference coincides with the pair of maximal comparison difficulty. These regions are depicted in blue.
    \item $x_3 \succ x_1 \succ x_2$ or $x_2 \succ x_1 \succ x_3$.\\ The additional rationalizability constraint is $\left|\rho(x_2, x_3)-\frac{1}{2}\right| > \max\{ \left|\rho(x_1, x_2)-\frac{1}{2}\right|, \left|\rho(x_1, x_3)-\frac{1}{2}\right| \}$. This yields the most restrictive rationalizable region, comprising one-third of the cube. These correspond to the case where maximal utility difference coincides with minimal comparison difficulty, depicted in orange.
\end{itemize}

\begin{figure}[H]
    \centering
    \begin{minipage}{0.48\linewidth}
    \includegraphics[width=\linewidth,trim={5cm 9cm 1cm 4cm},clip]{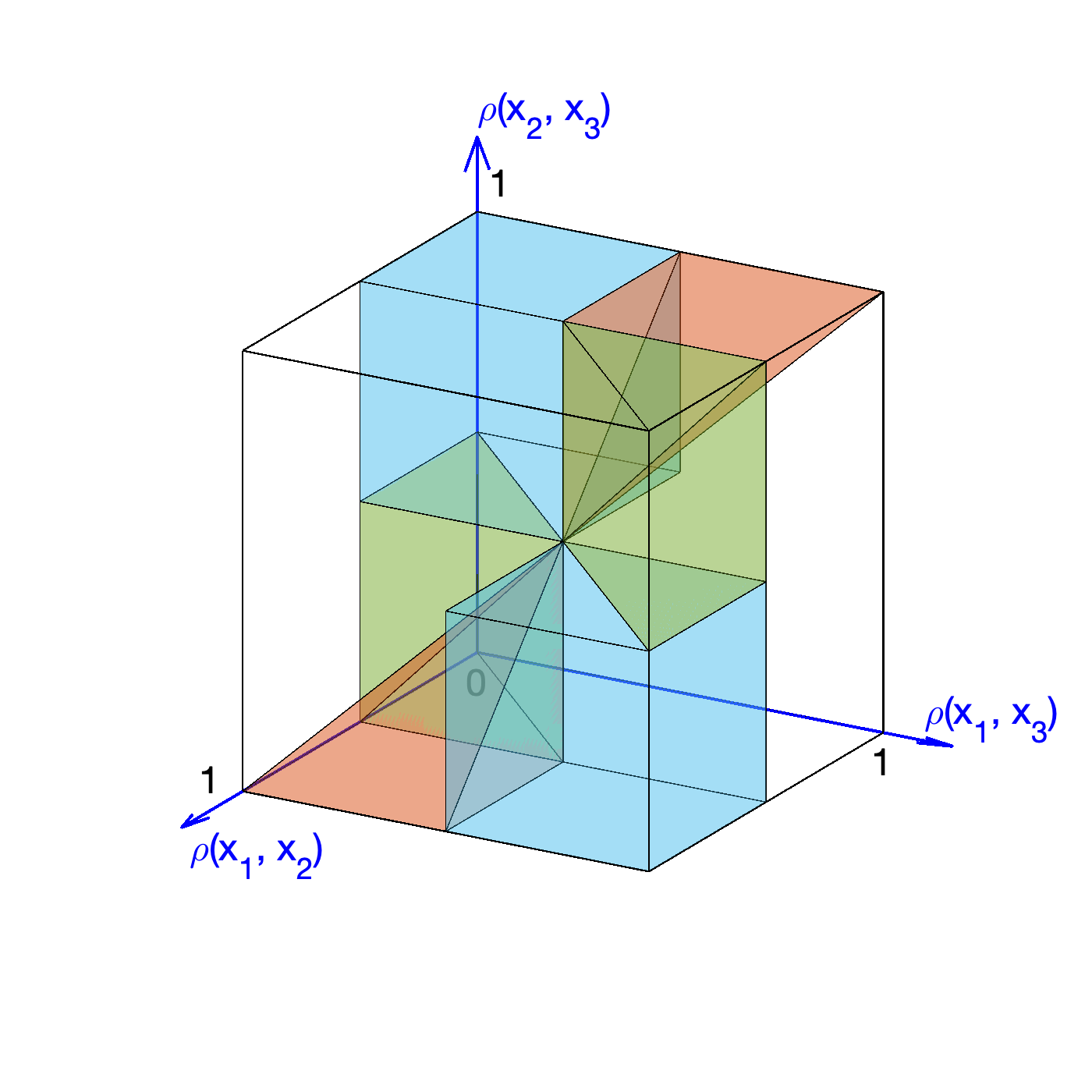}
    \caption{Individual WU-rationalizable \\ with $c(x_1,x_2)>c(x_1,x_3)>c(x_2,x_3)$}
    \label{fig:cplx-ranking-indiv}
    \end{minipage}
    \begin{minipage}{0.48
    \linewidth}
    \includegraphics[width=\linewidth,trim={5cm 9cm 1cm 4cm},clip]{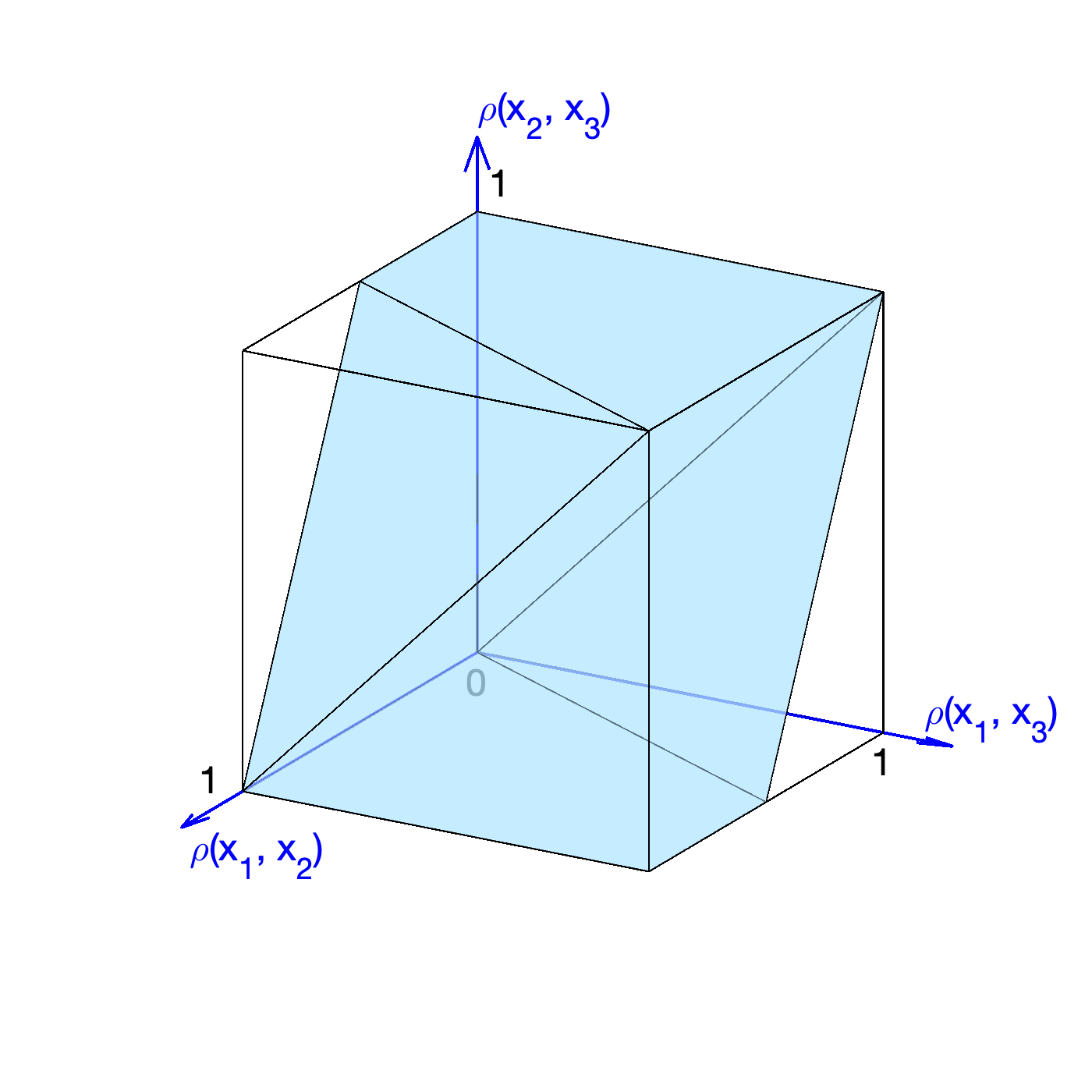}
    \caption{Collective  WU-rationalizable \\ with $c_i(x_1,x_2)>c_i(x_1,x_3)>c_i(x_2,x_3),\forall i$}
    \label{fig:cplx-ranking-collective}
    \end{minipage}
\end{figure}

\textbf{Ranking of comparison difficulty with heterogeneity.} From individual to collective rationalizability, I extend the given ranking of comparison difficulty to a heterogeneous population. In particular, I allow each individual in the population to have possibly different comparison difficulty functions $c_i(x,y)$, as long as all $c_i$ share a common ranking function. 
\begin{definition}[\textbf{Collective Weak Utility with a Common Ranking of Difficulty}]
\

    A \textit{Collective WU representation with a common ranking of comparison difficulty} has the form
        \begin{equation}
            \rho(x,y)= \sum_{i=1}^N \lambda_i F_i \left(\frac{u_i(x)-u_i(y)}{c_i(x,y)}\right) ,\quad \forall(x,y)\in Z^2
        \end{equation}
        with utility functions $u_i:Z\to\R$, semimetrics $c_i:Z^2\to\R^+$, and functions $F_i$ strictly increasing and satisfying $F_i(x)=1-F_i(-x)$, weights $\lambda_i\ge 0$ that sum to 1, and $c_i$ follows the same ranking function $\pi$ for any $i$.
\end{definition}
It naturally follows that the characterization for collective rationalizability with a common ranking of comparison complexity is the convex hull of that of the individual rationalizability. For the previous example, the collective WU-rationalizable set under the commonly shared ranking $c(x_1,x_2)>c(x_1,x_3)>c(x_2,x_3)$ is visualized in Figure \ref{fig:cplx-ranking-collective}, with the facet-defining inequalities
\begin{equation}
    1 \le 2\rho(x_1, x_2) + \rho(x_2, x_3) + \rho(x_3, x_1) \le 3.
\end{equation}
Recall that for collective WU without ranking requirements, the characterizing condition is $ \frac{1}{2} \le \rho(x_1, x_2) + \rho(x_2, x_3) + \rho(x_3, x_1) \le \frac{5}{2}$, which is cyclic symmetric and representing the stochastic transitivity while allowing for heterogeneity. Here with the ranking on comparison difficulty, we see an increased weight on the term $\rho(x_1, x_2)$ in the cyclic inequality. Intuitively, this is because $\rho(x_1, x_2)$ is pushed toward $\frac{1}{2}$ due to the large comparison difficulty, and dampens the information it reveals about the utility difference $u(x_1) - u(x_2)$. The increased weight compensates for this dampened information in the rationalizability.

\begin{figure}[H]
    \centering
    \begin{minipage}{0.32\linewidth}
    \centering
    \includegraphics[width=\linewidth,trim={5cm 9cm 1cm 4cm},clip]{figs/3d_viz_figs/collective_c12.png}
    \subcaption{\centering Collective WU with\\ $c(x_1,x_2)$ the largest}
    \end{minipage}
    \begin{minipage}{0.32\linewidth}
    \centering
    \includegraphics[width=\linewidth,trim={5cm 9cm 1cm 4cm},clip]{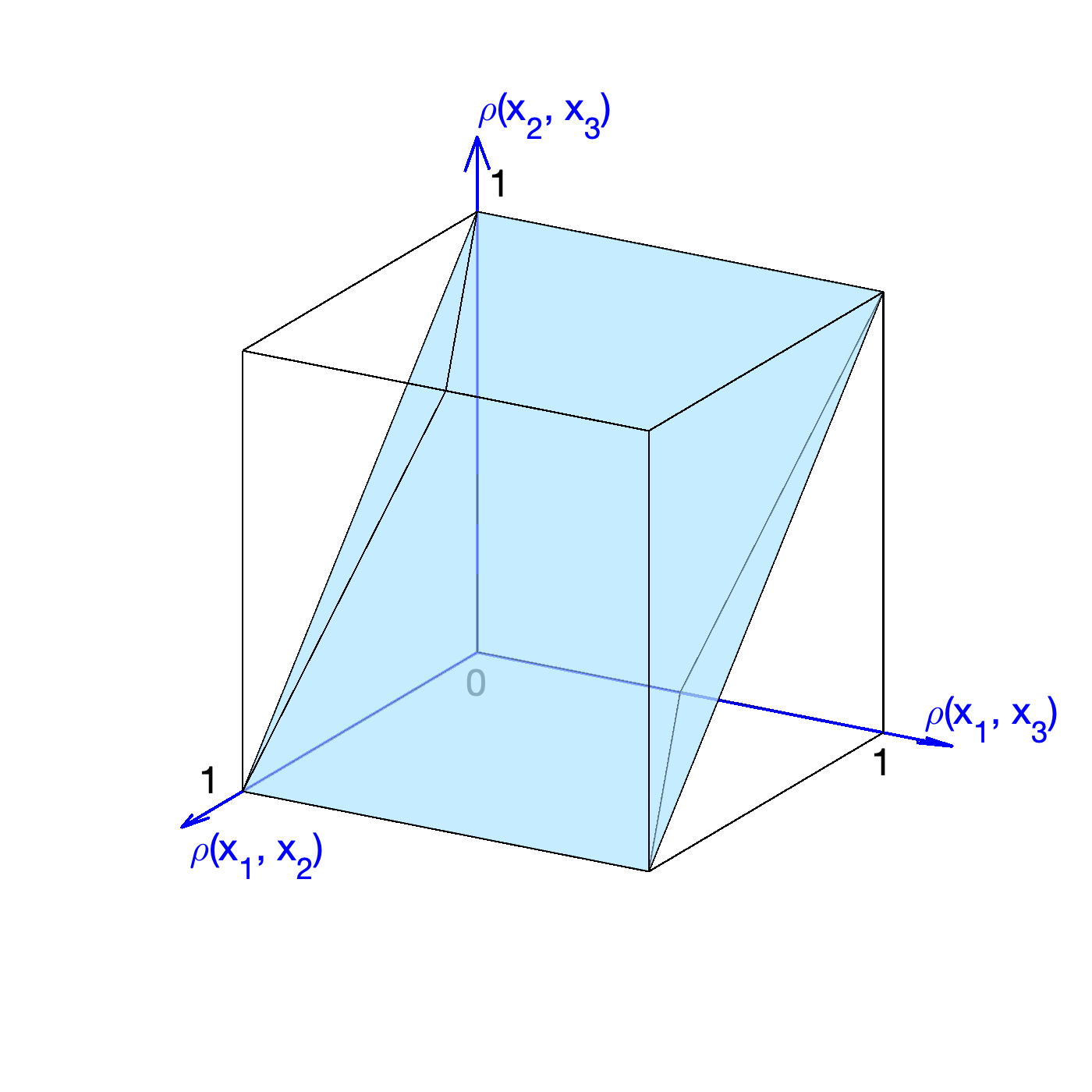}
    \subcaption{\centering Collective WU with\\ $c(x_1,x_3)$ the largest}
    \end{minipage}
    \begin{minipage}{0.32\linewidth}
    \centering
    \includegraphics[width=\linewidth,trim={5cm 9cm 1cm 4cm},clip]{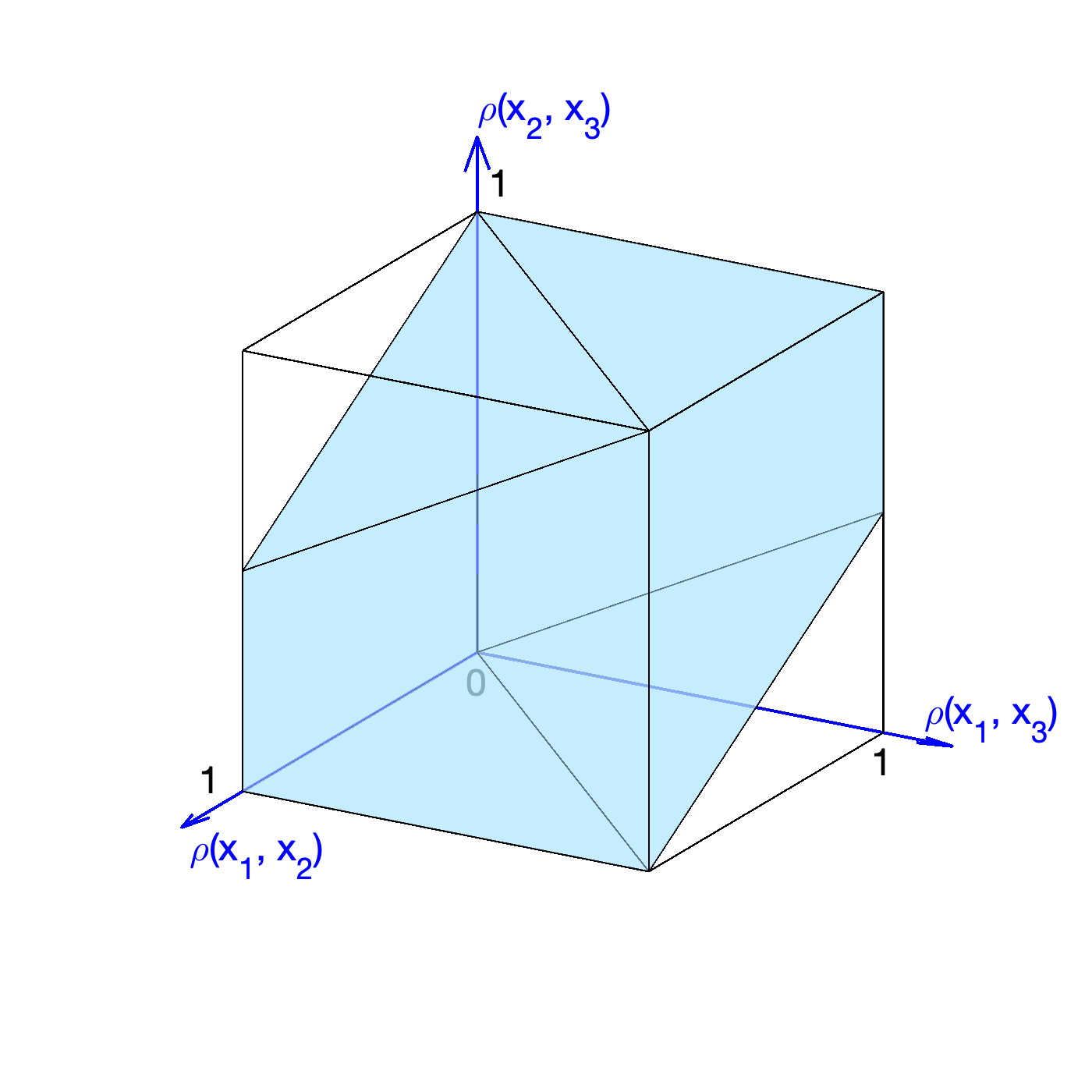}
    \subcaption{\centering Collective WU with\\ $c(x_2,x_3)$ the largest}
    \end{minipage}
    \caption{Collective WU rationalizable with different rankings of comparison difficulty}
    \label{fig:cplx-ranking}
\end{figure}
For all 6 possible rankings of $c$, we see that the rationalizable set only depends on which comparison is the most difficult, as shown in Figures \ref{fig:cplx-ranking}.

Furthermore, by comparing the difference between these rationalizable sets, one can identify when the pair with highest comparison difficulty can be \textit{uniquely} determined from the aggregate choices while allowing for different tastes. I formalize this result in the following Lemma \ref{lem:coll-wu-cplx-ranking}, with visualization in Figure \ref{fig:collective-1-ranking}. While Lemma \ref{lem:coll-wu-cplx-ranking} only states the conditions for $c(x_1, x_2) > \max(c(x_2, x_3), c(x_3,x_1))$, the other two cases can be obtained by simply re-indexing.

\begin{lemma}[\textbf{Collective WU with a Uniquely Identified Common Ranking of Difficulty}]
\label{lem:coll-wu-cplx-ranking}
\

Consider $n=3$ options $\{x_1,x_2,x_3\}$. Let the triple $\bm\rho=(\rho(x_1,x_2),\rho(x_2,x_3),\rho(x_3,x_1))\in[0,1]^3$ denote the population choice probability. Assume everyone shares the ranking of comparison difficulty. $\bm\rho$ is represented by collective WU  exclusively with the following common ranking of comparison difficulty
\[
c_i(x_1,x_2)>\max\{c_i(x_2,x_3),\,c_i(x_3,x_1)\},
\]
\textit{if and only if} either of the following systems of linear inequalities holds:

\medskip
\noindent
\begin{minipage}{0.45\textwidth}
\textbf{(A)}
\begin{equation*}
    \left\{\begin{matrix}
        \ 2\rho_{12}+\rho_{23}+\rho_{31} \le 3\\
        \ \rho_{12}+2\rho_{23}+\rho_{31} \ge 3 \\
        \ \rho_{12}+\rho_{23}+2\rho_{31} \ge 3
    \end{matrix}\right.
\end{equation*}
\end{minipage}
\hfill
\begin{minipage}{0.45\textwidth}
\textbf{(B)}
\begin{equation*}
    \left\{\begin{matrix}
        \ 2\rho_{12}+\rho_{23}+\rho_{31} \ge 1\\
        \ \rho_{12}+2\rho_{23}+\rho_{31} \le 1 \\
        \ \rho_{12}+\rho_{23}+2\rho_{31} \le 1
    \end{matrix}\right.
\end{equation*}
\end{minipage}
\end{lemma}

\begin{figure}[H]
    \centering
    \includegraphics[width=0.5\linewidth,trim={5cm 9cm 1cm 4cm},clip]{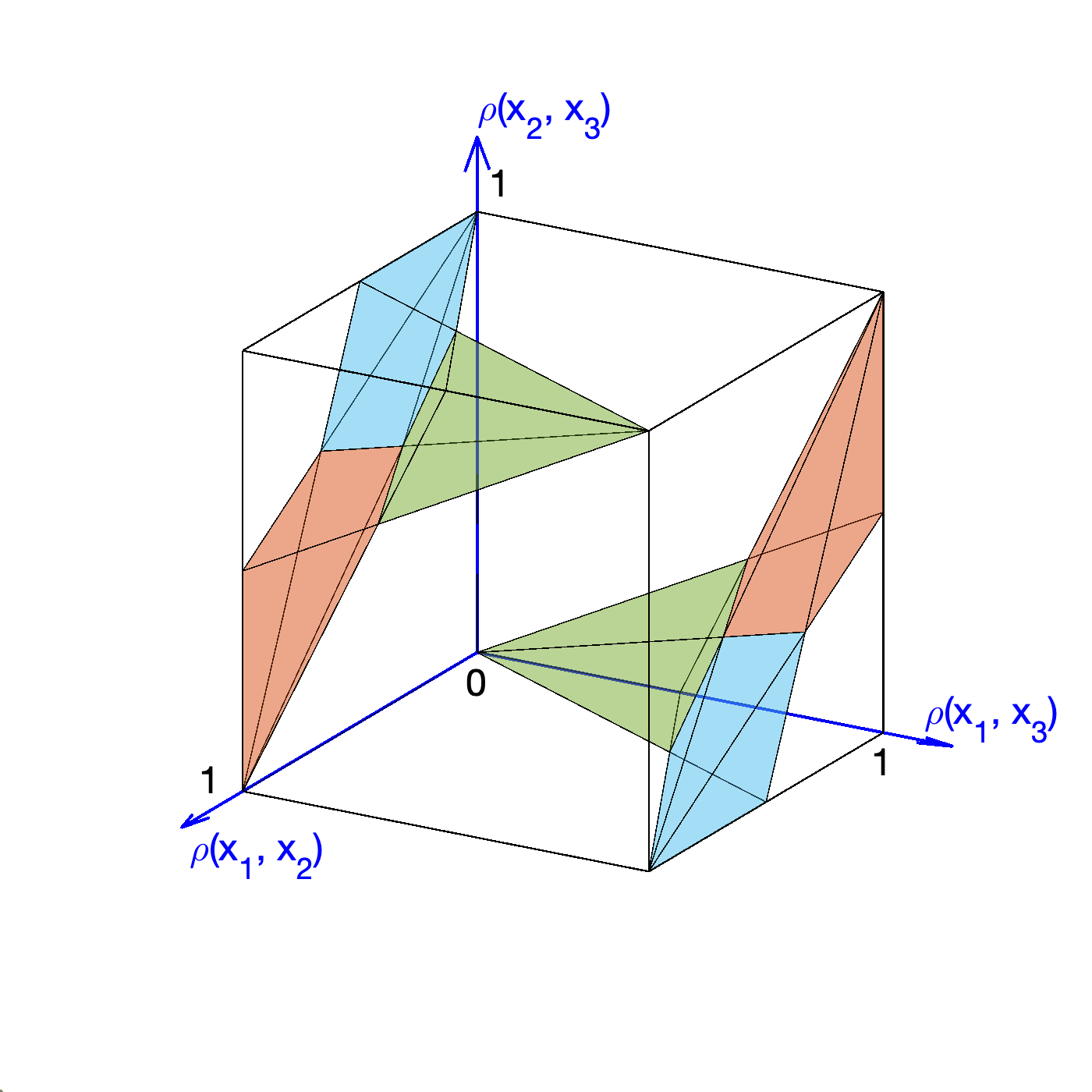}
    \caption{Collective WU representation where the pair with highest comparison difficulty can be determined. Blue: $c(x_1,x_2)>\max \{c(x_2,x_3),c(x_3,x_1)\}$. Green: $c(x_1,x_3)>\max\{c(x_1,x_2),c(x_2,x_3)\}$. Orange: $c(x_2,x_3)>\max\{c(x_1,x_2),c(x_1,x_3)\}$}
    \label{fig:collective-1-ranking}
\end{figure}

\newpage
\section{Statistical Testing of Collective Rationalizability}\label{sec: test}
The theoretical characterization in \autoref{thm:collective} establishes when choice patterns are collectively rationalizable while explicitly allowing for heterogeneity. This section addresses the question of whether estimated choice frequencies may, up to \textit{sampling variation}, have arisen from a collectively rationalizable model.

Testing collective rationalizability involves testing multiple inequalities, a problem extensively studied in econometrics. The highlight of this section is twofold: first, to provide a broad class of d.g.p.s over which the resulting test based on collective rationalizability is uniformly valid, where stochastic transitivity cannot be directly tested; second, to illustrate how to overcome the main challenge that the limiting distribution of the test statistic depends discontinuously on the null hypothesis. I adopt the numerical delta method \citep{fang2019inference, hong2018numerical} to provide uniformly valid inference, by using the fact that the target function is directionally differentiable.

The remainder proceeds as follows. Section \ref{sec: dgp} specifies the general class of d.g.p.s. Section \ref{sec: numerical-delta} develops the test statistic and establishes uniform validity. Section \ref{sec: monte-carlo} analyzes finite-sample performance through Monte Carlo simulations.

\subsection{Data Generating Process}\label{sec: dgp}
Let $Z$ be a finite set of $n$ choice options, generating $m:=\binom{n}{2}$ binary choice problems. Consider $N$ subjects randomly drawn from a population, where each subject faces a potentially different subset of these problems.
For each subject $i$ and binary choice problem $j$, researchers observe:
\begin{enumerate}[(i).]
    \item An indicator $I_{ij}=\mathbbm{1}[\text{subject $i$ faces problem $j$}]$.
    \item The average choice frequency $\bar{\rho}_{ij}$ when $I_{ij}=1$.
\end{enumerate} 
Note that complete observations are not required. Subjects may face only subsets of possible choices, and repetitions may be insufficient to estimate individual choice probabilities. These limitations prevent testing based on individual rationalizability. Instead, I construct the population choice frequency from the available data:
\begin{equation}
    \hat{\bm\rho}=[\hat{\rho}_1,\dots,\hat\rho_m]^T,\quad \text{where}\quad \hat\rho_j = \frac{\sum_{i=1}^N \bar{\rho}_{ij} I_{ij}}{\sum_{i=1}^N I_{ij}}, \ \forall j=1,\dots,m.
        \label{eqn:estimator-standard-dgp}
\end{equation}

The framework covers three key sources of variation. (i) \textit{Heterogeneity}: aggregating choices while assuming a representative agent may produce apparently irrational patterns simply due to heterogeneous preferences, rather than irrational individuals. (ii) \textit{Assignment variation}: subjects may face different subsets of choice problems. In experiments, this reflects time or budget constraints where each subject is assigned only a subset of questions; in observational data, this captures consumers encountering different products. (iii) \textit{Repeated choices}: subjects' responses may exhibit varying degrees of correlation when they face the same choice problem multiple times. This variation highlights the need for careful experimental design when researchers assume independent decisions to estimate individual choice probabilities.

I specify the first two sources of variation in Assumption \ref{assumption_1}, and the third in Assumption \ref{assumption_2}.

\begin{assumption}[\textbf{Sampling and Random Assignment}]\label{assumption_1}
\

\begin{enumerate}[(a)]
    \item Individual choice probabilities: $\bm\rho_i \sim_{\text{iid}} F$, where $F$ is a distribution over $[0,1]^m$ with mean $\bm\mu$ and covariance $\bm\Sigma_0$.
    \item Assignments: $\bm K_i \sim_{\text{iid}} G$, where $G$ is a distribution over $\{0, 1, \dots, \bar{K}\}^m$ and $K_{ij}$ denotes how many times subject $i$ faces problem $j$.
    \item Independence: $\bm\rho_i \perp \bm K_i$ (no selection bias between preferences and assignments).
\end{enumerate}
\end{assumption}

\begin{remark}{(For Assumption \ref{assumption_1})}
\
\begin{enumerate}[(i).]
    \item Assumption \ref{assumption_1}(a) is standard in the literature. (b) and (c) generalize \cite{kitamura2018nonparametric} to allow repeated observations. This framework encompasses various assignment mechanisms, such as random sampling with replacement, where subjects face randomly drawn problems; or lab experiments, where subjects are randomly assigned to sessions with predetermined choice problems.
    \item The independence assumption in (c) rules out selection bias between preferences and the problems subjects face. This assumption is naturally satisfied under random assignment in experiments, but may be violated in observational settings where consumers self-select into choice situations. For instance, price-sensitive consumers may seek out discount stores and thus more frequently encounter choice problems there.
\end{enumerate}
\end{remark}

When subjects face the same choice problem multiple times, their responses may exhibit varying degrees of correlation. I accommodate the full spectrum: from independent responses (treating each as a fresh decision) to perfect correlation (always repeating the initial choice) and any intermediate pattern. The test remains valid regardless of the actual correlation structure, offering robustness in settings where the degree of correlation is unknown or may vary across subjects.

Formally, let $\hat{\rho}_{ij}^{(k)}$ denote subject $i$'s response the $k^{\text{th}}$ time facing problem $j$, and let $$\hat{\bm\rho}_{ij} = [\hat{\rho}_{ij}^{(1)}, \dots, \hat{\rho}_{ij}^{(K_{ij})}]^T\in\{0,1\}^{K_{ij}}$$ denote the vector of all responses from subject $i$ facing problem $j$. For any $k$ and $p\in[0,1]$, let $H(k,p)$ denote a distribution over valid correlation matrices for a $k$-dimensional Bernoulli($p$) random vector.

\begin{assumption}[\textbf{Repeated Choice Responses}]\label{assumption_2}
\ 

For subject $i$ facing problem $j$ multiple times ($K_{ij} > 0$):
\begin{enumerate}[(a)]
    \item Each response: $\hat{\rho}_{ij}^{(k)} | \rho_{ij} \sim \text{Bernoulli}(\rho_{ij})$ for $k=1,\dots, K_{ij}$
    \item Correlation structure: $\text{Cov}(\hat{\bm\rho}_{ij} | K_{ij}, \bm\rho_{i}) = \rho_{ij}(1-\rho_{ij}) \bm V_{ij}$, where $\bm V_{ij} \sim H(K_{ij}, \rho_{ij})$
\end{enumerate}
\end{assumption}

The observed average choice frequency for subject $i$ on problem $j$ is:
\begin{equation}
    \bar{\rho}_{ij} = \left\{\begin{matrix}
            \frac{1}{K_{ij}} \sum_{k=1}^{K_{ij}} \hat{\rho}_{i j}^{(k)} & \text{ if } K_{ij} \neq 0 \\
            0 & \text{ if } K_{ij} = 0
        \end{matrix} \right. \ ,\quad \forall i=1,\dots,N,\ \forall j=1,\dots, m.
\end{equation}

\begin{lemma}\label{lemma:rho_hat}
Under Assumptions \ref{assumption_1}, \ref{assumption_2}, and Regularity Assumptions \autoref{assumption-regularity} in Appendix \ref{appdx: regularity}
    \begin{equation}
        \sqrt{N} \left( \hat{\bm\rho} - \bm\mu \right) \overset{d}{\longrightarrow}
    \mathcal{N}\big(\bm 0, \ \bm \Sigma \big),
    \end{equation}
    where $\bm\Sigma$ is defined in Appendix \ref{pf:rho_hat}.
\end{lemma}

The asymptotic covariance $\bm\Sigma$ captures the above three sources of variation: population distribution ($\bm\Sigma_0$), assignment patterns ($G$), and within-subject correlation across repeated observations. Under the Regularity Assumptions \ref{assumption-regularity} (d) and (e), $\bm\Sigma$ is positive definite.

\subsection{Test Statistic and Uniform Validity}\label{sec: numerical-delta}
Let $Q \subseteq [0,1]^m$ denote the polytope of collective rationalizability (SS, MU, or WU). By \autoref{thm:collective}, the testing problem is
\begin{equation*}
    H_0: \quad \bm\mu \in Q, \quad \text{against}\quad H_1: \quad \bm\mu \not\in Q.
\end{equation*}

The null hypothesis holds when the mean of population distribution lies within the rationalizability polytope. The test statistic measures the distance from the estimated frequencies to this polytope.
\begin{equation}
     T_N = \sqrt{N} \left\| \hat{\bm\rho} - \proj_Q(\hat{\bm\rho},\bm\Omega)\right\|_\Omega,
     \label{eqn: Tn}
\end{equation}
where $\left\|\bm x\right\|_\Omega = \sqrt{\bm x^T\bm \Omega \bm x}$, and $\proj_Q(\cdot,\bm\Omega)$ is the projection onto $Q$ with $\bm\Omega$-induced distance, defined as
\begin{equation}
    \proj_Q(\bm x,\bm\Omega) = \arg\min_{\bm q \in Q} \ \left\|\bm x - \bm q\right\|_\Omega.
\end{equation}

Define the transformation $\phi(\bm x) = \| \bm x - \proj_{Q}(\bm x) \|$,\footnote{Note that since the numerical delta method that will be introduced shortly provides a uniformly valid testing procedure, the choice of weighting matrix $\bm\Omega$ in the distance is immaterial: rescaling by $\bm\Omega$ corresponds to a linear transformation of the data and the null polytope. Hence, I set $\bm\Omega = I$ without loss of generality. See Appendix \ref{appdx: test-addi-norm} for more detailed argument.} then the null can then be written as
\begin{equation*}
    H_0:\quad \phi(\bm\mu)\le 0\quad \text{against}\quad H_1:\quad \phi(\bm\mu)> 0,
\end{equation*}
and $T_N = \sqrt{N}\left(\phi(\hat{\bm\rho})\right)$, where from Lemma \autoref{lemma:rho_hat}, $\sqrt{N}\left(\hat{\bm\rho} - \bm\mu\right)\to \G_0:= \mathcal{N}\big(\bm 0, \ \bm \Sigma \big)$. The function $\phi$ measures the Euclidean distance projecting onto a convex polytope. Its limiting distribution does not vary continuously in $\bm\mu$ when $\bm\mu$ lies on the boundary of $Q$: intuitively, as $\bm\mu$ moves across edges and facets of the polytope, the distribution changes abruptly. This discontinuity invalidates the consistency of standard bootstrap methods in this setting \citep{andrews2000inconsistency, fang2019inference}.\footnote{See Appendix \ref{appdx: subsampling} for the proof of \textit{pointwise} validity of the subsampling procedure in our case.} Fortunately, the function $\phi$ is \textit{directionally differentiable} and, moreover, Lipschitz and convex. These properties allow me to adopt the numerical delta method from \citet{hong2018numerical}, which provides \textit{uniformly valid inference}.

Formally, the key insight is that for directionally differentiable $\phi$,
\begin{equation}
    \sqrt{N}\left(\phi(\hat{\bm\rho}) - \phi(\bm\mu)\right) \to \phi'_{\bm\mu}(\G_0),
\end{equation}
where $\phi'_{\bm\mu}$ denotes the directional derivative of $\phi$ at $\bm\mu$ \citep{shapiro1991asymptotic}.\footnote{See Appendix \ref{appdx: Tn} for the detailed analysis on the pointwise asymptotic distribution of $T_N$, which relates closely to \citet{fang2019inference}'s discussion on convex set projections.} Define 
\begin{equation}
    \Z_N^* = \sqrt{N}\left(\hat{\bm\rho}_b - \hat{\bm\rho}\right),\label{eqn:bootstrap}
\end{equation}
where $\hat{\bm\rho}_b$ is parameter estimates obtained using standard bootstrap. Then for positive step size $\epsilon_N\to 0,\sqrt{N}\epsilon_N\to\infty$, the pointwise validity of the inference is obtained as
\begin{equation}
    \hat{\phi}'_N(\Z_N^*) := \frac{1}{\epsilon_N}\left(\phi(\hat{\bm\rho}+\epsilon_N\Z_N^*)-\phi(\hat{\bm\rho})\right)\to \phi'_{\bm\mu}(\G_0).
\end{equation}
Moreover, for the uniform inference, $\phi$ being convex implies that $\frac{1}{\epsilon_N}\left(\phi(\bm\mu+\epsilon_N\G_0)-\phi(\bm\mu)\right)$ first-order stochastically dominates $\phi'_{\bm\mu}(\G_0)$, and $\phi$ being Lipschitz implies that $\frac{1}{\epsilon_N}\left(\phi(\bm\mu+\epsilon_N\G_0)-\phi(\bm\mu)\right)$ can be replaced by its feasible sample version $\hat{\phi}'_N(\Z_N^*)$;
these together imply that the $1-\alpha$ quantile of $\hat{\phi}'_N(\Z_N^*)$ provides uniform size control \citep{romano2012uniform}.

\begin{theorem}[\textbf{Uniformly Valid Inference for Collective Rationalizability}]
\label{thm:uniform-test}
\

Consider the following numerical delta method with $S$ bootstrap replications. For each $s=1,\dots,S$, draw $\Z_s$ from the distribution of $\Z_N^*$ defined in (\ref{eqn:bootstrap}). For a given $\epsilon_N$ such that $\epsilon_N\to 0,\sqrt{N}\epsilon_N\to\infty$, evaluate for each $s$:
\begin{equation}
        \hat{\phi}'_N(\Z_s^*) := \frac{1}{\epsilon_N}\left(\phi(\hat{\bm\rho}+\epsilon_N\Z_s^*)-\phi(\hat{\bm\rho})\right).
\end{equation}
Let $\hat{c}_{1-\alpha}$ denote the $1-\alpha$ quantile of the empirical distribution of $\hat{\phi}'_N(\Z_s^*)$.
Under Assumptions \autoref{assumption_1}, \autoref{assumption_2}, and \autoref{assumption-regularity},
    \begin{equation}
        \liminf_{N\to \infty} \inf_{\bm\mu\in Q} \Pr[T_N\le \hat{c}_{1-\alpha}] \ge 1-\alpha.
    \end{equation}
\end{theorem}

\subsection{Monte Carlo Simulations}\label{sec: monte-carlo}
The performance of the proposed testing procedure is evaluated through Monte Carlo simulations, focusing on Collective Simple Scalability rationalizability with three choice options. 
I construct a power curve by varying the distance between the population mean $\bm\mu$ and the rationalizability polytope $Q\subset \mathbb{R}^3$. Starting from a boundary point $\bm\mu_0=[2/3,\ 1/3,\ 2/3]^T$ on a facet of $Q$, I move  along the normal vector $[1,-1,1]^T$ in both directions, setting $\bm\mu = \bm\mu_0 + d\cdot [1,\ -1,\ 1]^T$ where $d$ controls the distance from the boundary. For $d<0$, $\bm\mu$ lies inside $Q$ (collectively rationalizable); for $d>0$, $\bm\mu$ lies outside $Q$ (not collectively rationalizable). The simulation examines  nine values $d\in\{0,\ \pm0.01,\ \pm0.02,\ \pm0.04,\ \pm0.08\}$. Small deviations such as $d=0.01$ represent near-rational behavior, while $d=0.08$ indicates substantial departures from rationalizability.

\begin{table}[H]
    \centering
    \begin{threeparttable}
    \caption{Population means $\bm\mu$ used for the Monte Carlo simulations\tnote{a}}
    \begin{tabular}{|cccccccccc|}
    \hline
         & $\bm\mu_4$ & $\bm\mu_3$ & $\bm\mu_2$ & $\bm\mu_1$ & $\bm\mu_0$ & $\bm\mu_1'$ & $\bm\mu_2'$ & $\bm\mu_3'$ & $\bm\mu_4'$ \\
         \hline
         &&&&&&&&&\\[-0.2cm]
        $\bm\mu$ & $\begin{bmatrix}
            0.59\\0.41\\0.59
        \end{bmatrix}$ & $\begin{bmatrix}
           0.63\\ 0.37\\ 0.63
        \end{bmatrix}$ & $\begin{bmatrix}
           0.65\\ 0.35\\ 0.65
        \end{bmatrix}$ & $\begin{bmatrix}
           0.66\\ 0.34\\ 0.66
        \end{bmatrix}$ & $\begin{bmatrix}
           0.67\\ 0.33\\ 0.67
        \end{bmatrix}$ & $\begin{bmatrix}
           0.68\\ 0.32\\ 0.68
        \end{bmatrix}$ & $\begin{bmatrix}
           0.69\\ 0.31\\ 0.69
        \end{bmatrix}$ & $\begin{bmatrix}
           0.71\\ 0.29\\ 0.71
        \end{bmatrix}$ & $\begin{bmatrix}
           0.75\\ 0.26\\ 0.75
        \end{bmatrix}$\\[1cm]
        \hline
        $d$ & $-0.08$ & $-0.04$ & $-0.02$ & $-0.01$ & 0 & 0.01 & 0.02 & 0.04 & 0.08\\
        \hline
    \end{tabular}
    \label{tab:dgp}
    \begin{tablenotes}
\footnotesize
\item[a] Assume the individual choice frequency $\bm\rho_i\sim_{iid}F$, where $F$ is a distribution over $[0,1]^3$ with mean $\bm\mu$ and covariance $\bm\Sigma_0=0.01\times I_3$. $\bm\mu_0$ lies on the facet of the collective-SS rationalizable polytope $Q$. Define $\bm\mu=\bm\mu_0+d\cdot[1,-1,1]^T$. An increase of $d$ implies $\bm\mu$ getting deeper inside the polytope, with $d=0$ indicating $\bm\mu=\bm\mu_0$ on the boundary, and $d<0$ indicating $\bm\mu$ lies inside $Q$.
\end{tablenotes}
\end{threeparttable}
\end{table}
\vspace{-0.4cm}
I consider 18 d.g.p.s combining nine different values of population mean $\bm\mu$ (Table \ref{tab:dgp}) and two within-subject correlation structures: (i). Independent responses (d.g.p.s 1–9): subjects treat repeated problems as fresh decisions; (ii). Perfect correlation (d.g.p.s 10–18): subjects always repeat their initial choice. Individual choice probabilities are drawn from $\bm\rho_i\sim_{iid} \mathcal{N}(\bm\mu,\bm\Sigma_0)$ with $\bm\Sigma_0 = 0.01 \times I_3$. Each subject faces 10 randomly selected binary problems (with replacement). Sample sizes are $N\in\{100,\ 200,\ 500\}$. Figure \ref{fig: test_dgp} visualizes the population distributions for three representative cases, illustrating how individual heterogeneity generates a cloud of choice probabilities around the population mean.

\begin{figure}[H]
\centering
\begin{minipage}{0.3\textwidth}
\centering
\includegraphics[width=0.9\linewidth,trim={6cm 11cm 1cm 5cm},clip]{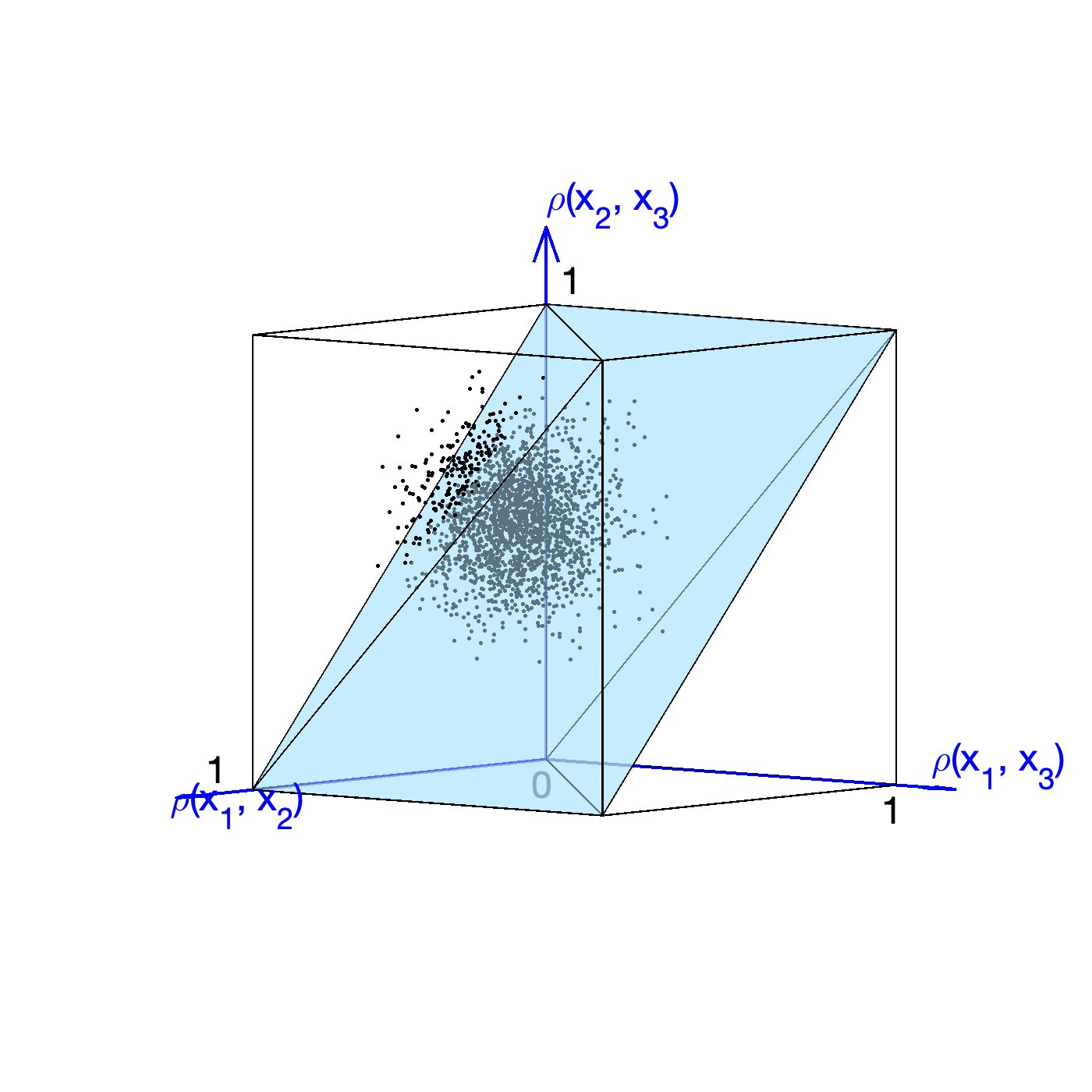}
\subcaption{$\bm\rho_i\sim_{iid}\mathcal{N}(\bm\mu_4,\bm\Sigma_0)$.}
\label{fig: test_interior}
\end{minipage}
\begin{minipage}{0.3\textwidth}
\centering
\includegraphics[width=0.9\linewidth,trim={6cm 11cm 1cm 5cm},clip]{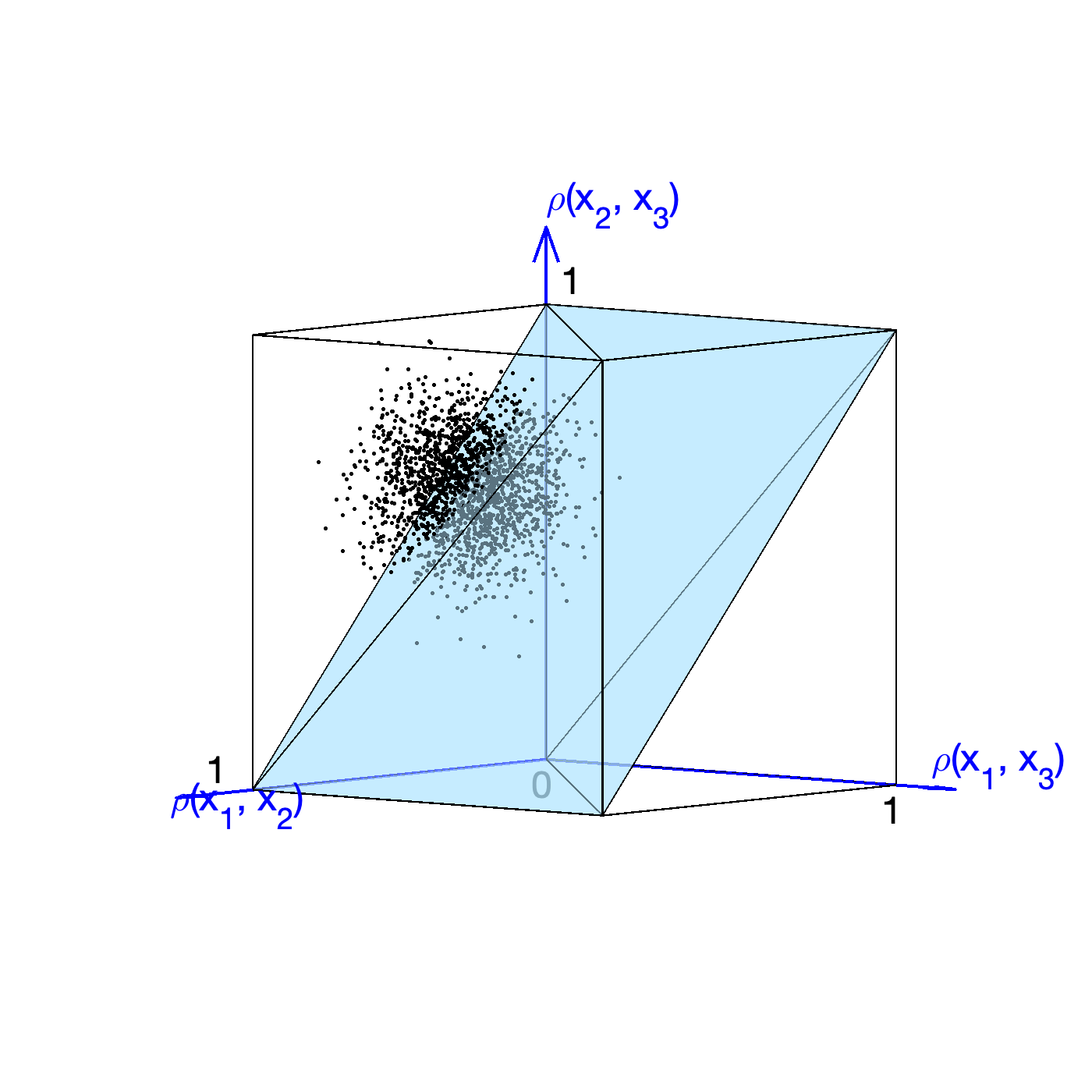}
\subcaption{$\bm\rho_i\sim_{iid}\mathcal{N}(\bm\mu_0,\bm\Sigma_0)$.}
\label{fig: test_boundary}
\end{minipage}
\begin{minipage}{0.3\textwidth}
\centering
\includegraphics[width=0.9\linewidth,trim={6cm 11cm 1cm 5cm},clip]{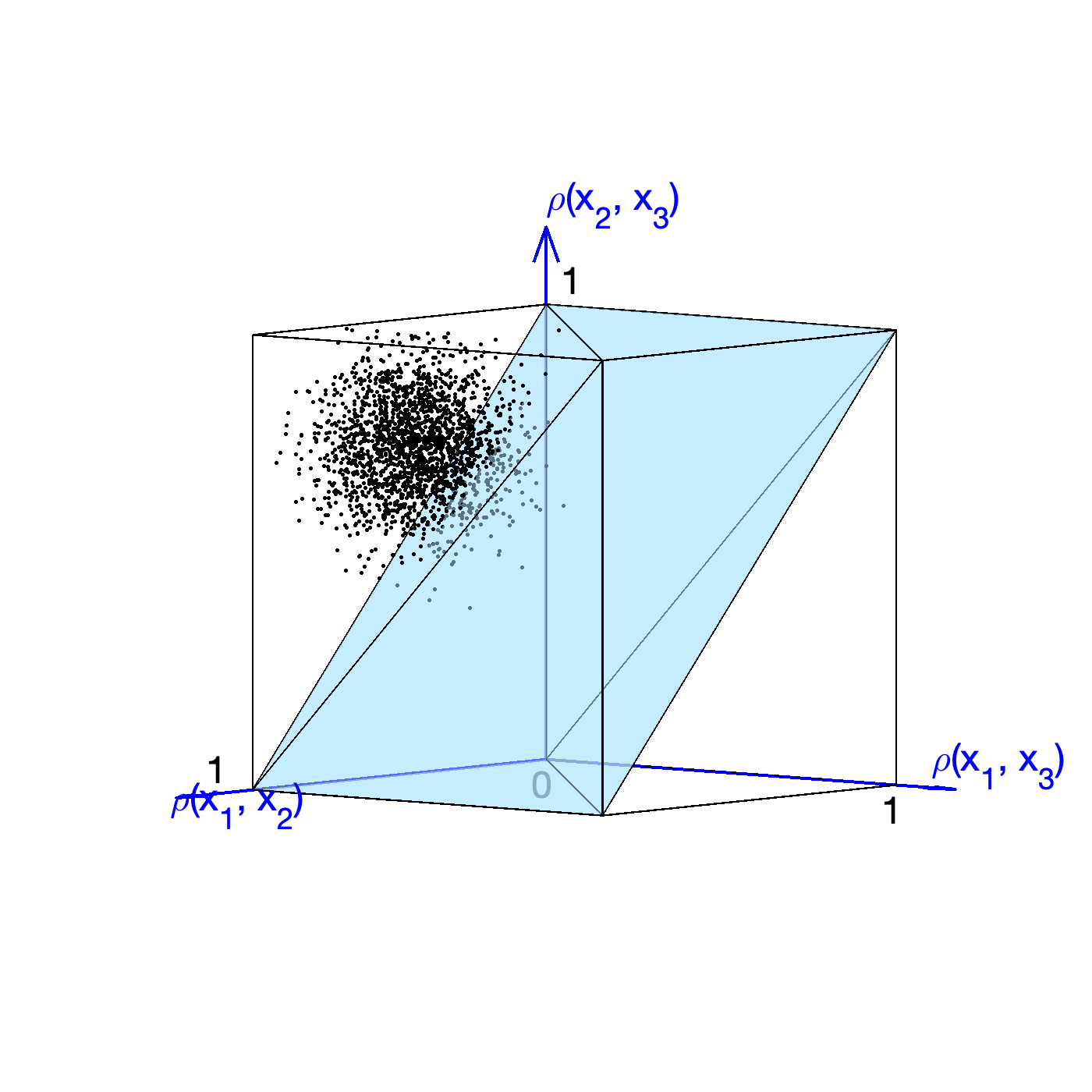}
\subcaption{$\bm\rho_i\sim_{iid}\mathcal{N}(\bm\mu_4',\bm\Sigma_0)$.}
\label{fig: test_outside}
\end{minipage}
\caption{Visualization of population distribution $\bm\rho_i$, with 2000 sample size, under different $\bm\mu$. \textbf{(a).} $\bm\mu_4$ is in the interior of $Q$. \textbf{(b).} $\bm\mu_0$ is on its boundary. \textbf{(c).} $\bm\mu_4'$ is outside the interior of $Q$.}
\label{fig: test_dgp}
\end{figure}
\vspace{-0.4cm}
I implement the numerical delta method with 5000 replications per d.g.p, using bootstrap sample size $N_b=N$, and two step sizes $\epsilon_N = N^{-1/3}$ or $\epsilon_N = N^{-1/6}$, satisfying $\epsilon_N\to0,\sqrt{N}\epsilon_N\to\infty$. The nominal size is set at $\alpha=0.05$.

\begin{table}[H]
    \centering
    \setlength{\tabcolsep}{10pt}
    \begin{threeparttable}
    \caption{Monte Carlo results: rejection rates\tnote{a}}
    \begin{tabular}{|cccccccccc|}
    \hline
      $N$  & $\bm\mu_4$ & $\bm\mu_3$ & $\bm\mu_2$ & $\bm\mu_1$ & $\bm\mu_0$ & $\bm\mu_1'$ & $\bm\mu_2'$ & $\bm\mu_3'$ & $\bm\mu_4'$\\
      \hline
       \multicolumn{10}{|p{12cm}|}{d.g.p. 1--9: Each decision is independent; step size $\epsilon_N = N^{-1/3}$.}\\
      \hline
      100  & 0.000 & 0.000 & 0.000 & 0.018 & 0.054 & 0.147 & 0.319 & 0.757 & 0.999\\
       200 & 0.000 & 0.000 & 0.001 & 0.010 & 0.048 & 0.214 & 0.495 & 0.950 & 1.000\\
       500 & 0.000 & 0.000 & 0.000 & 0.002 & 0.051 & 0.371 & 0.828 & 1.000 & 1.000\\
       \hline
       \multicolumn{10}{|l|}{d.g.p. 1--9: Each decision is independent; step size $\epsilon_N = N^{-1/6}$.}\\
       \hline
       100  & 0.000 & 0.000 & 0.002 & 0.017 & 0.056 & 0.154 & 0.334 & 0.761 & 0.999\\
       200 & 0.000 & 0.000 & 0.001 & 0.008 & 0.049 & 0.216 & 0.499 & 0.948 & 1.000\\
       500 & 0.000 & 0.000 & 0.000 & 0.002 & 0.050 & 0.360 & 0.821 & 1.000 & 1.000\\
       \hline
       \multicolumn{10}{|l|}{d.g.p. 10--18: Each decision replicates the first encounter; step size $\epsilon_N = N^{-1/3}$.}\\
       \hline
       100  & 0.000 & 0.000 & 0.007 & 0.019 & 0.053 & 0.119 & 0.245 & 0.590 & 0.980\\
       200 & 0.000 & 0.000 & 0.003 & 0.013 & 0.056 & 0.167 & 0.370 & 0.836 & 1.000\\
       500 & 0.000 & 0.000 & 0.000 & 0.004 & 0.053 & 0.263 & 0.652 & 0.994 & 1.000\\
       \hline
       \multicolumn{10}{|l|}{d.g.p. 10--18: Each decision replicates the first encounter; step size $\epsilon_N = N^{-1/6}$.}\\
       \hline
       100  & 0.000 & 0.000 & 0.006 & 0.020 & 0.056 & 0.137 & 0.260 & 0.602 & 0.981\\
       200 & 0.000 & 0.000 & 0.003 & 0.016 & 0.056 & 0.172 & 0.374 & 0.828 & 1.000\\
       500 & 0.000 & 0.000 & 0.000 & 0.003 & 0.050 & 0.261 & 0.656 & 0.995 & 1.000\\
       \hline
    \end{tabular}
    \label{tab:power-curve}
    \begin{tablenotes}
\footnotesize
\item[a] See Table \ref{tab:dgp} for definition of $\bm\mu$ vectors. Recall that $\bm\mu_0$ is on the boundary of $Q$. From $\bm\mu_0$ to $\bm\mu_4$, the population mean moves away from the polytope; from $\bm\mu_0$ to $\bm\mu_4'$, the population mean moves into the interior of the polytope. We keep the variance of the population distribution $\bm\Sigma_0$ the same across all d.g.p.s at $\bm\Sigma_0=0.01\times I_3$. In the d.g.p.s 1--9, subjects make each decision independently; while in the d.g.p.s 10--18, subjects replicate the decision they made when first encountering the choice problem.
    All entries are rejection rates computed from 5000 simulations, with sample sizes $N=\{100,200,500\}$ and bootstrap sample size $N_b=N$. The step size for the numerical delta is $\epsilon_N = N^{-1/3}$ or $\epsilon_N = N^{-1/6}$.
\end{tablenotes}
\end{threeparttable}
\end{table}

Table \ref{tab:power-curve} reports rejection rates under different d.g.p.s. The test maintains appropriate size control at the boundary $(\bm\mu_0)$, with empirical rejection rates close to the nominal $5\%$. Power increases monotonically with distance from the polytope and sample size.
In particular, comparing across d.g.p.s 1--9 and d.g.p.s 10--18, when decisions are independent across repetitions, a smaller sample size is required to achieve $95\%$ power compared to that when repetitions are perfectly correlated. Comparing within d.g.p.s 1--9 and d.g.p.s 10--18, the step size choice $\epsilon_N=N^{-1/6}$ yields slightly higher power than $\epsilon_N=N^{-1/3}$, although overall both choices perform similarly well. 

\section{Empirical Application: Two Existing Experiments}\label{sec: application}
The main contribution of this paper is the characterization of collective rationalizability (\autoref{thm:collective}) and its corresponding statistical testing (\autoref{thm:uniform-test}). Together, they enable testing for random choice models while allowing for heterogeneity, which is especially important when choices are collected from possibly heterogeneous subjects but researchers have to assume representative agent for estimation or testing. The framework addresses the empirical challenge that researchers are often limited to finite observations of individuals making few or no repeated choices from the same choice set. In the following, I demonstrate the practical applicability of collective rationalizability by re-analyzing data from two existing experiments, and show that comparison difficulty can be not only theoretically but also empirically confused with heterogeneity in both cases.

\textbf{Experiment 1: Snack food choices \citep{clithero2018improving}.} This paper studies binary choices among snack foods. Subjects $(N=31)$ choose between all possible pairs from 17 snacks, yielding 136 binary comparisons per subject. Each subject faces each pair exactly once, resulting in 4,216 total observations. The within-subject design provides complete choice data across all pairs, but with only single observation per subject-pair combination.\footnote{In particular, I only use the Two-Alternative Forced-Choice Task (2AFC-Task) for the test. More discussion on the Yes-No Choice Task (YN-Task) in \citet{clithero2018improving} will be discussed later.}

\textbf{Experiment 2: Perceptual similarity of images \citep{prashnani2018pieapp}.} The experiment examines perceptual judgments of image similarity. Given a reference image and two distorted versions based on it, subjects indicate which distortion appears more similar to the reference. The experiment uses 40 reference images, each evaluated alongside 15 distorted versions of itself, yielding 16 total versions per reference. With 40 responses collected per comparison across the resulting 4,800 binary choice problems, this yields 192,000 total observations.\footnote{In particular, I only use their test set, where 40 human responses for each pairwise comparison are collected.}

These datasets illustrate two key strengths of our collective rationalizability approach:

First, neither experiment has natural ``ground truth'' preferences. In food choices, researchers want to remain agnostic about which snacks people should prefer. For image similarity, the complexity of distortion types (blur, noise, rotations, etc.) makes it hard to assess the objective similarity.\footnote{Actually the primary research question asked in the original paper \citet{prashnani2018pieapp} is: what is the human perception of these different distortions, and how does it differ from existing metrics? Therefore, researchers of course do not want to impose strong assumptions about image similarities before the experiment.} Other existing methods \citep{tversky1969substitutability, enke2023quantifying} require either known preferences or specially designed tasks with objective answers, neither of which can be applied here.

Second, both datasets exemplify common data limitations that prevent direct application of existing revealed-preference approaches. In \citet{clithero2018improving}'s data, single observations per subject-pair makes estimating individual probabilities nontrivial. In \citet{prashnani2018pieapp}'s data, individual identifiers are unavailable publicly. Tests of stochastic transitivity require estimating individual choice probabilities, which is impossible in both cases. Our collective framework is designed precisely for such scenarios, testing models using aggregate patterns rather than requiring repeated individual observations.

\subsection{Application 1: Snack Food Choices}
I first apply the methods to the experimental data from \cite{clithero2018improving}, who studies binary choices among $n=17$ different snack foods. In their two-alternative forced choice (2AFC) task, subjects $(N=31)$ chose which of two items they would prefer to eat at the end of the experiment. Each subject faced all 136 possible binary pairs exactly once. Examples of the snacks are shown in Figure \ref{fig:candy-pictures}.

\begin{figure}[H]
\centering
\begin{minipage}{0.21\textwidth}
\centering
\includegraphics[width=\linewidth]{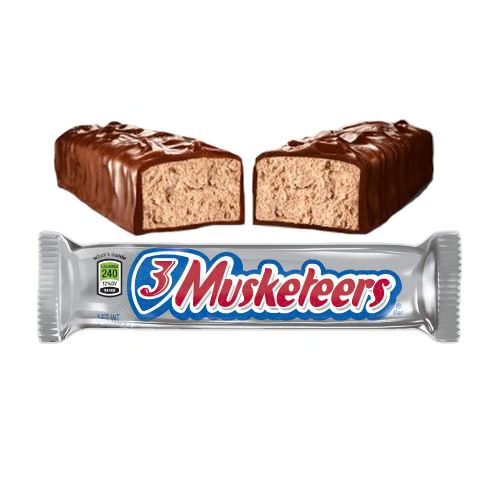}
\subcaption{3 Musketeers}
\end{minipage}
\begin{minipage}{0.21\textwidth}
\centering
\includegraphics[width=\linewidth]{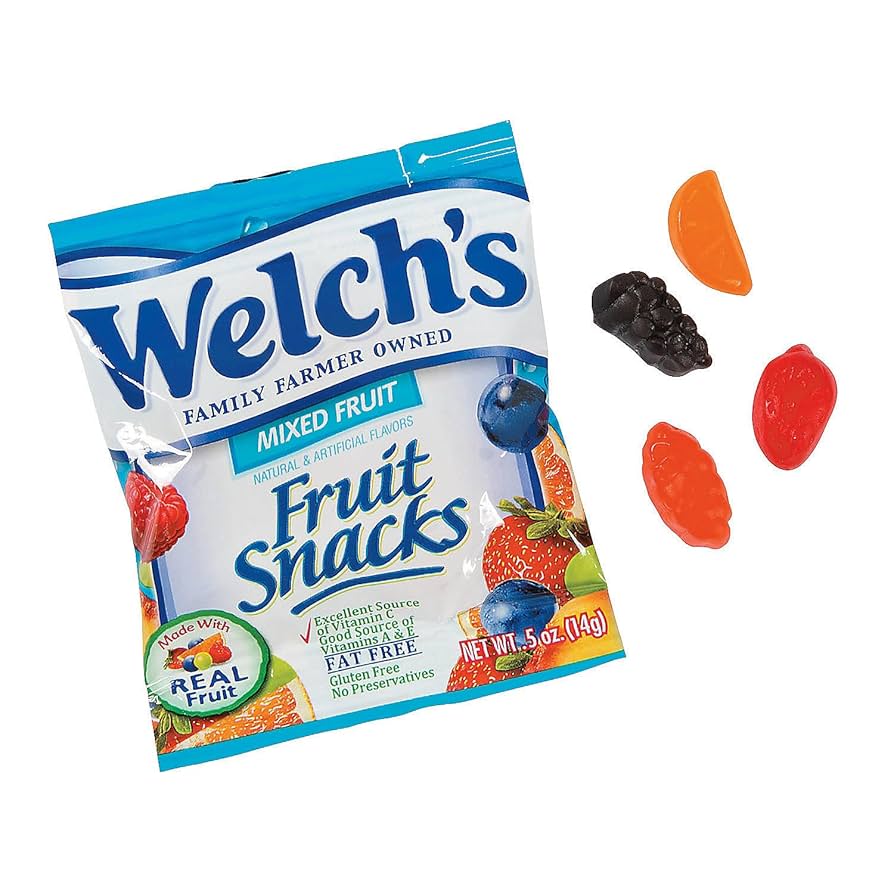}
\subcaption{Fruit Snacks}
\end{minipage}
\begin{minipage}{0.21\textwidth}
\centering
\includegraphics[width=\linewidth]{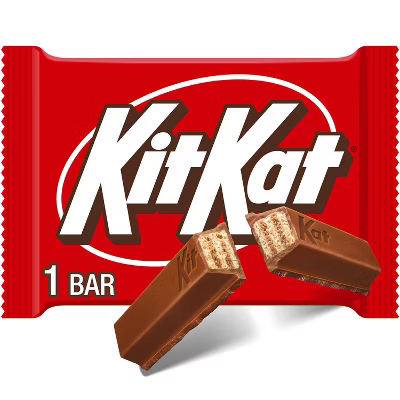}
\subcaption{Kit Kat}
\end{minipage}
\begin{minipage}{0.25\textwidth}
\centering
\includegraphics[width=\linewidth]{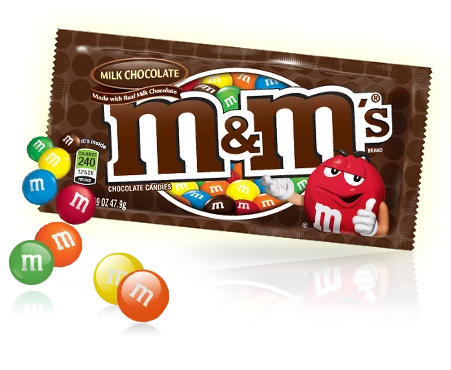}
\subcaption{M\&M's Chocolate}
\end{minipage}
\caption{Examples of the types of snack foods involved in the experiment.}
\label{fig:candy-pictures}
\end{figure}

\textbf{Comparison of the various models.} I analyze the aggregate choice frequencies, i.e., the fraction of subjects choosing snack $i$ over $j$, and examine the violation rate under the various models, SS/MU/WU and Collective SS. The results are summarized in Table \ref{tab:snack-data-violation-rate}. We see that Simple Scalability is heavily violated when treating the data as coming from a representative agent: among the 680 possible triples of snacks, 41.8\% violate strong stochastic transitivity, and this number increases to 99\% when we check all subsets of size 5. This raises a key question: does this violation stem from varying comparison difficulty across pairs, or can heterogeneity alone explain it? 

\begin{table}[H]
  \centering
  \begin{threeparttable}
  \caption{Rationalizability of the dataset under different models\tnote{a}}
  \label{tab:snack-data-violation-rate}
    \begin{tabular}{|cccccc|}
    \hline
      \makecell{Choice \\ [-3pt] subset size} & \makecell{Total subset \\ [-3pt] count} & \makecell{Individual SS \\ [-3pt] violation rate} & \makecell{Individual MU \\ [-3pt] violation rate} & \makecell{Individual WU \\ [-3pt] violation rate} & \makecell{Collective SS \\ [-3pt] violation rate} \\
      \hline
      3 & 680 & 0.418 & 0.107 & 0.040 & 0.000 \\
      4 & 2,380 & 0.856 & 0.324 & 0.130 & 0.000 \\
      5 & 6,188 & 0.990 & 0.577 & 0.264 & 0.000 \\
      \hline
    \end{tabular}
     \begin{tablenotes}
\footnotesize
\item[a] We enumerate all choice option subsets of size 3, 4 and 5 in the dataset with full pairwise choice rates, and report the total number of subsets, and the proportion with population choice rates violating individual SS/MU/WU and collective SS.
\end{tablenotes}
\end{threeparttable}
\end{table}

A researcher not accounting for heterogeneity across subjects may observe the significant reduction in violation rates after introducing comparison difficulty and conclude that comparison difficulty is the key factor. However, using the characterization developed in this paper, we see that all subsets with sizes 3, 4, and 5 are \textit{collectively} rationalizable by Simple Scalability, which cannot be rejected with statistical power, implying that heterogeneity alone can be sufficient for explaining all violations of Simple Scalability in this case.

\textbf{Examples of transitivity violation.} Consider a concrete example: Fruit Snacks, M\&M's, and KitKat. The aggregate choice rates are:
\begin{equation*}
    \rho(\mathrm{KitKat}, \mathrm{M\&M's}) = 0.839,\quad \rho(\mathrm{M\&M's}, \mathrm{FruitSnacks}) = 0.581,\quad \rho(\mathrm{KitKat}, \mathrm{FruitSnacks}) = 0.613.
\end{equation*}
They violate strong stochastic transitivity if assuming a representative agent: the  preference for KitKat over M\&M's and M\&M's over Fruit Snacks should imply an even stronger preference for KitKat over Fruit Snacks, hence $\rho(\mathrm{KitKat}, \mathrm{FruitSnacks})>0.839$, yet we observe only 0.613. This pattern can be explained by \textit{comparison difficulty} alone, that KitKat is more difficult to be compared with Fruit Snacks than with M\&M's. However, it can also be explained by \textit{heterogeneity} alone. The choice rates are \textit{collectively} rationalized by Simple Scalability: imagine around half the subjects prefer chocolate-based snacks while half prefer fruit flavors, and among the two chocolate-based snacks, most prefer KitKat. Each individual could satisfy Simple Scalability perfectly, yet their aggregation produces the observed violation.

In another concrete example: Fruit Snacks, M\&M's, and Reese's, the aggregate choice rates are:
\begin{equation*}
    \rho(\mathrm{M\&M's}, \mathrm{FruitSnacks}) = 0.581,\quad \rho(\mathrm{FruitSnacks}, \mathrm{Reese's}) = 0.581,\quad \rho(\mathrm{Reese's}, \mathrm{M\&M's}) = 0.581.
\end{equation*}
They violate even weak stochastic transitivity, which implies that violations of Simple Scalability in this example cannot be explained by comparison difficulty alone. However, they are again collectively rationalizable by Simple Scalability. In fact, one way to collectively rationalize the choice is reminiscent of the Condorcet paradox \citep{condorcet1785}; the only difference is that instead of deterministic choices, each individual chooses the preferred alternative with probability 0.87.

\textbf{Evidence of heterogeneity.} The collective rationalizability reveals that heterogeneity alone suffices to explain all violations of Simple Scalability in this dataset. Fortunately, independent validation is available: \citet{clithero2018improving}'s experimental design includes a separate Yes-No (YN) task that allows testing for heterogeneity directly. In the YN task, subjects are asked whether they preferred each of the 17 snacks over a common reference snack (3 Musketeers candy), with 10 repetitions per snack. Although one cannot directly check stochastic transitivity using the YN task since there is no cycle structure, these repetitions help us look at individual preferences more closely. Table \ref{tab:snack-data-heterogeneity} summarizes subjects' responses over the two questions involving Fruit Snacks and M\&M's.

\begin{table}[H]
  \centering
  \begin{threeparttable}
  \caption{Summary of subjects' responses in the YN task for certain questions\tnote{a}}
  \label{tab:snack-data-heterogeneity}
    \begin{tabular}{|ccc|}
    \hline
      Fruit Snacks $\succ$ 3 Musketeers? & M\&M's $\succ$ 3 Musketeers? & Number of subjects \\
      \hline
      ``Yes'' all 10 times & ``No'' all 10 times & 9 \\
      ``No'' all 10 times & ``Yes'' all 10 times & 6 \\
      ``Yes'' all 10 times & ``Yes'' all 10 times & 7 \\
      ``No'' all 10 times & ``No'' all 10 times & 4 \\
      \multicolumn{2}{|c}{Other response patterns} & 5 \\
      \hline
      \multicolumn{2}{|c}{Total} & 31 \\
      \hline
    \end{tabular}
     \begin{tablenotes}
\footnotesize
\item[a] The majority of subjects (26 out of 31) provided consistent responses to both questions across all 10 repetitions. Among these 26 subjects, responses were distributed across all four possible Yes/No combinations, indicating substantial heterogeneity in individual preferences.
\end{tablenotes}
\end{threeparttable}
\end{table}

Table \ref{tab:snack-data-heterogeneity} shows that the majority of subjects provided \textit{consistent} responses across all 10 repetitions of both questions. Among these subjects, 9 exhibited the preference ordering Fruit Snacks $\succ$ 3 Musketeers $\succ$ M\&M's, while 6 others displayed the exactly opposite ordering M\&M's $\succ$ 3 Musketeers $\succ$ Fruit Snacks. This finding aligns with our earlier conjecture that the violation of Simple Scalability may stem from a population with polarized preferences between Fruit Snacks and M\&M's.  Additionally, 11 subjects responded with either ``Yes'' or ``No'' consistently for both questions across all 10 repetitions. While these responses are inconclusive regarding the relative preference between Fruit Snacks and M\&M's, they provide further evidence of substantial preference heterogeneity among subjects.

Furthermore, using the YN task data, I perform a statistical test of whether all individuals share the same preference signs (i.e., whether $\sgn(\rho_i-\tfrac{1}{2})$ is identical across individuals). A likelihood-ratio test rejects this hypothesis ($p<0.001$), confirming the preference heterogeneity in the population (details in Appendix \ref{appdx:test-heter}).

\subsection{Application 2: Image Similarity Perception}
I next apply the framework to data from \citet{prashnani2018pieapp}, who study how people judge similarity between images. In their experiment, subjects compare pairs of distorted images and identify which is more similar to a reference image (Figure \ref{fig: image-eg}). The dataset contains 40 reference images with 15 distortion types each, yielding 4,800 binary comparisons with 40 responses per comparison (192,000 total observations).\footnote{In particular, we only use their test set, where 40 human responses for each pairwise comparison are collected.}

In \citet{prashnani2018pieapp}, they model these judgments using Simple Scalability with a representative agent, estimating a single ``similarity score'' for each distorted image. 
\begin{align}
    \rho(A,B) = F(s_A,s_B) = \frac{1}{1+e^{s_A-s_B}}.
\end{align}
While this perceptual task differs from economic choice, the framework can be applied by treating similarity judgments analogously to preferences. Each person's perception of how similar image $A$ is to the reference constitutes their ``preference'' for $A$. The comparison difficulty between images $A$ and $B$ captures the ``similarity of similarity'': how difficult is it to compare the similarity of image $A$ versus that of image $B$ to the same reference image? This framework allows testing whether the representative-agent Simple Scalability model adequately captures the data, or whether heterogeneity and comparison difficulty play important roles.

Table \ref{tab:empirical-result-direct-check} shows substantial violations of Simple Scalability if assuming a representative agent. Among size-3 subsets, approximately 30\% violate strong stochastic transitivity; this rises to 70\% for size-4 subsets and exceeds 90\% for size-5 subsets. The representative-agent SS model clearly fails to capture the aggregate patterns.

\begin{table}[H]
  \centering
  \begin{threeparttable}
  \caption{Rationalizability of the dataset under different models\tnote{a}}
  \label{tab:empirical-result-direct-check}
    \begin{tabular}{|cccccc|}
    \hline
      \makecell{Choice \\ [-3pt] subset size} & \makecell{Total subset \\ [-3pt] count} & \makecell{Individual SS \\ [-3pt] violation rate} & \makecell{Individual MU \\ [-3pt] violation rate} & \makecell{Individual WU \\ [-3pt] violation rate} & \makecell{Collective SS \\ [-3pt] violation rate} \\
      \hline
      3 & 22,400 & 0.298 & 0.034 & 0.009 & 0.001 \\
      4 & 72,800 & 0.700 & 0.117 & 0.032 & 0.003 \\
      5 & 174,720 & 0.926 & 0.246 & 0.073 & 0.007 \\
      \hline
    \end{tabular}
     \begin{tablenotes}
\footnotesize
\item[a] We enumerate all choice option subsets of size 3, 4 and 5 in the dataset with full pairwise choice rates, and report the total number of subsets, and the proportion with population choice rates violating individual SS/MU/WU and collective SS.
\end{tablenotes}
\end{threeparttable}
\end{table}

Remarkably, allowing for heterogeneity virtually eliminates these violations. Collective SS is violated in only 0.1\% of size-3 subsets and remains below 1\% even for size-5 subsets. 
This finding has important implications: researchers modeling perceptual judgments may not need complex models with varying comparison difficulty. Instead, acknowledging that people have heterogeneous perceptions of similarity may suffice. Here heterogeneous perceptions mean that different individuals genuinely perceive image similarity differently, and what appears as intransitivity in aggregate data might reflect diverse but internally consistent individual perceptions.

\begin{table}[H]
  \centering
  \begin{ThreePartTable}
  \caption{Statistical test result of collective Simple Scalability\tnote{a}}
  \label{tab:empirical-stat-test}
    \begin{tabular}{|cc|}
    \hline
      \makecell{Choice \\ subset size} & \makecell{Rejection instances \\ at $\alpha=0.05$} \\
      \hline
      3 & 1 \\
      4 & 13 \\
      5 & 66 \\
      \hline
    \end{tabular}
    \begin{tablenotes}
    \footnotesize
    \vspace{0.3cm}
    \parbox{\textwidth}{%
      \item[a] I conduct the statistical testing for Collective Simple Scalability from \autoref{thm:uniform-test}, and report the number of instances where collective SS is rejected at the 0.05 significance level. To avoid some p-hacking concerns, I employ sample splitting. Using the first half of the data, I identify which choice subsets violate Collective SS based on the theoretical characterization. For these subsets that exhibit theoretical violations, I then conduct the statistical test from \autoref{thm:uniform-test} on the second half of the sample, reporting the number of rejections at the 0.05 significance level.
    }
    \end{tablenotes}
  \end{ThreePartTable}
\end{table}

\begin{figure}[H]
\centering
\begin{minipage}{0.32\textwidth}
    \centering
    \includegraphics[width=0.9\linewidth]{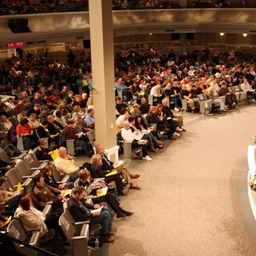}
    \subcaption*{Reference image $R$}
    \label{fig:ref}
\end{minipage}
\begin{minipage}{0.64\textwidth}
    \centering
    \begin{minipage}{0.48\textwidth}
        \centering
        \includegraphics[width=0.9\linewidth]{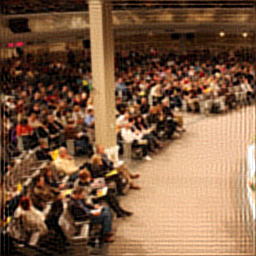}
        \subcaption{Distorted image $A$}
        \label{fig:distort-1}
    \end{minipage}
    \begin{minipage}{0.48\textwidth}
        \centering
        \includegraphics[width=0.9\linewidth]{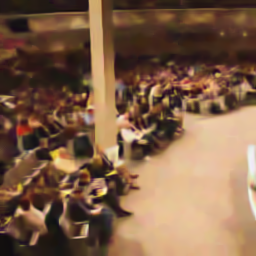}
        \subcaption{Distorted image $B$}
        \label{fig:distort-2}
    \end{minipage}\\[0.5em]
    \begin{minipage}{0.48\textwidth}
        \centering
        \includegraphics[width=0.9\linewidth]{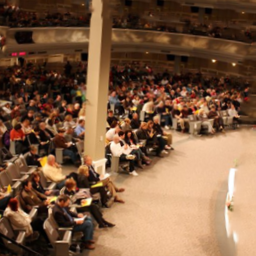}
        \subcaption{Distorted image $C$}
        \label{fig:distort-3}
    \end{minipage}
    \begin{minipage}{0.48\textwidth}
        \centering
        \includegraphics[width=0.9\linewidth]{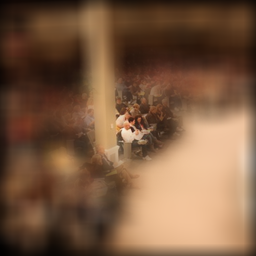}
        \subcaption{Distorted image $D$}
        \label{fig:distort-4}
    \end{minipage}
\end{minipage}
\caption{Examples of reference and distorted images. In \citet{prashnani2018pieapp}'s Amazon Mechanical Turk (MTurk) experiments, there are 40 reference images, and 15 different distortions are applied to each reference image. Subjects are asked to choose between all \textit{pairwise comparison}: Which image is more similar to the reference? In total, there are 4800 binary choice problems, and 40 responses for each pairwise comparison are collected. Individual-level data regarding which responses come from the same individual is not available. \\
\textbf{(a)(b).} Which image, $A$ or $B$, is more similar to the reference $R$? In this example of pairwise image comparison, 66.7\% of the subjects choose image $A$ as the more similar.\\ \textbf{(c)(d).} Which image, $C$ or $D$, is more similar to the reference $R$? In this example of pairwise image comparison, most people (92.5\%) choose image $C$ as the more similar. }
\label{fig: image-eg}
\end{figure}


\subsection{Comparison of Model Restrictiveness}
At this point, one might wonder: does Collective Simple Scalability fit the empirical data well simply because it is too permissive? To explore this question, I compute the relative volumes of the rationalizable sets for each model, or equivalently, the proportion of uniformly random choice data that each model would accept.

\begin{table}[H]
  \centering
  \begin{threeparttable}
  \caption{Relative volume of the rationalizable set of different models\tnote{a}}
  \label{tab:empirical-result-random-data}
    \begin{tabular}{|ccccc|}
    \hline
      \makecell{Choice \\ [-3pt] subset size} & Individual SS & Individual MU & Individual WU & Collective SS \\
      \hline
      3 & 0.250 & 0.500 & 0.750 & 0.667 \\
      4 & 0.008 & 0.092 & 0.375 & 0.250 \\
      5 & 0.00002 & 0.005 & 0.117 & 0.048 \\
      \hline
    \end{tabular}
     \begin{tablenotes}
\footnotesize
\item[a] 
Volumes computed analytically from model characterizations, except Individual MU for n=4,5 which required Monte Carlo simulation due to non-convexity.
\end{tablenotes}
\end{threeparttable}
\end{table}

Table \ref{tab:empirical-result-random-data} shows that Weak Utility (adding comparison difficulty only) is actually \textit{more permissive} than Collective Simple Scalability (adding heterogeneity only). For $n=3$, Collective SS accepts 66.7\% of the choice space compared to WU's 75\%. This gap widens dramatically as $n$ increases: for $n=5$, Collective SS accepts only 4.8\% versus WU's 11.7\%.
This finding has important implications. Despite being more restrictive, Collective Simple Scalability achieves lower violation rates on the real experimental datasets than Weak Utility for both empirical applications (see, Tables \ref{tab:snack-data-violation-rate}, \ref{tab:empirical-result-direct-check}). Following the idea of \citet{selten1991properties} and \citet{beatty2011demanding}, heterogeneity does a better job explaining the violation of Simple Scalability than comparison difficulty in both ways. 

\section{Conclusion}\label{sec: conclusion}
This paper develops a revealed preference framework for testing random choice models while explicitly allowing for heterogeneity. The central contribution, collective rationalizability, provides both theoretical characterizations and statistical tests to determine whether heterogeneity alone, comparison difficulty alone, or both required, explains violations of Simple Scalability when data is collected from possibly heterogeneous subjects. This framework addresses practical data limitations, the absence of known preferences and limited individual observations, that existing approaches often face in practice.

The empirical applications demonstrate the framework's applicability. In both the food choice and image similarity experiments, violations of Simple Scalability that appear substantial under representative-agent assumptions largely disappear when heterogeneity is properly accounted for through Collective Simple Scalability. Without this framework, these violations might be attributed to comparison difficulty. This demonstrates that not only theoretically, but also empirically, heterogeneity can be confused with comparison difficulty to explain the apparent intransitivities.

Several extensions would enhance the framework's theoretical foundations and empirical applicability:

\textbf{Theoretical extensions to multinomial choices.} Extending collective rationalizability from binary to multinomial choice with comparison difficulty would significantly broaden empirical applications. Many experiments involve choices from sets larger than two, and the framework's insights about heterogeneity versus comparison difficulty would prove valuable in these richer settings.

\textbf{Identification of difficulty with new experiments.} Future research could validate the test for collective rationalizability through carefully designed experiments. Building on the theoretical results for recovering the ranking of comparison difficulty developed in this paper, such experiments could empirically test different difficulty metrics in the laboratory settings.

\textbf{Uniformly valid statistical testing for transitivity.} Developing uniformly valid tests for individual rationalizability remains an open challenge. Such tests would enable more precise comparisons between heterogeneity-based and difficulty-based explanations of choice patterns, particularly when individual-level data is available.

\textbf{Micro-foundations.} The violations of Collective Simple Scalability provide valuable cases where heterogeneity alone cannot explain observations. New theories understanding when and why such violations of Collective Simple Scalability arise, and what they reveal about the underlying cognitive foundations of decision-making, is an interesting question for the future research.

\newpage
\bibliographystyle{apalike}
\bibliography{ref}

\newpage
\appendix
\section{Appendix}
\subsection{Half-space Representation for Rationalizability without Heterogeneity}\label{appdx: prf-a-b}

We first consider $\bm A_k$, $\bm b_k$ for SS, which requires us to write the inequality constraints from strong stochastic transitivity into the standard linear inequality form $\bm A \bm \rho \le \bm b$. Let $\pi_k: [1,\dots,n] \to [1,\dots,n]$, $k=1,\dots,n!$ be the different permutations over the set $[1,\dots,n]$. Also, write $\pi_k(u)$ alternatively as $\pi_k^u$ for conciseness. The inequality constraints can thus be written as 
\begin{align}
    d(\pi_k^u \pi_k^v, \pi_k^u\pi_k^w; \bm\rho) &\le c(\pi_k^u \pi_k^v, \pi_k^u\pi_k^w) \quad \forall 1\le u < v < w \le n, \ k=1,\dots, n! \\
    d(\pi_k^v \pi_k^w, \pi_k^u\pi_k^w; \bm\rho) &\le c(\pi_k^v \pi_k^w, \pi_k^u\pi_k^w) \quad \forall 1\le u < v < w \le n, \ k=1,\dots, n!
\end{align}
over all triplets of rankings $(u,v,w)$ and all $n!$ permutations.
The quantities $d$ and $c$ are defined as
\begin{equation}
    d(xy,zw; \bm\rho) = \left\{ \begin{matrix}
        \rho(x,y) - \rho(z,w) & \text{if } x < y \text{ and } z < w \\
        -\rho(y,x) - \rho(z,w) & \text{if } x > y \text{ and } z < w  \\
        \rho(x,y) + \rho(w,z) & \text{if } x < y \text{ and } z > w \\
        -\rho(y,x) + \rho(w,z) & \text{if } x > y \text{ and } z > w \\
    \end{matrix}
    \right.,\ 
    c(xy,zw) = \left\{ \begin{matrix}
        0 & \text{if } x < y \text{ and } z < w \\
        -1 & \text{if } x > y \text{ and } z < w \\
        1 & \text{if } x < y \text{ and } z > w \\
        0 & \text{if } x > y \text{ and } z > w \\
    \end{matrix}
    \right.
\end{equation}
in such a way that $d(xy,zw; \bm\rho) \le c(xy,zw)$ is equivalent to the expression $\rho(x,y) \le \rho(z,w)$, but rewritten to follow the $\frac{n(n-1)}{2}$-dimensional vector representation of $\bm\rho$, namely only containing values $\rho(x_i, x_j)$ with indices $i<j$. For a given set of values $(x,y,z,w)$, $d(xy,zw; \bm\rho)$ is a linear function of $\bm\rho$, whereas $c(xy,zw)$ is a constant. Therefore, these are linear inequalities of the form $\bm A \bm\rho \le \bm b$. For each given permutation $\pi_k$, stacking all such linear inequalities together for all triplets $(u,v,w)$ gives the expression for $\bm A_k$ and $\bm b_k$.

We then consider $\bm A_{k,\ell}$, $\bm b_{k,\ell}$ for MU. Recall that for MU we need to break down the entire space of choice rates into simplexes based on different $\bm\rho^* \in R_{\mathrm{deta}}$ and different orderings of the choice magnitudes in $\bm\rho$.

Let $\Pi$ denote all permutations over the set $\{1,2,\dots,n\}$. For each of the $n!$ possible permutations $\pi\in \Pi$, we correspond it to one strict preference over the $n$ options, representing $x_{\pi(1)} \succ \dots \succ x_{\pi(n)}$. In fact, we can choose any strict preference as the ``canonical preference'' (WLOG we choose $x_1 \succ \dots \succ x_n$), and then any strict preference $\sigma$ can be converted to this canonical preference by re-mapping the indices $1,\dots,n$ to $\pi(1),\dots,\pi(n)$.

Furthermore, let $\mathcal{T}$ denote all permutations over the set $\{1,2,\dots,\frac{n(n-1)}{2}\}$.  For each of the $(\frac{n(n-1)}{2})!$ possible permutations $\tau\in \mathcal{T}$, we correspond it to one ordering over the $\frac{n(n-1)}{2}$ pairwise choice rates. representing $\left| \rho(x_{\tau_1(1)}, x_{\tau_2(1)}) - \frac{1}{2} \right| \le \dots \le \left| \rho(x_{\tau_1(\frac{n(n-1)}{2})}, x_{\tau_2(\frac{n(n-1)}{2})}) - \frac{1}{2} \right|$. Here we identify each value in $\{1,2,\dots,\frac{n(n-1)}{2}\}$ with an index pair in $\{(i,j) \,|\, 1\le i< j\le n\}$ in the same way that $\bm\rho \in \R^{\frac{n(n-1)}{2}}$ corresponds to $\rho(x_i, x_j)$, and define $(\tau_1(k),\tau_2(k))$ as the index pair corresponding to $\tau(k)$ for $k=1,\dots,\frac{n(n-1)}{2}$.

Not all $\tau \in \mathcal{T}$ satisfy the MUM constraints. In fact, under the canonical preference we chose above, we have $\rho(i,j)\ge \frac{1}{2}$ for all $1\le i < j \le n$, and hence the choice rate ordering reduces to $\rho(x_{\tau_1(1), \tau_2(1)}) \le \dots \le \rho(x_{\tau_1(\frac{n(n-1)}{2}), \tau_2(\frac{n(n-1)}{2})})$. In order for $\tau$ to satisfy MUM, it must satisfy $\tau^{-1}(i,k) \ge \min\left(\tau^{-1}(i,j), \ \tau^{-1}(j,k)\right)$ for all triplets $1\le i < j < k \le n$. Denote this set of $\tau$ satisfying MUM under the canonical preference order as $\mathcal{T}^*$, with $L = |\mathcal{T}^*|$.

Now we are ready to characterize the system of inequalities under any given strict preference $\sigma$ and choice rate ranking $\tau$. We first characterize the conditions under the canonical preference, and then generalize it to arbitrary canonical preference via index re-mapping.

For the canonical strict preference $x_1 \succ \dots \succ x_n$, we have a total of $L$ polytopes in $\R^{\frac{n(n-1)}{2}}$, each of which corresponds to one $\tau \in \mathcal{T}^*$ and is a simplex. the conditions are
\begin{align*}
    &\frac{1}{2} \le \rho(i,j) \le 1, \quad\quad \forall 1 \le i < j \le n, \\
    &\rho(\tau_1(k), \tau_2(k)) \le \rho(\tau_1(k+1), \tau_2(k+1)), \quad\quad \forall k=1,\dots, \frac{n(n-1)}{2}-1.
\end{align*}

With a general strict preference $\pi$, we essentially need to replace each $i=1,\dots,n$ with $\pi(i)$. The caveat is that this will lead to $\rho(i,j)$ terms where $i>j$, which do not directly exist in the $\bm\rho$ vector and need to converted into $1-\rho(j,i)$. In fact, let $\rho[i,j]$ denote the ``unordered'' indexing of $\rho$, such that $\rho[i,j] = \rho[j,i] = \rho(i,j)$ for $i<j$, and we can write $\rho(i,j) = \sgn(j-i) \cdot \rho[i,j] + \mathbbm{1}[i>j]$. With this notation, we can perform the index remapping with $\pi$ as
\begin{align*}
    &\frac{1}{2} \le \sgn(\pi(j)-\pi(i)) \cdot \rho[\pi(i),\pi(j)] + \mathbbm{1}[\pi(i)>\pi(j)]  \le 1, \quad\quad \forall 1 \le i < j \le n,
\end{align*}
and 
\begin{align*}
    & \sgn(\pi(\tau_2(k))-\pi(\tau_1(k))) \cdot \rho[\pi(\tau_1(k)),\pi(\tau_2(k))] + \mathbbm{1}[\pi(\tau_1(k))>\pi(\tau_2(k))] \le \\
    &\quad\quad\quad \sgn(\pi(\tau_2(k+1))-\pi(\tau_1(k+1))) \cdot \rho[\pi(\tau_1(k+1)),\pi(\tau_2(k+1))] + \mathbbm{1}[\pi(\tau_1(k+1))>\pi(\tau_2(k+1))], \\ 
    &\quad\quad\quad\quad\quad\quad \forall k=1,\dots, \frac{n(n-1)}{2}-1.
\end{align*}

For each $\pi\in \Pi$ and $\tau\in \mathcal{T}^*$, we have one set of inequality conditions above, in the standard linear inequality form $\bm A \bm\rho \le \bm b$, which corresponds to one simplex. Since $|\Pi| = n!$ and $|\mathcal{T^*}| = L$, this gives the form of all $\bm G_{k,\ell}$ and $\bm b_{k,\ell}$.

\newpage
\subsection{Proof of Lemma \ref{lem:extreme-point}}

\begin{proof}
    We first show the following lemma.

\begin{lemma}
    \label{lem:extreme-point-simple}
    Let $P$ be a convex polytope with affine dimension $n$, defined as 
    \begin{equation}
        P = \left\{ \bm x \in [0,1]^n \ \middle|\ x_i \le x_j, \ \forall (i,j) \in E  \right\},
        \label{eqn:simple-01-polytope}
    \end{equation}
    where $E \subset \{1,2,\dots,n\}^2$ is some collection of index pairs. Then any extreme point of $P$ must have all entries taking value in $\{0, 1\}$.
\end{lemma}

\begin{proof}
    We know that at any extreme point $\bm x = [x_1, \dots, x_n]^T$ of a convex polytope $P \subset \R^n$ with affine dimension $n$, there exist $n$ linearly independent facet-defining inequalities of $P$ that are attain equality. For a polytope defined as \eqref{eqn:simple-01-polytope}, the facet-defining inequalities must take the following two forms. Type-A constraints are those of the form $x_i \ge 0$ or $x_i \le 1$. Type-B constraints are those of the form $x_i \le x_j$. We now consider the $n$ inequalities constraints that attain equality at extreme point $\bm x$.

    A type-B constraint that attains equality implies $x_i = x_j$. Thus, all the type-B constraints that attain equality collectively partitions the index set $\{1,2,\dots,n\}$ into some equivalence classes, which we denote as $I_1, \dots, I_K$. In other words, whenever two distinct indices $i, j$ belong to the same equivalence class, it means they are connected by some type-B constraints and their values must be equal. Note that none of the constraints contain indices that belong to different equivalence class. This allows us to refer to every constraint as ``corresponding to'' an equivalence class.

    For each equivalence class $I_k$, there can at most be $|I_k|$ linearly independent constraints corresponding to it, because this is the total number of variables concerning any such constraints. Also, this maximum number of $|I_k|$ linearly independent constraints is only achievable if there exists at least one type-A constraint, because if instead all constraints were type-B, the maximum possible rank of the constraints would be $|I_k|-1$, since there is a degree of freedom in the $x_i$'s values that cannot be pinned down by any type-B constraint. 

    In total across all equivalence class, we can have a maximum of $\sum_{k=1}^K |I_k| = n$ linearly independent constraints, with equality obtainable only when every equivalence class has at least one corresponding type-A constraint. This means that the indices in every equivalence class is pinned down to be equal to either 0 or 1, by the corresponding type-A constraint. Therefore, all entries of $\bm x$ must take value in $\{0, 1\}$.
\end{proof}

    Now we prove the result of Lemma \ref{lem:extreme-point} regarding the collective WU/MU/SS polytopes. First consider the specific rational deterministic choice vector $\bm\rho^* = [0,0,\dots,0]^T \in \R^\frac{n(n-1)}{2}$, for which we have 
    \begin{equation}
        \mathrm{Cube}(\bm\rho^*) = [0, \frac{1}{2}]^\frac{n(n-1)}{2}.
    \end{equation}
    As discussed previously, we can partition this small hypercube $\mathrm{Cube}(\bm\rho^*)$ into a total of $\left(\frac{n(n-1)}{2}\right)!$ simplexes, each corresponding to one ordering of the $\frac{n(n-1)}{2}$ entries of $\bm\rho$. Each such simplex is a polytope that satisfies the condition of Lemma \ref{lem:extreme-point-simple} (after scaling by 2), and thus it follows that every extreme point of these simplexes must have all entries taking value in $\{0, \frac{1}{2}\}$.

    For any other rational deterministic choice vector $\bm\rho^*$, due to symmetry, its corresponding $\mathrm{Cube}(\bm\rho^*)$ and the simplexes can be obtained by taking that of $[0,0,\dots,0]^T$ and re-mapping the indices $\{1,\dots,n\}$ according to the permutation that defines $\bm\rho^*$. Some entries of $\rho$ will get reversed (from $\rho(x_i, x_j)$ to $\rho(x_j, x_i) = 1 - \rho(x_i, x_j)$) in this process. Therefore, the resulting extreme points of all simplexes from $\mathrm{Cube}(\bm\rho^*)$ will have all entries taking value in $\{0, \frac{1}{2}, 1\}$.

    Finally, because these simplexes constitute the smallest unit of the space that individual WU/MU/SS has implications over, we know that the rationalizable set for individual WU/MU/SS must be a finite union of such simplexes. Hence, the extreme points of the polytope (of individual rationalizability) will only take value $\{0, \frac{1}{2}, 1\}$.

\end{proof}

\newpage
\subsection{Proof of Theorem \ref{thm:collective}}
\begin{proof}
For all three models, the main part of the proof for (i) $\iff$ (ii) has already been stated the main text, using the key fact that for any finite collection of sets $S_1, S_2, \dots, S_k \subseteq \mathbb{R}^d$, 
\begin{equation}
\conv\left(\bigcup_{i=1}^k \mathrm{conv}(S_i)\right)=\;
\conv\left(\bigcup_{i=1}^k S_i\right).
\end{equation} 
Furthermore, recall that the main text has provided characterization of the closure of the rationalizable set for individual WU/MU/SS. Here we again aim to characterize the closure of the rationalizable set of collective WU/MU/SS, noting that the collective characterization is the convex hull of the \textit{original} (i.e. non-closure) individual WU/MU/SS. Hence we need to use the fact that for a finite number of convex sets, the closure of convex hull is equal to the convex hull of closure. In other words, for convex sets $C_1, \dots, C_K$,
\begin{align}
    \text{cl}\left( \conv\left( \bigcup_{k=1}^{K} C_k \right) \right) =  \conv\left( \bigcup_{k=1}^{K} \text{cl} \left( C_k \right) \right). 
\end{align}
This implies that taking the convex hull of the closure of individual rationalizability set gives the closure of the collective rationalizability set. 

It remains to prove (i) $\iff$ (iii). The proof for all three models follow the same argument, so we will illustrate WU as an example.

As discussed in the main text, the rationalizable set for individual WU can be represented as the union of $n!$ small hypercubes. This union can be modeled as $\bm\rho$ satisfying the following feasibility problem:
\begin{equation}
    \begin{aligned}
         \sum_{k=1}^{n!} \bm\rho_k &= \bm\rho \\
    -\bm\rho_k &\le x_k \left(\frac{1}{2} \bm 1 - \bm\rho^*_k \right)\quad k=1,\dots, n! \\
    \bm\rho_k &\le x_k \left(\frac{1}{2} \bm 1 + \bm\rho^*_k \right) \quad k=1,\dots, n! \\
    \bm 0 &\le \bm \rho_k \le x_k \bm 1 \quad \quad \quad \ 
    k=1,\dots, n! \\
    \sum_{k=1}^{n!} x_k &= 1 \\
    \bm x&\in\{0,1\}^{n!}.
    \end{aligned}
    \label{eqn_integer}
\end{equation}
Here we define $n!$ auxiliary indicator variables $x_k \in \{0,1\}$, indicating whether $\bm\rho$ actually belongs to the $k$-th hypercube. We also define $n!$ copies of the $\bm\rho$ vector as $\bm\rho_k$, each lying in one hypercube and corresponding to one indicator $x_k$.

With the help of techniques from integer programming models and modeling disjunction \citep{balas1971intersection}, the convex hull of solutions to the linear system (\ref{eqn_integer}) is simply obtained by dropping the integrality restriction. Hence, we have $\bm\rho\in\conv\left( \bigcup_{k=1}^{n!} \text{cl} \left(\text{Cube}(\bm\rho^*_k) \right) \right)$ iff
\begin{equation}
    \begin{aligned}
        \sum_{k=1}^{n!} \bm\rho_k &= \bm\rho \\
    -\bm\rho_k &\le x_k \left(\frac{1}{2} \bm 1 - \bm\rho^*_k \right)\quad k=1,\dots, n! \\
    \bm\rho_k &\le x_k \left(\frac{1}{2} \bm 1 + \bm\rho^*_k \right) \quad k=1,\dots, n! \\
    \bm 0 &\le \bm \rho_k \le x_k \bm 1 \quad \quad \quad \ 
    k=1,\dots, n! \\
    \sum_{k=1}^{n!} x_k &= 1 \\
    \bm x&\in[0,1]^{n!}.
    \end{aligned}
\end{equation}
This is exactly the form in Theorem \ref{thm:collective} (iii).

The derivation for MU and SS are entirely analogous, with the only difference being the different inequality constraints that characterize individual rationalizability. Therefore we avoid repetition here.

\end{proof}

\newpage
\subsection{Extra Result on Theorem \ref{thm:collective}}
\begin{lemma}\label{lem: extra-collective}
    For a choice vector $\bm\rho\in[0,1]^m$, the following are equivalent:
\begin{enumerate}
\item $\bm\rho$ is rationalizable by collective WU.
\item The following inequality holds for any $\bm\gamma\in \R^{\frac{n(n-1)}{2}}$:
\begin{equation}
    (\bm\rho - \frac{1}{2}\bm 1)^T \bm\gamma \ge - \frac{1}{2} \sum_{j\,:\, \sgn(\gamma_j) = \sgn(\rho^*_k - \frac{1}{2})} |\gamma_j| \vspace{-.2cm}
    \label{eqn:dualcond}
\end{equation}
where $k\in \{1,\dots,n!\}$ is such that
\vspace{-.2cm}
\begin{equation}
    k =\arg\min_{k'} \ \sum_{j \,:\, \sgn(\rho^*_{k'j} - \tfrac{1}{2}) \neq \sgn(\gamma_j)} |\gamma_j|.
\end{equation}
\end{enumerate}
\end{lemma}

\begin{proof}

This Lemma provides an additional characterization for Collective WU, aimed at deriving its exact facet-defining inequalities. This is currently a partial result in that endeavor.

Writing all inequality constraints from (c) into a full matrix form, we want to check whether there exists a vector $\bm v$ such that
\begin{align*}
    \bm A \bm v &\le \bm b \\
    \bm A' \bm v &= \bm b'
\end{align*}
where
\begin{equation}
    \bm v = \begin{bmatrix}
        \bm \rho_1 \\
        \vdots \\
        \bm \rho_{n!} \\
        x_1 \\
        \vdots \\
        x_{n!}
    \end{bmatrix}, \quad
    \bm A' = \left[\begin{array}{ccc|ccc}
        \bm I_{\frac{n(n-1)}{2}} & \cdots & \bm I_{\frac{n(n-1)}{2}} & & & \\
        \hline
        & & & & \bm 1_{n!}^T & \\
    \end{array}\right], \quad
    \bm b' = \begin{bmatrix}
        \bm\rho \\
        1
    \end{bmatrix}
\end{equation}
\begin{equation}
    \bm A = \left[\begin{array}{ccc|ccc}
        & & & -\frac{1}{2}\bm 1_{\frac{n(n-1)}{2}} - \bm\rho^*_1 & & \\
        & \bm I_{n!\frac{n(n-1)}{2}} & &  & \ddots & \\
        & & & & & -\frac{1}{2}\bm 1_{\frac{n(n-1)}{2}} - \bm\rho^*_{n!} \\
        \hline
        & & & -\frac{1}{2}\bm 1_{\frac{n(n-1)}{2}} + \bm\rho^*_1 & & \\
        & -\bm I_{n!\frac{n(n-1)}{2}} & &  & \ddots & \\
        & & & & & -\frac{1}{2}\bm 1_{\frac{n(n-1)}{2}} + \bm\rho^*_{n!} \\
        \hline
        & & & & & \\
        & -\bm I_{n!\frac{n(n-1)}{2}} & & & & \\
        & & & & & \\
        \hline
        & & & -\bm 1_{\frac{n(n-1)}{2}} & & \\
        & \bm I_{n!\frac{n(n-1)}{2}} & &  & \ddots & \\
        & & & & & -\bm 1_{\frac{n(n-1)}{2}} \\
        \hline
        & & & & -\bm I_{n!} & \\
        \hline
        & & & & \bm I_{n!} & \\
    \end{array}\right], \quad 
    \bm b = \left[\begin{array}{c}
        \bm 0 \\
        \hline
        \bm 0 \\
        \hline
        \bm 0 \\
        \hline
        \bm 0 \\
        \hline
        \bm 0 \\
        \hline
        \bm 1_{n!}
    \end{array}\right]
\end{equation}

From Farka's lemma, this is equivalent to there not existing $\bm\lambda \ge \bm 0 \in \R_+^{n!(2n(n-1)+2)}$ and $\bm\lambda' \in \R^{\frac{n(n-1)}{2}+1}$ that satisfies
\begin{align}
    \bm A^T \bm \lambda + \bm A'^T \bm\lambda' &= \bm 0 \label{eqn:dual1} \\
    \bm b^T \bm\lambda + \bm b'^T \bm\lambda' &< 0. \label{eqn:dual2}
\end{align}

\vspace{1cm}

Let 
\begin{equation}
    \bm\lambda = \begin{bmatrix}
        \bm\lambda_1 \\
        \bm\lambda_2 \\
        \bm\lambda_3 \\
        \bm\lambda_4 \\
        \bm\lambda_5 \\
        \bm\lambda_6
    \end{bmatrix}, \quad
    \bm\lambda' = \begin{bmatrix}
        \bm\lambda'_1 \\
        \lambda'_2
    \end{bmatrix}
\end{equation}
where $\bm\lambda_1, \bm\lambda_2, \bm\lambda_3, \bm\lambda_4 \in \R^{n!\frac{n(n-1)}{2}}$, $\bm\lambda_5, \bm\lambda_6 \in \R^{n!}$, $\bm\lambda'_1 \in \R^{\frac{n(n-1)}{2}}$, $\lambda'_2 \in \R$. Further, for $i=1,2,3,4$ write $\bm\lambda_i = [\bm\lambda_{ik}]_{k=1}^{n!}$ with $\bm\lambda_{ik} \in \R^{\frac{n(n-1)}{2}}$, and for $i=5,6$ write $\bm\lambda_i = [\lambda_{ik}]_{k=1}^{n!}$ with $\lambda_{ik} \in \R$.

The conditions \eqref{eqn:dual1} \eqref{eqn:dual2} can be written as
\begin{align}
    \bm\lambda_{1k} - \bm\lambda_{2k} - \bm\lambda_{3k} + \bm\lambda_{4k} + \bm\lambda'_1 = \bm 0, \ \ \forall 1\le k\le n! \label{eqn:dualcond-v1-1}\\
    \left(-\frac{1}{2}\bm 1 - \bm\rho^*_k\right)^T \bm\lambda_{1k} + \left(-\frac{1}{2}\bm 1 + \bm\rho^*_k\right)^T \bm\lambda_{2k} - \bm 1^T \bm\lambda_{4k} - \lambda_{5k} + \lambda_{6k} + \lambda'_2 = 0, \ \ \forall 1\le k\le n! \label{eqn:dualcond-v1-2} \\
    \bm 1^T \bm\lambda_6 + \bm\rho^T \bm\lambda'_1 + \lambda'_2 < 0 \label{eqn:dualcond-v1-3}
\end{align}

We now simplify the system above by variable elimination.

$\bm\lambda_{3k}$ only appears in \eqref{eqn:dualcond-v1-1}, and additionally we require $\bm\lambda_{3k} \ge \bm 0$, hence we can eliminate $\bm\lambda_{3k}$ and \eqref{eqn:dualcond-v1-1} is equivalent to 
\begin{equation}
    \bm\lambda_{1k} - \bm\lambda_{2k} + \bm\lambda_{4k} + \bm\lambda'_1 \ge \bm 0, \ \ \forall 1\le k\le n!
    \label{eqn:dualcond-v1-4}
\end{equation}

$\lambda_{5k}$ only appears in \eqref{eqn:dualcond-v1-2}, and additionally we require $\lambda_{5k} \ge 0$, hence we can eliminate $\lambda_{5k}$ and \eqref{eqn:dualcond-v1-2} is equivalent to 
\begin{equation}
    \left(-\frac{1}{2}\bm 1 - \bm\rho^*_k\right)^T \bm\lambda_{1k} + \left(-\frac{1}{2}\bm 1 + \bm\rho^*_k\right)^T \bm\lambda_{2k} - \bm 1^T \bm\lambda_{4k}  + \lambda_{6k} + \lambda'_2 \ge 0, \ \ \forall 1\le k\le n! \label{eqn:dualcond-v1-5}
\end{equation}

By comparing $\lambda'_2$ in conditions \eqref{eqn:dualcond-v1-3} and \eqref{eqn:dualcond-v1-5}, we can eliminate $\lambda'_2$ and combine these two conditions into the equivalent form
\begin{equation}
    \left(-\frac{1}{2}\bm 1 - \bm\rho^*_k\right)^T \bm\lambda_{1k} + \left(-\frac{1}{2}\bm 1 + \bm\rho^*_k\right)^T \bm\lambda_{2k} - \bm 1^T \bm\lambda_{4k}  + \lambda_{6k} > \bm 1^T \bm\lambda_6 + \bm\rho^T \bm\lambda'_1, \ \ \forall 1\le k\le n!
    \label{eqn:dualcond-v1-6}
\end{equation}

Now the conditions have been transformed to \eqref{eqn:dualcond-v1-4} \eqref{eqn:dualcond-v1-6}. $\lambda_{6k}$ only appears in \eqref{eqn:dualcond-v1-6}, and additionally we require $\lambda_{6k} \ge 0$, hence we can eliminate $\lambda_{6k}$ and \eqref{eqn:dualcond-v1-6} is equivalent to 
\begin{equation}
    \left(-\frac{1}{2}\bm 1 - \bm\rho^*_k\right)^T \bm\lambda_{1k} + \left(-\frac{1}{2}\bm 1 + \bm\rho^*_k\right)^T \bm\lambda_{2k} - \bm 1^T \bm\lambda_{4k} > \bm\rho^T \bm\lambda'_1, \ \ \forall 1\le k\le n!
    \label{eqn:dualcond-v1-7}
\end{equation}

By comparing $\bm\lambda_{4k}$ in \eqref{eqn:dualcond-v1-4} and \eqref{eqn:dualcond-v1-7}, we can eliminate $\bm\lambda_{4k}$ by setting $\bm\lambda_{4k} = \max\!\big(\bm\lambda_{2k}-\bm\lambda_{1k} - \bm\lambda'_1, \ \bm 0\big)$ in \eqref{eqn:dualcond-v1-6}. Hence the system of conditions \eqref{eqn:dualcond-v1-1}-\eqref{eqn:dualcond-v1-3} is equivalent to the following one inequality constraint: 
\begin{equation}
    \left(-\frac{1}{2}\bm 1 - \bm\rho^*_k\right)^T \bm\lambda_{1k} + \left(-\frac{1}{2}\bm 1 + \bm\rho^*_k\right)^T \bm\lambda_{2k} - \bm 1^T \max\!\big(\bm\lambda_{2k}-\bm\lambda_{1k} - \bm\lambda'_1, \ \bm 0\big) > \bm\rho^T \bm\lambda'_1, \ \ \forall 1\le k\le n!
    \label{eqn:dualcond-v2}
\end{equation}

Define a change of variables $\bm\alpha_{k} = \bm\lambda_{1k} + \bm\lambda_{2k}$ and $\bm\beta_{k} = \bm\lambda_{2k} - \bm\lambda_{1k}$, $\forall 1\le k \le n!$. The constraints $\bm\lambda_{1k}, \bm\lambda_{2k} \ge \bm 0$ is equivalent to $\bm\alpha_k \ge |\bm\beta_k|$. Plugging this into \eqref{eqn:dualcond-v2}, and also rename $\bm\lambda'_1$ as $\bm\gamma$, we get
\begin{equation}
    -\frac{1}{2}\bm 1^T \bm\alpha_k + \bm\rho^*_k \!\ ^T \bm\beta_k - \bm 1^T \max\!\big(\bm\beta_k - \bm\gamma, \ \bm 0\big) > \bm\rho^T \bm\gamma, \ \ \forall 1\le k\le n!
\end{equation}
We can eliminate $\bm\alpha_k$ by taking $\bm\alpha_k = |\bm\beta_k|$, and the condition is equivalent to
\begin{equation}
    -\frac{1}{2}\bm 1^T |\bm\beta_k| + \bm\rho^*_k \!\ ^T \bm\beta_k - \bm 1^T \max\!\big(\bm\beta_k - \bm\gamma, \ \bm 0\big) > \bm\rho^T \bm\gamma  \ \ \forall 1\le k\le n! 
    \label{eqn:dualcond-v3}
\end{equation}

Let $\beta_{kj}$ be the $j$-th entry of $\bm\beta_k$, $1\le j \le \frac{n(n-1)}{2}$. The LHS of \eqref{eqn:dualcond-v3} is a piecewise linear function of each individual $\beta_{kj}$, we can find the $\beta_{kj}$ values that maximizes LHS as 
\begin{equation}
    \beta_{kj} = \left\{ \begin{matrix}
        \max(\gamma_j, 0) & \text{ if } \rho^*_{kj} = 1 \\
        \min(\gamma_j, 0) & \text{ if } \rho^*_{kj} = 0
    \end{matrix} \right.
\end{equation}
This further simplifies to 
\begin{equation}
    \beta_{kj} = \left\{ \begin{matrix}
        \gamma_j & \text{ if } \mathrm{sgn}(\rho^*_{kj}-\tfrac{1}{2}) = \mathrm{sgn}(\gamma_j) \\
        0 & \text{otherwise}
    \end{matrix} \right.
\end{equation}
We can eliminate $\bm\beta_k$ by plugging this maximizer into \eqref{eqn:dualcond-v3}, and get
\begin{equation}
    -\frac{1}{2} \sum_{J_k^{++}} \gamma_j + \frac{1}{2} \sum_{J_k^{--}} \gamma_j + \sum_{J_k^{-+}} \gamma_j + \sum_{J_k^{++} \cup J_k^{--}} \rho^*_{kj} \gamma_j  > \sum_j \rho_j \gamma_j, \ \ \forall 1\le k\le n!  \label{eqn:dualcond-v4}
\end{equation}
where we define index subsets
\begin{equation}
    J_k^{s_1s_2} = \left\{ 1\le j \le \frac{n(n-1)}{2} \ \middle|\ \mathrm{sgn}(\gamma_j) = s_1, \ \mathrm{sgn}(\rho^*_{kj}-\tfrac{1}{2}) = s_2  \right\}, \quad s_1, s_2 \in \{+,-\}.
\end{equation}
Because $\rho^*_{kj}\in \{0,1\}$ are known given the index subset $J_k^{s_1s_2}$ (determined by the sign $s_2$), \eqref{eqn:dualcond-v4} can further be organized into
\begin{equation}
    \sum_{j\in J_k^{++}} (\tfrac{1}{2} - \rho_j) \gamma_j + \sum_{j\in J_k^{--}} (\tfrac{1}{2} - \rho_j) \gamma_j + \sum_{j\in J_k^{+-}} (-\rho_j) \gamma_j + \sum_{j\in J_k^{-+}} (1-\rho_j) \gamma_j > 0, \ \ \forall 1\le k\le n!
    \label{eqn:dualcond-v5}
\end{equation}

Up to now, we have shown that a vector $\bm\rho \in \R^{\frac{n(n-1)}{2}}$ lies in $\overline{\mathrm{conv}(S)}$ (i.e. is population-WST) iff for any $\bm\gamma\in \R^{\frac{n(n-1)}{2}}$, there exists $1\le k\le n!$ such that 
\begin{equation}
    \sum_{j\in J_k^{++}} \tfrac{1}{2}\gamma_j + \sum_{j\in J_k^{--}} \tfrac{1}{2}\gamma_j + \sum_{j\in J_k^{-+}}\gamma_j \le \sum_{j}\rho_j\gamma_j.
    \label{eqn:dualcond-v5-2}
\end{equation}

In the following, we characterize which $k$ among the $n!$ possible choices minimizes LHS of $\eqref{eqn:dualcond-v5-2}$, which is the only $k$ we will need to consider.

The admissible values of $\bm\rho^*_k$ is a strict subset of $\{0,1\}^{\frac{n(n-1)}{2}}$, but for now imagine a relaxation where we allow any $\bm\rho^* \in \{0,1\}^{\frac{n(n-1)}{2}}$. In this case, we can determine the optimal $\rho^*_j$ for each index $j$. For a given $j$, if $\gamma_j > 0$ then LHS of \eqref{eqn:dualcond-v5-2} is minimized when $\rho^*_j = 0$; if $\gamma_j < 0$ then LHS of \eqref{eqn:dualcond-v5-2} is minimized when $\gamma_j > 0$; if $\gamma_j = 0$ then $\rho^*_j$ does not matter. In summary, the optimal $\bm\rho^* \in \{0,1\}^{\frac{n(n-1)}{2}}$ that minimizes LHS of \eqref{eqn:dualcond-v5-2} should satisfy $\sgn(\rho^*_j - \tfrac{1}{2}) = - \sgn(\gamma_j)$ whenever $\gamma_j \neq 0$.

Now we return to the actual scenario where $\bm\rho^*$ is only allowed to take a more restricted set of values $P^* = \{\bm\rho^*_k\}_{k=1}^{n!} \subsetneq \{0,1\}^{\frac{n(n-1)}{2}}$. For a given $\bm\rho^*_k$, for each $j$ where $\sgn(\bm\rho^*_{kj} - \tfrac{1}{2})$ and $-\sgn(\gamma_j)$ differ, LHS of \eqref{eqn:dualcond-v5-2} increases by $\tfrac{1}{2} |\gamma_j|$. The optimal $k$ is one that minimizes this total increase, which is 
\begin{equation}
    k^* =\arg\min_k \ \sum_{j \,:\, \sgn(\rho^*_{kj} - \tfrac{1}{2}) = \sgn(\gamma_j)} |\gamma_j|.
\end{equation}
We know that $\bm\rho^*\in P^* \iff \bm 1 - \bm\rho^* \in P^*$, and can let $\widehat{k}$ be such that $\bm\rho^*_{\widehat{k}} = \bm 1 - \bm\rho^*_{k^*}$, namely the ``reverse ordering'' of $k^*$. Then we have 
\begin{equation}
    \widehat{k} =\arg\min_k \ \sum_{j \,:\, \sgn(\rho^*_{kj} - \tfrac{1}{2}) \neq \sgn(\gamma_j)} |\gamma_j|.
    \label{eqn:dualcond-nearest-k}
\end{equation}
This has a very intuitive interpretation. For a given vector $\bm\gamma$, imagine making minimal changes (in the $L_1$ sense) to its elements such that the sign pattern aligns with some WST sign pattern. $\bm\rho^*_{\widehat k}$ is the WST ordering ``nearest'' to $\bm\gamma$. Also notably, $\widehat k$ depends only on $\bm\gamma$ and is independent of $\bm\rho$.

Using the fact that $k^*$ and $\widehat k$ give exactly the reverse signs, we can rewrite the condition as follow. A choice rate $\bm\rho \in \R^{\frac{n(n-1)}{2}}$ is collective WUM-rationalizable iff for any $\bm\gamma\in \R^{\frac{n(n-1)}{2}}$, we have
    \begin{equation}
        (\bm\rho - \frac{1}{2}\bm 1)^T \bm\gamma \ge - \frac{1}{2} \sum_{j\,:\, \sgn(\gamma_j) = \sgn(\rho^*_k - \frac{1}{2})} |\gamma_j|,
        \label{eqn:dualcond}
    \end{equation}
    where $k\in \{1,\dots,n!\}$ is such that
    \begin{align}
        k =\arg\min_{k'} \ \sum_{j \,:\, \sgn(\rho^*_{k'j} - \tfrac{1}{2}) \neq \sgn(\gamma_j)} |\gamma_j|.
        \label{eqn:dualcond-v6}
    \end{align}
\end{proof}

\newpage
\subsection{Additional Results on General Non-Binary Menus}\label{sec:appdx-multinomial}

So far, I have focused on stochastic choice over binary menus in $Z^2$ over a choice set $Z$. Meanwhile, some of the random choice models we discussed can also be generalized to non-binary menus, namely the family of all subsets $2^Z$. Here I take Simple Scalability as an example, and discuss the application of collective rationalizability to Simple Scalability over general menus.

A natural generalization of Simple Scalability to general menus is given in \citep{tversky1972choice}. Let $\rho(x,A)$ denote the probability of choosing $x \in A$ in a menu $A \subseteq Z$. Let $u$ and $F_a,\ 2\le a\le n$, be real-valued functions, then a choice rule satisfies Simple Scalability if for any choice option in any choice subset $A = \{x, y,\dots,z\} \subseteq Z$, we have
\begin{equation}
    \rho(x,A) = F_a(u(x), u(y), \dots, u(z)),
\end{equation}
with $F_a$ symmetric in its last $a-1$ arguments, and satisfying: $F_a$ is strictly increasing in the first argument and strictly decreasing in the remaining $a-1$ arguments provided $\rho(x,A)\neq 0,1$. If $\rho(x,A)$ is 0 or 1, $F_a$ is non-decreasing in the first argument and non-increasing in the others.

Generalized to multiple choices, Simple Scalability asserts that the alternatives can be scaled so that choice probability is represented as a monotonic function of the respective scale values.
Simple Scalability with multiple choices is equivalent to the \textit{order-independence} property \citep{tversky1972choice}, which requires that for any two menus $A, B \subseteq Z$, $\forall x, y \in B\backslash A$, $\forall z \in A$, we have 
\begin{equation}
    \rho(x, B) \ge \rho(y, B) \iff \rho(z, A\cup \{x\}) \le \rho(z, A\cup \{y\}),
\end{equation}
provided the choice probabilities on the two sides of either inequality are not 0 or 1. This characterization provides a directly testable condition for Simple Scalability over general menus.

Thus, following the procedure developed in this paper, one can obtain the collective rationalizability characterization by taking the convex hull of the individual rationalizable set. Below, I provide a concrete example for the case of three options $Z = \{x,y,z\}$.

With the inclusion of general non-binary menus, a choice rule $\rho$ now corresponds to a 5-dimensional choice vector:
\begin{equation}
    \bm\rho = [\rho(x,\{x,y\}),\ \rho(x,\{x,z\}),\ \rho(y,\{y,z\}),\ \rho(x,\{x,y,z\}),\ \rho(y,\{x,y,z\})]^T.
\end{equation}
Similar to the characterization of individual SS with binary menus, here the rationalizable set can again be written as the union of $3!=6$ convex polytopes, each of which is 5-dimensional. For example, for the convex polytope corresponding to the preference ordering $x \succ y \succ z$, in addition to the constraints $\rho(x, \{x,z\}) \ge \rho(x, \{x,y\})$ and $\rho(x, \{x,z\}) \ge \rho(y, \{y,z\})$ from strong stochastic transitivity, we furthermore have constraints $\rho(x, \{x,y,z\}) \ge \rho(y, \{x,y,z\})$ and $\rho(y, \{x,y,z\}) \ge 1 - \rho(x, \{x,y,z\}) - \rho(y, \{x,y,z\})$. Finally, applying the convex hull operation on the half-space representation, one can readily obtain the characterization for collective SS over general menus.

Interestingly, aggregating choices from heterogeneous individuals under representative agent assumption can generate \textit{choice reversal}, even when each individual satisfies Simple Scalability.
\begin{example}
    Consider two individuals with choice probabilities
    \begin{align*}
        \bm\rho_1 &= [\rho_1(x,\{x,y\}),\ \rho_1(x,\{x,z\}),\ \rho_1(y,\{y,z\}),\ \rho_1(x,\{x,y,z\}),\ \rho_1(y,\{x,y,z\})]^T = [0.2,0.6,0.9,0.3,0.6]^T,\\
        \bm\rho_2 &= [\rho_2(x,\{x,y\}),\ \rho_2(x,\{x,z\}),\ \rho_2(y,\{y,z\}),\ \rho_2(x,\{x,y,z\}),\ \rho_2(y,\{x,y,z\})]^T = [0.7,0.6,0.4,0.8,0.05]^T.
    \end{align*}
    It is straightforward to verify that both individuals satisfy Simple Scalability, with Subject 1 having the preference ordering $y\succ x\succ z$, while Subject 2 $x\succ z\succ y$. However, aggregating their choice probabilities with equal weights yields
    \begin{align*}
        \rho(y,\{x,y\}) >0.5,\ \rho(x,\{x,y,z\}) >0.5.
    \end{align*}
    This is a choice reversal that $y$ is chosen more often than $x$ in the binary menu $\{x,y\}$, but $x$ is chosen more than 50\% in the menu $\{x,y,z\}$.
\end{example}
Importantly, this choice reversal violates regularity: adding alternative $z$ increases $x$'s choice probability from 0.45 to 0.55. Such regularity violations cannot arise from heterogeneity in random utility models where individuals choose deterministically, but
can emerge from aggregating stochastic choices that individually satisfy Simple Scalability with heterogeneous tastes.

\newpage
\subsection{Proof of Proposition \ref{prop:loop-rank}}
\begin{proof}
The loop choice can be rationalized by WU iff we can find the comparison difficulty $\{c_j\}_{j=1}^m$ and a function $F$ such that
\begin{align}
    \sum_{j\in J^+} c_{j} F^{-1}(\rho_{j}) = -\sum_{j'\in J^-} c_{j'} F^{-1}(\rho_{j'}).
\label{eqn:loop-balance}
\end{align}
Here $F^{-1}$ can be any strictly increasing function with $F^{-1}(\tfrac{1}{2}) = 0$ and $|F^{-1}(x)| \ge |F^{-1}(y)| \iff |x-\tfrac{1}{2}| \ge |y-\tfrac{1}{2}|$. Hence \eqref{eqn:loop-balance} can be rewritten as
\begin{equation}
    \sum_{j\in J^+} c_{j} G(\mu_{j}) = \sum_{j'\in J^-} c_{j'} G(\mu_{j'}),
    \label{eqn:loop-balance-2}
\end{equation}
where $G$ can be any strictly increasing function with $G(0)=0$.

\vspace{.5cm}
We first prove part 1 of the result.
Because $J^+$ and $J^-$ do not ranking-dominate each other, it follows that either strict inequality in \eqref{eqn:ranking-domination-def} in both directions holds for some $(s,t)$, or that equality holds for all $(s,t)$. That latter immediately implies that $J^+$ and $J^-$ have the identical set of values, and hence an arbitrary function $G$ with arbitrary constant $\{c_j\}_{j=1}^m$ sequence will satisfy \eqref{eqn:loop-balance-2}. In the following, we consider the other case where strict inequality in both directions must hold for some $(s,t)$.

Pick an arbitrary function $G$ (strictly increasing with $G(0)=0)$ and an arbitrary sequence of values of $\{c_j\}_{j=1}^m$ that satisfies the ordinal ranking $\pi$. If this arbitrary choice does not satisfy equality in \eqref{eqn:loop-balance-2}, suppose we have $\sum_{j\in J^+} c_{j} G(\mu_{j}) > \sum_{j'\in J^-} c_{j'} G(\mu_{j'})$. Our assumptions imply that there exist some $s,t\in\{1,\dots,m\}$ such that
\begin{align}
    \#\{j\in J^+\mid\sigma(j)\le s, \pi(j)\le t\}\  <\  \#\{j\in J^-\mid\sigma(j)\le s, \pi(j)\le t\}.
    \label{eqn:rank-domination-implication}
\end{align}

Consider the following operation: for a given a constant $M>0$, increase all values $c_j$ with $\pi(j) \le t$ by $M$, and increase all function values $G(\mu_j)$ with $\sigma(j) \le s$ by $M$. The former preserves the ranking $\pi$ of $\{c_j\}_{j=1}^m$, and the latter can be done in a way that preserves the monotonicity and continuity of $G$ because $\sigma$ is the ranking function of $\{\mu_j\}$. Define
\begin{align*}
    \delta(M) &= \sum_{j\in J^+} \big(c_j + M\mathbbm{1}[\pi(j)\le t]\big) \big(G(\mu_j) + M\mathbbm{1}[\sigma(j)\le s]) \ - \\
    &\quad\quad \sum_{j\in J^-} \big(c_j + M\mathbbm{1}[\pi(j)\le t]\big) \big(G(\mu_j) + M\mathbbm{1}[\sigma(j)\le s])
\end{align*}
as the difference between the two sides of \eqref{eqn:loop-balance-2} after the operation. We have $\delta(0) > 0$ by our assumption. As $M$ becomes sufficiently large, $\delta(M)$ is dominated by the quadratic term of $M$, and \eqref{eqn:rank-domination-implication} implies that the coefficient of the quadratic term is negative. Hence for sufficiently large $M$ we have $\delta(M) < 0$, and by continuity there must exist some $M$ such that $\delta(M) = 0$. Therefore the resulting modified $\{c_j\}_{j=1}^m$ and $G$ satisfied the desired equation \eqref{eqn:loop-balance-2}, and hence the data can be rationalized by RUM with the ranking of comparison difficulty $\pi$.

The argument is identical if instead the initial arbitrary choice $\{c_j\}_{j=1}^m$ gives the opposite sign.

\vspace{.5cm}
We then prove part 2 of the result. 
Define
\begin{equation}
    \delta_j = \mathbbm{1}[j \in J^+] - \mathbbm{1}[j \in J^-], \quad j=1,\dots,m.
\end{equation} and
\begin{equation} 
    \Delta = \sum_{j=1}^m c_j G(\mu_j) \delta_j.
    \label{eqn:balance-rewrite-proof2}
\end{equation}
It suffices to prove that if $J^+$ ranking-dominates $J^-$, and no choice rates equal $\frac{1}{2}$, then we have $\Delta > 0$. The other case where $J^-$ ranking-dominates $J^+$ is symmetric.

We sort the values $c_j$ and $G(\mu_j)$ by groups. Denote the distinct values of $c_j$ as $x_1 > x_2 > \dots > x_K > 0$, and the distinct values of $G(\mu_j)$ as $y_1 > y_2 > \dots > y_L > 0$. Here $x_K > 0$ because complexity $c_j$ must be positive, and $y_L > 0$ due to our assumption that $\mu_j\neq 0, \forall j$. For a given $j \in \{1,\dots,m\}$, define the corresponding indices $k_j\in \{1,\dots,K\}$ such that $x_{k_j} = c_j$, and $l_j\in \{1,\dots,L\}$ such that $y_{l_j} = G(\mu_j)$. Then, the sum $\Delta$ can be rewritten as 
\begin{equation}
    \Delta = \sum_{j=1}^m x_{k_j} y_{l_j} \delta_j.
\end{equation}

The key step is using summation by parts twice. First, based on the sorted values $x_k$, rewrite the sum $\Delta$  as
\begin{equation}
    \Delta = \sum_{k=1}^K x_k  \left( \sum_{j=1}^m  y_{l_j} \delta_j \mathbbm{1}[k_j = k] \right) = \sum_{k=1}^K x_k z_k,
\end{equation}
where we define $z_k = \sum_{j}  y_{l_j} \delta_j \mathbbm{1}[k_j = k]$. Let $Z_k = \sum_{i=1}^k z_i$ be its cumulative sum with $Z_0=0$. Applying summation by parts to the sum $\sum x_k z_k$:
\begin{align}
\Delta &= \sum_{k=1}^K x_k (Z_k - Z_{k-1}) \nonumber \\
&= x_K Z_K + \sum_{k=1}^{K-1} (x_k - x_{k+1}) Z_k. \label{eq:sum_of_W}
\end{align}

Furthermore, the term $Z_k$ can be rewritten as
\begin{equation}
Z_k = \sum_{i=1}^k \sum_{j=1}^m  y_{l_j} \delta_j \mathbbm{1}[k_j = i] = \sum_{j=1}^m  y_{l_j} \delta_j \mathbbm{1}[k_j \le k] = \sum_{l=1}^L y_l \left( \sum_{j=1}^m \delta_j \mathbbm{1}[k_j \le k] \mathbbm{1}[l_j = l] \right) = \sum_{l=1}^L y_l w_l,
\end{equation}
where we define $w_l = \sum_{j=1}^m \delta_j \mathbbm{1}[k_j \le k] \mathbbm{1}[l_j = l]$. Let $W_l = \sum_{i=1}^l w_i$ be its cumulative sum with $W_0=0$. Applying summation by parts to the sum $\sum y_l w_l$:
\begin{align}
Z_k &= \sum_{l=1}^L y_l (W_l - W_{l-1})  \nonumber \\
&= y_L W_L + \sum_{l=1}^{L-1} (y_l - y_{l+1}) W_l. \label{eq:sum_of_D}
\end{align}

Now, notice that the term $W_l$ can be rewritten as 
\begin{equation}
    W_l = \sum_{i=1}^l \sum_{j=1}^m \delta_j \mathbbm{1}[k_j \le k] \mathbbm{1}[l_j = i] = \sum_{j=1}^m \delta_j \mathbbm{1}[k_j \le k] \mathbbm{1}[l_j \le l].
\end{equation}
This is exactly the term in the ranking-dominance condition!

From the assumption that $J^+$ ranking-dominates $J^-$, we know that $W_k = \sum_{j=1}^m \delta_j \mathbbm{1}[k_j \le k] \mathbbm{1}[l_j \le l] \ge 0$ for all $(k,l)$, with strict inequality for some $(k,l)$. Furthermore, because all terms $x_K, x_k-x_{k+1}, y_L, y_l-y_{l+1}$ are strictly positive, we conclude that $\Delta > 0$. This completes the proof.
\end{proof}

\newpage

\subsection{Additional Results on Statistical Testing}\label{appdx: test-addi}
\subsubsection{Regularity Conditions for Data Generating Processes}\label{appdx: regularity}
Define index subsets based on whether the distribution is degenerate:
\begin{align*}
    J &= \left\{ j=1,\dots,m \ \middle| \ \mu_j \not\in \{0,1\} \right\}, \\
    J' &= \left\{ j=1,\dots,m \ \middle| \ \mu_j \in \{0,1\} \right\}.
\end{align*}

When $j\in J'$, $\rho_{ij}$ is degenerate distributed, which follows that $\hat\rho_{ij}$ and hence $\hat\rho_j$ are also degenerate, equaling 0 or 1 almost surely. In other words, the subset $J'$ consists of indices where all individuals choose one specific option with probability 1. This is the set of ``trivial'' questions, and will contain no variation in the data. Hence for the purpose of analysis this subset can be omitted, as it contributes nothing to the data distribution. In the rest of this subsection, we will only focus on the index subset $J$ for clarity of exposition, as if $J'=\varnothing$. Importantly this is with the understanding that when $J' \neq \varnothing$, we simply add degenerate constant 0 or 1 dimensions to the asymptotic distribution of $\hat{\bm\rho}$, which does not change the result of statistical testing.
\begin{assumption}[\textbf{Regularity Conditions}]\label{assumption-regularity}
\

    \begin{enumerate}[(a)]
        \item $\Pr[K_{ij}>0]>0,\ \forall j$.
        \item $\bm V_{ij}$ and $\bm V_{i'j'}$ are independent for $i\neq i'$ or $j\neq j'$.
        \item $\hat{\bm\rho}_{i j} | \bm\rho_i$ and $\hat{\bm\rho}_{i' j'} | \bm\rho_{i'}$ are conditionally independent for $i \neq i'$ or $j \neq j'$.
        \item The distribution $F$ is non-degenerate.
        \item The distribution of $[I_{i1}, \dots, I_{im}]^T$ is non-degenerate.
    \end{enumerate}
\end{assumption}
\begin{remark}[For Regularity Conditions]
    \
    \begin{enumerate}[(i).]
        \item Assumption \ref{assumption-regularity}(a) rules out the case where the $j$-th problem will have zero observations with probability 1, in which case we can simply drop that problem index.
        \item Assumption \ref{assumption-regularity}(c) combined with Assumption \ref{assumption_1} (c) implies that $\hat{\bm\rho}_{i j}$ and $\hat{\bm\rho}_{i' j'}$ are (unconditionally) independent when $i\neq i'$.
        \item Assumption \ref{assumption-regularity}(d) is a non-degeneracy requirement over the distributions $F$, where non-degeneracy of a random vector means that there exist no nonzero linear combination of its entries that is almost surely constant. Intuitively, here it requires that there exist no deterministic relation among choice rates of different problems, except for problems where all choice rates are constant 0 or 1. This condition is equivalent to that $\bm\Sigma_0$ is invertible. Also note that having $J'\neq\varnothing$ will inevitably lead to degeneracy of $F$ on those dimensions, and as we discussed above, we omit those dimensions here and handle separately.
        \item Assumption \ref{assumption-regularity}(e) is a non-degeneracy requirements over the distributions $G$, namely requiring there exists no nonzero linear combination of its entries that is almost surely constant. Intuitively, this requires there exists no deterministic relation among whether a subject faces various problems. 
    \end{enumerate}
\end{remark}

\subsubsection{Normalizing $\bm\Omega$ to $\bm I$}\label{appdx: test-addi-norm}
Recall \autoref{eqn: Tn}, I use the statistic $T_N = \sqrt{N} \left\| \hat{\bm\rho} - \proj_Q(\hat{\bm\rho},\bm\Omega)\right\|_\Omega$, where $\left\|\bm x\right\|_\Omega = \sqrt{\bm x^T\bm \Omega \bm x}$, and $\proj_Q(\cdot,\bm\Omega)$ is the projection onto $Q$ with $\bm\Omega$-induced distance, defined as $\proj_Q(\bm x,\bm\Omega) = \arg\min_{\bm q \in Q} \ \left\|\bm x - \bm q\right\|_\Omega$.

Note that the expression with $\bm\Omega$-induced norm can be transformed into one with standard Euclidean norm by redefining the variables.

By letting $\hat{\bm\rho}' = \bm\Omega^{1/2} \hat{\bm\rho}$ and $Q' = \bm\Omega^{1/2} Q = \{\bm\Omega^{1/2} \bm x \ |\ \bm x \in Q\}$. We have 
$$
    \proj_Q(\hat{\bm\rho}, \Omega) = \arg\min_{\bm q \in Q} \ \left\| \bm \hat{\bm\rho} - \bm q \right\|_\Omega = \bm\Omega^{-1/2} \arg\min_{\bm q' \in Q'} \ \left\| \bm \hat{\bm\rho}' - \bm q' \right\| =  \bm\Omega^{-1/2} \proj_{Q'}(\hat{\bm\rho}'),
$$ 
where $\|\bm x\| = \sqrt{\bm x^T \bm x}$ and $\proj_Q(x)$ are the Euclidean norm and projection. Hence $T_n$ can be rewritten as 
$$
    T_n = \sqrt{n} \left\| \hat{\bm\rho}' - \proj_{Q'}(\hat{\bm\rho}') \right\|,\
    \text{where }\sqrt{n}(\hat{\bm\rho}' - \bm\Omega^{1/2} \bm\mu) \overset{d}{\longrightarrow}
    \mathcal{N}\big(\bm 0, \ \bm\Omega^{1/2} \bm\Sigma \bm\Omega^{1/2} \big).
$$
Hence without loss of generality we only need to consider the case $\bm\Omega = \bm I$.
\subsubsection{Pointwise Asymptotic Distribution of Statistic}\label{appdx: Tn}
\begin{lemma}\label{lemma:T_n}
Under Assumption \ref{assumption_1}, \ref{assumption_2}, and \ref{assumption-regularity},
    \begin{equation}
        T_N\overset{d}{\longrightarrow} T,
    \end{equation}
    where the limiting distribution satisfies either one of the following:
    \begin{enumerate}[i.]
        \item $T=0$ with probability 1, or
        \item $T^2$ is the mixture between a point mass at 0 and a mixture of generalized $\chi^2$ distributions, i.e.
        \begin{align}
            \Pr[T=0] &= \pi_0,\\
            T^2\ |\ T>0 &\overset{d}{=} \bm Z^T\bm\Lambda \bm Z,
        \end{align}
        where $\pi_0\in[0,1/2]$ is the probability mass at 0, and
         $\bm\Lambda = \diag(\lambda_1,\dots,\lambda_m)$ is positive semi-definite, and $\bm Z \sim \mathcal{N}(\bm 0, \bm I)$.
    \end{enumerate}
\end{lemma}

\begin{remark}[For Lemma \autoref{lemma:T_n}]
\

\begin{enumerate}[(i).]
    \item The intuition behind the asymptotic distribution of $T_N$ hinges crucially on the location of the true population mean $\bm\mu$ relative to the polytope $Q$. If $\bm\mu$ lies strictly in the interior of the polytope, the limiting distribution of $T_N$ degenerates to a point mass at zero. If $\bm\mu$ lies on the boundary of $Q$, the limiting distribution becomes non-degenerate: it is characterized as a mixture of a point mass at zero and a mixture of $\chi^2$ distribution. The non-trivial portion of this distribution emerges from the possibility of projecting onto different facets, edges, and vertices of the polytope.
    \item The probability mass at zero never exceeds 1/2, which makes simulating a critical value via subsampling pointwise valid. However, since the distribution of $T_N$ does not vary continuously in $\bm\mu$, it is unclear whether uniform validity is obtained via subsampling.
\end{enumerate}
\end{remark}

\begin{proof}
Recall that $\hat{\bm\rho}$ has asymptotic distribution 
$$
    \sqrt{N} \left( \hat{\bm\rho} - \bm\mu \right) \overset{d}{\longrightarrow}
    \mathcal{N}\big(\bm 0, \ \bm \Sigma \big),
$$
or
\begin{equation}
    \hat{\bm\rho} \approx \bm\mu + \frac{1}{\sqrt{N}} \bm\Sigma^{\frac{1}{2}} \bm Z,
    \label{eqn:asymptotic-rhohat}
\end{equation}
where $\bm Z \sim \mathcal{N}(\bm 0, \bm I)$ follows standard normal distribution. 

\vspace{0.5cm}

We will now derive the asymptotic distribution of $T_N$. Let $C$ be the hypercube $[0,1]^m$, and $Q$ is a $m$-dimensional polytope with $Q\subset C$. We will consider different cases of $\bm\mu \in Q$ under the null hypothesis, and show that the asymptotic distribution of $T_N^2$ is either a constant $0$, or a mixture between $\chi^2$ distributions and a point mass at 0 where the weight of the point mass does not exceed $\frac{1}{2}$. Interestingly, the following analysis relates closely to \citet{fang2019inference}'s discussion on convex set projections.

\begin{enumerate}[\text{\textbf{Case}} I.]

\item $\bm\mu \in \interior(Q) \cap \interior(C)$. As $N\to\infty$, we have $\Pr[\hat{\bm\rho} \in Q] \to 1$. Because $\hat{\bm\rho} \in Q$ implies $\pi_Q(\hat{\bm\rho}) = \hat{\bm\rho}$ and thus $T_N = 0$, this means $\Pr[\hat{\bm\rho} = 0] \to 1$ and thus $T_N$ converges to 0 in probability.

\item $\bm\mu \in \partial Q \cap \interior(C)$. In this case, the distribution of $T_N$ is more complicated and depends on the face of $Q$ that $\bm\mu$ lies on. We provide a characterization below. As discussed in the main text, we only need to consider the case of $\bm\Omega = \bm I$.

We now analyze the distribution of $T_N^2 = N \left\| \hat{\bm\rho} - \proj_{Q}(\hat{\bm\rho}) \right\|^2$, the squared projection distance of $\hat{\bm\rho}$ onto the polytope $Q$. Consider the normal cone $N_{Q}(\bm\mu)$ and tangent cone $T_{Q}(\bm\mu)$ at $\bm\mu$. 


As $n\to\infty$, $\hat{\bm\rho}$ lies close to $\bm\mu$, and thus we have $\proj_{Q}(\hat{\bm\rho}) = \bm\mu + \proj_{T_{Q}(\bm\mu)} (\hat{\bm\rho} - \bm\mu)$, which is the first-order expansion of the projection operator. Therefore, we have 
$$\hat{\bm\rho} - \proj_{Q}(\hat{\bm\rho}) = \hat{\bm\rho} - ( \bm\mu + \proj_{T_{Q}(\bm\mu)} (\hat{\bm\rho} - \bm\mu) ) = \proj_{N_{Q}(\bm\mu)} (\hat{\bm\rho} - \bm\mu),$$ by Moreau's Decomposition Theorem.

Now from the asymptotic distribution of $\hat{\bm\rho}$ as $\hat{\bm\rho} \approx \bm\mu + \frac{1}{\sqrt{N}} \bm\Sigma^{\frac{1}{2}} \bm Z$ where $\bm Z\sim \mathcal{N}(\bm 0, \bm I)$, it follows that the asymptotic distribution of $T_N^2 = N \left\| \hat{\bm\rho} - \proj_{Q}(\hat{\bm\rho}) \right\|^2$ is equal to the distribution $\| \proj_{N_{Q}(\bm\mu)} (\bm\Sigma^{\frac{1}{2}} \bm Z) \|^2$. Note that since $\mu\in\partial Q$, the normal cone $N_{Q}(\bm\mu)$ is always non-trivial, i.e. $N_{Q}(\bm\mu) \neq \{\bm 0\}$.

We now show the distribution of $T_N$ is a mixture of square root of generalized chi-squared distributions and a point mass at $\bm 0$. Note that the projection operator $\proj_{N_{Q}(\bm\mu)}$ is piecewise linear, namely that the entire space $\R^m$ can be partitioned into subsets $R_1, \dots, R_k$ such that 
$$\proj_{N_{Q}(\bm\mu)}(\bm x) = \bm P_k \bm x\ \text{ for } \bm x\in R_k,$$
where $\bm P_k$ are projection matrices, thus symmetric and idempotent. Each of these regions correspond to the projected point lying on a different face of $Q$. In particular, one of these regions is exactly the tangent cone $T_{Q}(\bm\mu)$ and has $\bm P_k = \bm 0$, giving the point mass at $\bm 0$. 

For the other regions where $\bm P_k \neq \bm 0$, we can write 
$$T_N^2 = \| \bm P_k \bm \Sigma^{\frac{1}{2}} \bm Z \|^2 = \bm Z^T \bm\Sigma^{\frac{1}{2}} \bm P_k \bm\Sigma^{\frac{1}{2}} \bm Z.$$
The matrix $\bm\Sigma^{\frac{1}{2}} \bm P_k \bm\Sigma^{\frac{1}{2}}$ is symmetric and thus has eigenvalue decomposition $\bm U \bm\Lambda \bm U^T$ where $\bm U$ is an orthogonal matrix and $\bm\Lambda = \diag(\lambda_1,\dots,\lambda_m)$. Because $\bm U \bm Z \overset{d}{=} \bm Z \sim \mathcal{N}(\bm 0, \bm I)$, we have 
$$T_N^2 = \bm Z^T \bm U \bm\Lambda \bm U^T \bm Z \overset{d}{=} \bm Z^T \bm\Lambda \bm Z = \sum_{j=1}^m \lambda_j Z_j^2,$$
which is a generalized chi-squared distribution.

Hence, $\| \proj_{N_{Q}(\bm\mu)} (\bm\Sigma^{\frac{1}{2}} \bm Z) \|^2$ follows a mixture distribution of generalized $\chi^2$ distributions and a point mass at $\bm 0$. Since $\bm\Sigma$ is full ranked under Assumption \ref{assumption-regularity} (d) and (e), and hence
$\bm\Sigma^{\frac{1}{2}} \bm Z$ is non-degenerate,
the point mass at $\bm 0$ corresponds only to when $\bm\Sigma^{\frac{1}{2}} \bm Z$ lies in the polar cone of $N_{Q}(\bm\mu)$, namely the tangent cone $T_{Q}(\bm\mu)$. For any closed convex cone $C' \subsetneq \R^m$, we have $\Pr[\bm Z \in C'] \le \frac{1}{2}$, with equality iff $C'$ is a half-space.

Take any point $A \neq C'$, and apply the Separating Hyperplane Theorem to the two sets $C'$ and $A$, namely that, there exists $\bm a\neq 0$ such that
\begin{equation}
    C'\subseteq \{\bm x\in\R^m:\bm a^T \bm x \le b \}.
\end{equation}
It is impossible for any point $\bm x\in C'$ to have $\bm a^T \bm x > 0$, because otherwise conic properties would then imply that $\bm a^T \bm x$ has no upper bound. Hence we know that $C'\subseteq \{\bm x\in\R^m:\bm a^T \bm x \le 0 \}$, i.e., the cone is contained in some closed half-space $H$. Then we have
\begin{equation}
    \Pr[\bm Z\in C']\le \Pr[\bm Z\in H] = \Pr[\bm a^T\bm Z\ge 0]=\Pr\left[\frac{\bm a^T\bm Z}{\|\bm a
    \|}\ge 0\right] = \Pr[\mathcal{N}(0,1)\ge 0]=\frac{1}{2}.
\end{equation}

\item $\bm\mu \in Q \cap \partial C$. $\bm\mu$ being on the boundary of $C$ implies that certain dimensions of $\bm\mu$ equal 0 or 1, and the distributions of $\bm\rho_i$ and $\hat{\bm\rho}_i$ are degenerate on those dimensions. In this case, we can project $Q$ and $C$ onto the non-degenerate dimensions, obtaining polytope $Q'$ and hypercube $C'$. Other quantities such as $\bm\mu$ and $\bm\Sigma$ are marginalized over the non-degenerate dimensions as $\bm\mu'$ and $\bm\Sigma'$. This reduces the problem to a case with smaller $m$ and $\bm\mu' \in \interior(C)$ that is already discussed above.
\end{enumerate}
\end{proof}

\subsubsection{Pointwise Validity of Subsampling}\label{appdx: subsampling}
According to Assumption \ref{assumption_1} (c), $(\bar{\bm \rho}_1,\bm I_1),\dots,(\bar{\bm \rho}_N,\bm I_N)$ is a sample of $n$ independent and identically distributed random variables, where
\begin{align}
    \bar{\bm \rho}_i = [\bar{\rho}_{i1},\dots,\bar{\rho}_{ij}]^T,\
    \bm I_i = [I_{i1},\dots,I_{ij}]^T,\ \forall i=1,\dots,N.
\end{align}

Consider a subsample size of $b<N$, with $b/N \to 0$ and $b\to\infty$ as $N\to \infty$. Label the total of $N_b = {N\choose b}$ size-$b$ subsets from 1 to ${n\choose b}$, and let $T_{N,b,i}$ denote the same statistic evaluated on the $i$-th subset. Let $G_N(x)$ denote the cumulative distribution function of the sampling distribution of $T_N$, namely
\begin{equation}
    G_N(x) = \Pr[T_N \le x].
\end{equation}
Let $\hat{G}_{N,b}$ denotes the subsampling approximation of the distribution above, namely
\begin{equation}
    \hat{G}_{N,b}(x) = \frac{1}{N_b} \sum_{i=1}^{N_b} \mathbbm{1}[T_{N,b,i} \le x].
\end{equation}

Define $g_{N,b}(1-\alpha)$ as the $1-\alpha$ quantile of $\hat{G}_{N,b}(x)$, specifically
\begin{equation}
    g_{N,b}(1-\alpha) = \inf\left\{ x \,\middle|\, \hat{G}_{N,b}(x) \ge 1-\alpha \right\}.
\end{equation}

\begin{lemma}[Pointwise Validity of Subsampling]\label{thm:subsampling-pointwise}
Assume $b/N \to 0$ and $b\to\infty$ as $N\to \infty$. Under Assumption \ref{assumption_1}, \ref{assumption_2}, and \ref{assumption-regularity}, for any $\boldsymbol{\mu}$ on the boundary of $Q$, $G(x)$ is continuous at $g(1-\alpha), \forall \alpha<0.5$. Consequently, taking the infimum over all $\boldsymbol{\mu} \in Q$ yields:
\begin{equation}
    \inf _{\boldsymbol{\mu} \in Q} \lim _{N \rightarrow \infty} \operatorname{Pr}\left[T_{N} \leq g_{N, b}(1-\alpha)\right]=1-\alpha.
\end{equation}
\end{lemma}
\begin{proof}
   Recall $(\bar{\bm \rho}_1,\bm I_1),\dots,(\bar{\bm \rho}_N,\bm I_N)$ is a sample of $N$ independent and identically distributed random variables.
Let $G_N(x)$ denote the cumulative distribution function of the sampling distribution of $T_N$, and $G(x)$ the corresponding limiting cumulative distribution function, whose $1-\alpha$ quantile is $g(1-\alpha)$. According to Lemma \ref{lemma:T_n}, 
\begin{enumerate}
    \item Under some $\bm \mu\in Q$, $G_N(x)$ converges in distribution to $G(x)=\mathbbm{1}\{x\ge 0\}$, i.e. the limiting distribution $T$ is a degenerate distribution equals 0 with probability 1.  Since $g_{N,b}(1-\alpha)\ge 0$, 
    \begin{equation}
        \lim_{N\to \infty}\Pr\{T_N\le g_{N,b}(1-\alpha)\}=1.
    \end{equation}
    \item Under some $\bm \mu\in Q$, $G_n(x)$ converges in distribution to $G(x)$, a mixture between a point mass at 0 and chi-squared distributions, with the point mass at 0 having weight no greater than $\frac{1}{2}$.

    The validity of this subsampling test process is  established in \cite{PolitisRomanoWolf1999}, Theorem 2.2.1 and Theorem 2.6.1. Since the weight of the point mass 0 never exceeds $\frac{1}{2}$, it follows that $G(x)$ is  continuous at $g(1-\alpha)$, for any $\alpha<\frac{1}{2}$. In practice, we usually choose $\alpha=0.05$. Therefore 
    \begin{align}
        &g_{N,b}(1-\alpha)\longrightarrow g(1-\alpha),\\
        \text{and } \quad &\lim_{N\to\infty}\Pr[T_N\le g_{N,b}(1-\alpha)]=1-\alpha.
    \end{align}
\end{enumerate}
Combine Case (a) and (b), 
\begin{equation}
    \inf_{\bm\mu\in Q}\lim_{N\to\infty}\Pr[T_N\le g_{N,b}(1-\alpha)]=1-\alpha.
\end{equation} 
\end{proof}

\subsubsection{Remarks on Testing Without Individual Data}\label{sec: test-no-indi}
The numerical delta method with bootstrap procedure described above requires knowing which responses come from the same individual. In some scenarios, such individual-level data may be unavailable and we only have access to population-level aggregate response. In this subsection, we discuss a possible alternative bootstrap procedure, and conditions under which is valid.

Consider the scenario where $n$ individuals respond to $m$ binary choice problems following Assumptions \ref{assumption_1}, \ref{assumption_2}, and \ref{assumption-regularity}. Due to the randomness in $\bm K_i$, each individual may respond each problem a different number of times. Assume there is no information on which responses come from which individuals, and we only observe a collection of $N_j$ responses for each problem $j=1,\dots,m$, with $N_0 = \sum_{j=1}^m N_j$. In this case, one natural alternative bootstrap procedure would be to treat the $N_0$ responses as if each comes from a different individual who responded to only 1 problem. Then bootstrap proceeds as normal over this ``assumed'' population. 

It can be shown that under some additional assumptions, the critical value obtained from this alternative bootstrap is still valid. This is formalized as follows. We first state the additional assumption required for this equivalence:
\begin{assumption}
    Following the setup of Assumption \ref{assumption_1},
    \begin{enumerate}[(a).]
        \item $\bar{K} = 1$.
        \item For all $j\neq j'$, either $\Pr[K_{ij}>0, \ K_{ij'}>0] = 0$ or $(\Sigma_0)_{jj'} = 0$.
    \end{enumerate}
    \label{assumption_4}
\end{assumption}

The first assumption $\bar{K}=1$ means that each individual responds to each problem at most once. This assumption naturally holds in a wide range of real-world scenarios, such as polls and elections. In lab experiments, participants are more likely to be asked to respond to the same problem multiple times, but that is also when individual-level data is more likely to be available. When this assumption is violated and an individual may provide multiple responses to the same problem, these responses provide less information about the true population $\mu$ than we assumed, leading to under-estimation of the true covariance, and an over-confidence in testing.

The second assumption is more delicate. It implies that for any two choice problems that may be played by the same individual, their responses must be uncorrelated over the population. This is a rather strong assumption, especially for choice problems that share a common choice option. For problems that share no common choice options, this assumption is relatively more justifiable despite still being demanding. The good news is that even when this assumption is violated, so long as the cross covariance $(\Sigma_0)_{jj'}$ are small, the distortion introduced will also small, as we shall see that adopting this alternative bootstrap procedure is equivalent to mis-specifying the covariance matrix by dropping the off-diagonal terms.

The proposed alternative testing procedure corresponds to the following alternative d.g.p.:
\begin{assumption}[\textbf{Alternative d.g.p.}]\label{assumption_5}
\ 
\begin{enumerate}[(a).]
    \item There is a set of $n$ choice options, forming $m:=\binom{n}{2}$ binary choice problems.
    \item Let $F$ be a non-degenerate distribution over $[0,1]^m$, with mean $\bm\mu$ and covariance $\bm\Sigma_0$.
    \item Let $\mathcal{J}$ be a distribution over the discrete set $\{1,2,\dots,m\}$.
    \item Randomly sampling $N'$ individuals from a large population facing choice problems:
    \begin{enumerate}[I.]
        \item $(\bm\rho_i,J_i)\sim_{\iid} F\times \mathcal{J}$.
        \item $J_i$ denotes the problem index that individual $i$ responds to. Note that each individual only responds to one problem.
        \item $\bm\rho_i$ and $J_i$ are independent.
    \end{enumerate}
    \item $\hat{\rho}_{i j} | \rho_{ij} \sim \mathrm{Bernoulli}(\rho_{ij})$.
\end{enumerate}
\end{assumption}
Under this alternative d.g.p., an estimator $\hat{\bm\rho}^\mathrm{alt}$ can be similarly constructed as 
\begin{equation}
    \hat{\rho}^{\mathrm{alt}}_j = \frac{\sum_{i=1}^{N'} \hat{\rho}_{ij} \mathbbm{1}[J_i=j]}{\sum_{i=1}^{N'} \mathbbm{1}[J_i=j]}, \qquad j=1,\dots,m.
    \label{eqn:estimator-alt-dgp}
\end{equation}

The correspondence between the two d.g.p.s are established by the following result:

\begin{lemma}
\label{prop:validity-without-indiv-data}
    Consider the two d.g.p.s defined by Assumptions \ref{assumption_1}-\ref{assumption_4} and Assumption \ref{assumption_5}. Let $p_j = \Pr[K_{ij} > 0]$ and $q_j = \Pr[J_i = j]$.
    If $q_j = p_j / (\sum_{j'=1}^m p_{j'})$ and $N' = N \left(\sum_{j=1}^m p_j\right)$, then the two estimators $\hat{\bm\rho}$ and $\hat{\bm\rho}^{\mathrm{alt}}$ defined by \eqref{eqn:estimator-standard-dgp} and \eqref{eqn:estimator-alt-dgp} have the same asymptotic distribution. Furthermore, their bootstrap estimators also have the same asymptotic distribution.
\end{lemma}

\begin{proof}
Denote $q_j = \Pr[J_i = j]$. Also we adopt the same notations $p_j = \Pr[K_{ij} > 0]$ and $p_{jj'} = \Pr[K_{ij} > 0, K_{ij'} > 0]$ from previous analysis.

Recall the estimator $\hat{\bm\rho}^{\mathrm{alt}}$ under the alternative d.g.p.
\begin{equation}
    \hat{\rho}^{\mathrm{alt}}_j = \frac{\sum_{i=1}^{N_0} \hat{\rho}_{ij} \mathbbm{1}[J_i=j]}{\sum_{i=1}^{N_0} \mathbbm{1}[J_i=j]}, \quad j=1,\dots,m.
\end{equation}
Following very similar derivation of the asymptotic distribution in Lemma \ref{lemma:rho_hat}, we have
\begin{equation}
    \sqrt{N'} (\hat{\bm\rho}^{\mathrm{alt}} - \bm\mu) \overset{d}{\longrightarrow} \mathcal{N}(\bm 0, \bm\Sigma^{\mathrm{alt}}),
\end{equation}
where 
\begin{align}
    \Sigma^{\mathrm{alt}}_{jj} &= \frac{1}{q_j} \mu_j(1-\mu_j), \quad \forall j \\
    \Sigma^{\mathrm{alt}}_{jj'} &= 0, \qquad\qquad\qquad \forall j\neq j'.
\end{align}

To compare this to the distribution of the original estimator $\hat{\boldsymbol{\rho}}$ (which is scaled by $\sqrt{N}$), we adjust the scaling:
$$\sqrt{N}\left(\hat{\boldsymbol{\rho}}^{\mathrm{alt}}-\boldsymbol{\mu}\right) \xrightarrow{d} \mathcal{N}\left(\mathbf{0}, \frac{N}{N^{\prime}} \boldsymbol{\Sigma}^{\mathrm{alt}}\right).$$

On the other hand, for the original estimator $\hat{\bm\rho}$, its asymptotic covariance $\bm\Sigma$ when $\bar{K} = 1$ reduces to 
\begin{align}
    \Sigma_{jj} &= \frac{1}{p_j} \mu_j (1-\mu_j), \quad \forall j \\
    \Sigma_{jj'} &= \frac{p_{jj'}}{p_j p_{j'}} (\Sigma_0)_{jj'}, \ \quad \forall j, j'.
\end{align}

From Assumption \ref{assumption_4} (b), the off-diagonal terms $\Sigma_{jj'}$ vanish. Also, the correspondence between the two d.g.p. implies that $N' q_j = N p_j$. Therefore, the two estimators $\hat{\bm\rho}^{\mathrm{alt}}$ and $\hat{\bm\rho}$ have exactly the same asymptotic distribution. 

From here, it immediately follows that the bootstrap estimators also share the same asymptotic distributions, which is equal to that of the original estimators.
\end{proof}

\newpage\subsection{Proof of Lemma \ref{lemma:rho_hat}}\label{pf:rho_hat}
Define
$$
    S_{Nj} = \frac{1}{N} \sum_{i=1}^N \bar{\rho}_{ij} I_{ij}, \quad 
    T_{Nj} = \frac{1}{N} \sum_{i=1}^N I_{ij}, \quad \forall j=1,\dots,m,
$$
and we have $\hat{\rho}_j = S_{Nj}/T_{Nj}$. Multivariate Central Limit Theorem gives
$$
    \sqrt{N} \left( 
    \begin{bmatrix}
        S_{N1} \\
        T_{N1} \\
        \vdots \\
        S_{Nm} \\
        T_{Nm}
    \end{bmatrix} - \bm u \right) \overset{d}{\longrightarrow}
    \mathcal{N}\big(\bm 0, \ \bm S\big),
$$
where
$$
    \bm u = \begin{bmatrix}
        \E[\bar{\rho}_{i1} I_{i1}] \\
        \E[I_{i1}] \\
        \vdots \\
        \E[\bar{\rho}_{im}I_{im}] \\
        \E[I_{im}]
    \end{bmatrix}
$$
and
$$
    \bm S = \begin{bmatrix}
        \Var(\bar{\rho}_{i1} I_{i1}) & \Cov(\bar{\rho}_{i1} I_{i1}, I_{i1}) & \hdots & \Cov(\bar{\rho}_{i1} I_{i1}, \bar{\rho}_{im} I_{im}) & \Cov(\bar{\rho}_{i1} I_{i1}, I_{im}) \\
        \Cov(I_{i1}, \bar{\rho}_{i1} I_{i1}) & \Var(I_{i1}) & \hdots & \Cov(I_{i1}, \bar{\rho}_{im} I_{im}) & \Cov(I_{i1}, I_{im}) \\
        \vdots & \vdots & \ddots & \vdots & \vdots \\
        \Cov(\bar{\rho}_{im} I_{im}, \bar{\rho}_{i1} I_{i1}) & \Cov(\bar{\rho}_{im} I_{im}, I_{i1}) & \hdots & \Var(\bar{\rho}_{im} I_{im}) & \Cov(\bar{\rho}_{im} I_{im}, I_{im}) \\
        \Cov(I_{im}, \bar{\rho}_{i1} I_{i1}) & \Cov(I_{im}, I_{i1}) & \hdots & \Cov(I_{im}, \bar{\rho}_{im} I_{im}) & \Var(I_{im})
    \end{bmatrix}.
$$

Consider the function $\phi: \R^{2m} \to \R^m$ defined as
$$
    \phi \left(\begin{bmatrix}
        x_1 \\
        y_1 \\
        \vdots \\
        x_m \\
        y_m
    \end{bmatrix}\right) = \begin{bmatrix}
        x_1/y_1 \\
        \vdots \\
        x_m/y_m
    \end{bmatrix}.
$$
Its Jacobian $\bm J \in \R^{2m \times m}$ is 
$$
    \bm J = \begin{bmatrix}
        1/y_1 & & \\
        -x_1/y_1^2 & & \\
         & \ddots & \\
         & & 1/y_m \\
         & & -x_m/y_m^2
    \end{bmatrix}.
$$
Assuming $\phi$ is differentiable at $\bm u$ (checked in later analysis), from the Delta method, we have
\begin{equation}
    \sqrt{n} \left( \hat{\bm\rho} - \phi(\bm u) \right) \overset{d}{\longrightarrow}
    \mathcal{N}\big(\bm 0, \ \bm J(\bm u)^T \bm S \bm J(\bm u) \big)
    \label{eqn:delta-method-rhohat}
\end{equation}
where $\bm J(\bm u)$ is the Jacobian at the point $\bm u$.

We now derive the expression of $\bm u$,  $\bm S$ and $\bm J$, and verify that $\phi$ is differentiable at $\bm u$.

We first compute the conditional mean and covariance of $\bar{\rho}_{ij} | \rho_{ij}, K_{ij}$. We have
$$
    \E[\bar{\rho}_{ij} | \rho_{ij}, K_{ij}] = \frac{1}{K_{ij}} \sum_{k=1}^{K_{ij}} \E \left[\hat{\rho}_{ij}^{(k)} \middle| \rho_{ij}, K_{ij} \right] = \frac{1}{K_{ij}} \sum_{k=1}^{K_{ij}} \rho_{ij} = \rho_{ij}.
$$
Also,
\begin{align*}
    \Var[\bar{\rho}_{ij} | \rho_{ij}, K_{ij}] &= \Var\left[ \frac{1}{K_{ij}} \sum_{k=1}^{K_{ij}} \hat{\rho}_{ij}^{(k)} \middle| \rho_{ij}, K_{ij} \right] \\
    &= \frac{1}{K_{ij}^2} \sum_{1 \le k_1, k_2 \le K_{ij}} \Cov\left(\hat{\rho}_{ij}^{(k_1)}, \hat{\rho}_{ij}^{(k_2)} \middle| \rho_{ij}, K_{ij}\right) \\
    &= \frac{1}{K_{ij}^2} \sum_{1 \le k_1, k_2 \le K_{ij}} \E_{\bm V_{ij} \sim H(K_{ij}, \rho_{ij})} \left[ \Cov\left(\hat{\rho}_{ij}^{(k_1)}, \hat{\rho}_{ij}^{(k_2)} \middle| \rho_{ij}, K_{ij}, \bm V_{ij}\right) \right] \\
    &= \frac{1}{K_{ij}^2} \sum_{1 \le k_1, k_2 \le K_{ij}} \E_{\bm V_{ij} \sim H(K_{ij}, \rho_{ij})} \left[ \rho_{ij}(1-\rho_{ij}) (V_{ij})_{k_1k_2} \right] \\
    &= \rho_{ij}(1-\rho_{ij}) v^2_{K_{ij}, \rho_{ij}},
\end{align*}
where we define
$$
    v^2_{k,p} = \frac{1}{k^2} \sum_{1 \le k_1, k_2 \le k} \E_{\bm V \sim H(k,p)}\left[ V_{k_1k_2} \right], \quad \forall k=1,\dots, \bar{K} = \E_{\bm V \sim H(k,p)}\left[ \bm 1^T \bm V \bm 1 \right], \ \forall p\in [0,1].
$$
Note we have $v^2_{k,p} \ge 0$ because $\bm V \sim H(k,p)$ is positive semi-definite.
Furthermore, for $j \neq j'$, we have
\begin{align*}
    \E[\bar{\rho}_{ij} \bar{\rho}_{ij'} | \rho_{ij}, \rho_{ij'}, K_{ij}, K_{ij'} ] &= \E\left[ \frac{1}{K_{ij} K_{ij'}}  \left(\sum_{k=1}^{K_{ij}} \hat{\rho}_{ij}^{(k)}\right) \left(\sum_{k'=1}^{K_{ij'}} \hat{\rho}_{ij'}^{(k')}\right) \middle| \rho_{ij}, \rho_{ij'}, K_{ij}, K_{ij'} \right] \\
    &= \frac{1}{K_{ij} K_{ij'}} \sum_{k=1}^{K_{ij}} \sum_{k'=1}^{K_{ij'}} \E\left[ \hat{\rho}_{ij}^{(k)} \hat{\rho}_{ij'}^{(k')} \middle| \rho_{ij}, \rho_{ij'}, K_{ij}, K_{ij'} \right] \\
    &= \frac{1}{K_{ij} K_{ij'}} \sum_{k=1}^{K_{ij}} \sum_{k'=1}^{K_{ij'}} \E\left[ \hat{\rho}_{ij}^{(k)} \middle| \rho_{ij}, K_{ij} \right] \E\left[ \hat{\rho}_{ij'}^{(k')} \middle| \rho_{ij'}, K_{ij'} \right] \text{(Assumption \ref{assumption_2} (d))} \\
    &= \frac{1}{K_{ij} K_{ij'}} \sum_{k=1}^{K_{ij}} \sum_{k'=1}^{K_{ij'}} \rho_{ij} \rho_{ij'} \\
    &= \rho_{ij} \rho_{ij'},
\end{align*}
where we used the conditional independence of $\hat{\rho}_{ij}^{(k)}$ and $\hat{\rho}_{ij'}^{(k')}$ given $\rho_{ij}, \rho_{ij'}, K_{ij}, K_{ij'}$.

\vspace{1cm}

Define $p_j = \E[I_{ij}] = \Pr[K_{ij} > 0]$ and $p_{jj'} = \E[I_{ij} I_{ij'}] = \Pr[K_{ij} > 0, K_{ij'} > 0]$, which are both determined by the distribution $G$.

To compute $\bm u$, we have
\begin{align*}
    \E[\bar{\rho}_{ij} I_{ij}] &= \E_{\rho_{ij}, K_{ij}}[\E[\bar{\rho}_{ij} | \rho_{ij}, K_{ij}] \cdot I_{ij}] \\
    &= \E_{\rho_{ij}, K_{ij}}[\rho_{ij} I_{ij}] \\
    &= p_j \mu_{j}. \\
    \E[I_{ij}] &= p_j.
\end{align*}

To compute $\bm S$, we have
\begin{align*}
    \Var[\bar{\rho}_{ij} I_{ij}] &= \E[\bar{\rho}_{ij}^2 I_{ij}] - \E[\bar{\rho}_{ij} I_{ij}]^2 \\
    &= \E_{\rho_{ij}, K_{ij}} \left[\E[\bar{\rho}_{ij}^2 | \rho_{ij},K_{ij}] \cdot I_{ij} \right] - p_j^2 \mu_j^2 \\
    &= \E_{\rho_{ij}, K_{ij}} \left[\left( \Var(\bar{\rho}_{ij}| \rho_{ij}, K_{ij}) + \E[\bar{\rho}_{ij}| \rho_{ij}, K_{ij}]^2 \right) \cdot I_{ij} \right] - p_j^2 \mu_j^2 \\
    &= \E_{\rho_{ij}, K_{ij}} \left[\left( \rho_{ij}(1-\rho_{ij}) v^2_{K_{ij}, \rho_{ij}} + \rho_{ij}^2 \right) \cdot I_{ij} \right] - p_j^2 \mu_j^2 \\
    &= \E_{\rho_{ij}, K_{ij}} \left[\rho_{ij}(1-\rho_{ij}) \cdot v^2_{K_{ij}, \rho_{ij}} I_{ij} \right] + p_j \E[\rho_{ij}^2] - p_j^2 \mu_j^2 \\
    &= \bar{v}^2_j + p_j \left( (\Sigma_0)_{jj} + \mu_j^2 \right) - p_j^2 \mu_j^2 \\
    &= p_j(1-p_j) \mu_j^2 + \bar{v}^2_j + p_j (\Sigma_0)_{jj}.
\end{align*}
where we define $\bar{v}^2_j = \E_{\rho_{ij}, K_{ij}} \left[\rho_{ij}(1-\rho_{ij}) \cdot v^2_{K_{ij}, \rho_{ij}} I_{ij} \right] \ge 0$, and use the fact that $\E[\rho_{ij}^2] = \Var(\rho_{ij}) + \E[\rho_{ij}]^2 = (\Sigma_0)_{jj} + \mu_j^2$. Also,
$$
    \Var[I_{ij}] = \E[I_{ij}^2] - \E[I_{ij}]^2 = p_j(1-p_j).
$$
and 
\begin{align*}
    \Cov(\bar{\rho}_{ij} I_{ij}, I_{ij}) = \E[\bar{\rho}_{ij} I_{ij}^2] - \E[\bar{\rho}_{ij} I_{ij}]  \cdot \E[I_{ij}] = p_j (1-p_j) \mu_j.
\end{align*}
For the terms where $j \neq j'$,
\begin{align*}
    \Cov(\bar{\rho}_{ij} I_{ij}, \bar{\rho}_{ij'} I_{ij'}) &= \E[\bar{\rho}_{ij} \bar{\rho}_{ij'} I_{ij} I_{ij'}] - \E[\bar{\rho}_{ij} I_{ij}]  \cdot \E[\bar{\rho}_{ij'} I_{ij'}] \\
    &= \E_{\rho_{ij}, \rho_{ij'}, K_{ij}, K_{ij'}} \left[\E[\bar{\rho}_{ij} \bar{\rho}_{ij'} | \rho_{ij},\rho_{ij'}, K_{ij},K_{ij'}] \cdot I_{ij} I_{ij'} \right] - p_j p_{j'} \mu_j \mu_{j'} \\
    &= \E_{\rho_{ij}, \rho_{ij'}, K_{ij}, K_{ij'}} \left[ \rho_{ij} \rho_{ij'} \cdot I_{ij} I_{ij'} \right] - p_j p_{j'} \mu_j \mu_{j'}\\
    &= p_{jj'} (\mu_j \mu_{j'} + (\Sigma_0)_{jj'}) - p_j p_{j'} \mu_j \mu_{j'} \\
    &= (p_{jj'} - p_j p_{j'}) \mu_j \mu_{j'} + p_{jj'} (\Sigma_0)_{jj'}.
\\
    \Cov(\bar{\rho}_{ij} I_{ij}, I_{ij'}) &= \E[\bar{\rho}_{ij} I_{ij} I_{ij'}] - \E[\bar{\rho}_{ij} I_{ij}]  \cdot \E[I_{ij'}] \\
    &= \E_{\rho_{ij}, \rho_{ij'}, K_{ij}, K_{ij'}} \left[\E[\bar{\rho}_{ij} | \rho_{ij}, K_{ij}] \cdot I_{ij} I_{ij'} \right] - p_j p_{j'} \mu_j \\
    &= \E_{\rho_{ij}, \rho_{ij'}, K_{ij}, K_{ij'}} \left[ \rho_{ij} \cdot I_{ij} I_{ij'} \right] - p_j p_{j'} \mu_j \\
    &= (p_{jj'} - p_j p_{j'}) \mu_j.
\\
    \Cov(I_{ij}, I_{ij'}) &= \E[I_{ij} I_{ij'}] - \E[I_{ij}]  \cdot \E[I_{ij'}] \\
    &= p_{jj'} - p_j p_{j'}.
\end{align*}
Furthermore, the Jacobian $\bm J(\bm u)$ at point $\bm u$ is
$$
    \bm J(\bm u) = \begin{bmatrix}
        1/p_1 & & \\
        -\mu_1/p_1 & & \\
         & \ddots & \\
         & & 1/p_m \\
         & & -\mu_m/p_m
    \end{bmatrix}.
$$
Note that our assumption $p_j > 0, \forall j$ implies that $\phi$ is differentiable at $\bm u$. Hence the Jacobian is well defined, and the Delta method applies. This gives the expressions of $\bm u$, $\bm S$ and $\bm J$.

\vspace{1cm}

Finally, we plug the expressions of $\bm u$, $\bm S$ and $\bm J$ into \autoref{eqn:delta-method-rhohat}.

The asymptotic mean of $\hat{\bm\rho}$ is $\phi(\bm u) = \bm\mu$.

Denote the asymptotic variance of $\hat{\bm\rho}$ as $\bm\Sigma = \bm J^T \bm S \bm J$. From the block diagonal structure of $\bm J$, the diagonal entries of $\bm\Sigma$ are
\begin{align*}
    \Sigma_{jj} &= \begin{bmatrix}
        1/p_j & -\mu_j/p_j
    \end{bmatrix} \begin{bmatrix}
        p_j(1-p_j) \mu_j^2 + \bar{v}^2_j + p_j (\Sigma_0)_{jj} & p_j (1-p_j) \mu_j \\
        p_j (1-p_j) \mu_j & p_j (1-p_j)
    \end{bmatrix} \begin{bmatrix}
        1/p_j \\ -\mu_j/p_j
    \end{bmatrix} \\
    &= \frac{1}{p_j^2} \left( \bar{v}^2_j + p_j (\Sigma_0)_{jj} \right).
\end{align*}
The off-diagonal entries of $\bm\Sigma$ are
\begin{align*}
    \Sigma_{jj'} &= \begin{bmatrix}
        1/p_j & -\mu_j/p_j
    \end{bmatrix} \begin{bmatrix}
        (p_{jj'} - p_j p_{j'}) \mu_j \mu_{j'} + p_{jj'} (\Sigma_0)_{jj'} & (p_{jj'} - p_j p_{j'}) \mu_j \\
        (p_{jj'} - p_j p_{j'}) \mu_{j'} & p_{jj'} - p_j p_{j'}
    \end{bmatrix} \begin{bmatrix}
        1/p_{j'} \\ -\mu_{j'}/p_{j'}
    \end{bmatrix} \\
    &= \frac{p_{jj'}}{p_j p_{j'}} (\Sigma_0)_{jj'}.
\end{align*}

To sum up, the asymptotic distribution of $\hat{\bm\rho}$ is 
$$
    \sqrt{N} \left( \hat{\bm\rho} - \bm\mu \right) \overset{d}{\longrightarrow}
    \mathcal{N}\big(\bm 0, \ \bm \Sigma \big)
$$
where $\bm\Sigma$ is a matrix with
\begin{align*}
    \Sigma_{jj} &= \frac{1}{p_j^2} \bar{v}^2_j + \frac{1}{p_j} (\Sigma_0)_{jj}, \quad \forall j, \\
    \Sigma_{jj'} &= \frac{p_{jj'}}{p_j p_{j'}} (\Sigma_0)_{jj'}, \quad \forall j\neq j'.
\end{align*}

We can show that under Assumption \ref{assumption-regularity} (d) and (e), the matrix $\bm\Sigma$ is positive-definite. Define matrix $\bm M \in \R^{m\times m}$ such that $M_{jj} = \frac{1}{p_j}$ and $M_{jj'} = \frac{p_{jj'}}{p_j p_{j'}}$ for $j\neq j'$. Consider random vector $\bm U = \left[ \frac{1}{p_1}I_{i1}, \dots, \frac{1}{p_m}I_{im} \right]^T$. We have $\bm M = \E[\bm U \bm U^T]$ as the Gram matrix of $\bm U$. From Assumption \ref{assumption-regularity} (e), we know $\bm U$ has non-degenerate distribution, and thus $\bm M$ is positive definite. From Assumption \ref{assumption-regularity} (d), we know $\bm\Sigma_0$ is positive definite. Hence by Schur product theorem, the Hadamard product $\bm M \circ \bm\Sigma_0$ is positive definite. Finally, because $\bm\Sigma = \bm M \circ \bm\Sigma_0 + \diag([\frac{1}{p_1^2} \bar{v}^2_1, \dots, \frac{1}{p_m^2} \bar{v}^2_m])$ and $\bar{v}^2_j \ge 0$, it follows that $\bm\Sigma$ is positive definite.

\newpage
\subsection{Proof of Theorem \ref{thm:uniform-test}}
The proof largely follows Theorem 3.5 of \cite{hong2018numerical}. In the following, we will verify that the assumptions of Theorem 3.5 holds.

We first verify that $\phi$ is convex. For any two points $\bm x_1, \bm x_2$, let $\bm p_1 = \proj_Q(\bm x_1)$ and $\bm p_2 = \proj_Q(\bm x_2)$. For any $\lambda\in (0,1)$, we have
\begin{align*}
    \phi(\lambda\bm x_1 + (1-\lambda) \bm x_2) &= \| \lambda\bm x_1 + (1-\lambda) \bm x_2 - \proj_Q\left(\lambda\bm x_1 + (1-\lambda) \bm x_2\right) \| \\
    &\le \| \lambda\bm x_1 + (1-\lambda) \bm x_2 - \left( \lambda \bm p_1 + (1-\lambda) \bm p_2 \right) \| \\
    &= \| \lambda(\bm x_1 - \bm p_1) + (1-\lambda) (\bm x_2 - \bm p_2) \| \\
    &\le \lambda \|\bm x_1 - \bm p_1\| + (1-\lambda) \| \bm x_2 - \bm p_2 \| \\
    &= \lambda \phi(\bm x_1) + (1-\lambda) \phi(\bm x_2),
\end{align*}
which shows the convexity of $\phi$. Here the first inequality follows from the optimality of the projection operator, and the second inequality follows from triangle inequality.

We next verify that $\phi$ is $1$-Lipschitz. For any two points $\bm x_1, \bm x_2$, for any $\bm q\in Q$, triangle inequality gives $\|\bm x_1 - \bm q\| \le \|\bm x_1 - \bm x_2\| + \|\bm x_2 - \bm q\|$. Take the infimum of both sides over $\bm q\in Q$, we have $\phi(\bm x_1) \le \|\bm x_1 - \bm x_2\| + \phi(\bm x_2)$, using the fact that $\phi(\bm x) = \inf_q \|\bm x - \bm q\|$. Swapping $\bm x_1$ and $\bm x_2$ gives $\phi(\bm x_2) \le \|\bm x_1 - \bm x_2\| + \phi(\bm x_1)$. Combining the two inequalities, we get the desired 1-Lipschitz property
\begin{equation*}
    |\phi(\bm x_1) - \phi(\bm x_2)| \le \|\bm x_1 - \bm x_2\|.
\end{equation*}

It remains to verify that Assumptions 3.1 and 3.2 from \cite{hong2018numerical} hold in our case. Notice that we already showed in Lemma \ref{lemma:rho_hat} that $\sqrt{N} (\hat{\bm\rho} - \bm\mu)$ converges to the limiting distribution $\mathcal{N}(\bm 0, \bm\Sigma(\bm\mu))$.
Because the observed choice data are strictly bounded in $[0,1]^m$, all moments are uniformly bounded across all data-generating processes corresponding to $\boldsymbol{\mu} \in Q$. By the uniform multivariate Central Limit Theorem, this guarantees that the weak convergence holds uniformly over $\boldsymbol{\mu} \in Q$. Furthermore, because $Q$ is a compact set and the asymptotic covariance $\boldsymbol{\Sigma}(\boldsymbol{\mu})$ is smooth and bounded, the uniform asymptotic tightness required by Assumption 3.1 in \citet{hong2018numerical} is fully satisfied.

For Assumption 3.1(ii), in our setup the variable $\mathbb{Z}^*_n$ from \cite{hong2018numerical}'s notation is the bootstrap estimator residual $\sqrt{N}(\hat{\bm\rho}_b - \hat{\bm\rho})$. The validity of bootstrap process implies that the bootstrap estimator has the same (pointwise) asymptotic distribution. Thus the same argument as Assumption 3.1(i) above applies, giving convergence in $L^\infty$ of the bootstrap estimator $\sqrt{N}(\hat{\bm\rho}_b - \bm\mu)$ to $\mathcal{N}(\bm 0, \bm\Sigma(\bm\mu))$. Combining this with the convergence in $L^\infty$ of $\sqrt{N}(\hat{\bm\rho} - \bm\mu)$ and triangle inequality, it follows that $\sqrt{N}(\hat{\bm\rho}_b - \hat{\bm\rho})$ also converges in $L^\infty$ to $\mathcal{N}(\bm 0, \bm\Sigma(\bm\mu))$.

For Assumption 3.2 from \cite{hong2018numerical}, we temporarily adopt the notations from that paper. Notice that in our case, the limiting distribution $\mathbb{G}_0$ is $\mathcal{N}(\bm 0, \bm\Sigma(\bm\mu))$, a multivariate Gaussian with full-rank covariance. Also notice that $\mathcal{C}_{a,d,x}$ is a polytope, hence $\partial \mathcal{C}_{a,d,x}$ is a set with Borel measure zero. Therefore we always have $\Pr[\mathbb{G}_0 \in \partial \mathcal{C}_{a,d,x}] = 0$, and Assumption 3.2 is satisfied.

Therefore, all assumptions of Theorem 3.5 of \cite{hong2018numerical} are satisfied, and we obtain the desired result.

\newpage
\subsection{Statistical Test for Heterogeneous Tastes}\label{appdx:test-heter}

Model each participant $i$ as having an ``underlying'' preference vector $\bm\rho_i \in [0,1]^{17}$, and their choice for the $j$-th menu is a $\mathrm{Bernoulli}(\rho_{ij})$ random variable independent across repetitions and menus. We want to test whether all participants share the same taste, in the sense that whether they prefer one type over another are consistent, and only the strength of the preference may vary. Concretely, the null hypothesis to test is that for any two participants $i, i'$, we have $\sgn(\rho_{ij} - \frac{1}{2}) = \sgn(\rho_{i'j} - \frac{1}{2})$ for all $j$. \footnote{Here we omit the specific case of $\rho_{ij} = \frac{1}{2}$, which can be considered a limiting case both when $\sgn(\rho_{ij}-\frac{1}{2}) = 1$ and when $\sgn(\rho_{ij}-\frac{1}{2}) = 1$.}

Let $\hat{\rho}_{ij}^{(k)}$ be participant $i$'s $k$-th repetition response to binary menu $j$, and denote $\bar\rho_{ij} = \frac{1}{10} \sum_{k=1}^{10} \hat{\rho}_{ij}^{(k)}$. To test this null hypothesis, first notice that it can be decomposed into a total of $2^{17}$ individual hypotheses, each corresponding to one fixed sign pattern of $\sgn(\bm\rho_i - \frac{1}{2})$. For each individual hypothesis, we perform the likelihood ratio test using the maximum likelihood estimation of parameters $\bm\rho_i$. Let $\bm s \in \{\pm 1\}^{17}$ represent a given sign pattern of $\sgn(\bm\rho - \frac{1}{2})$, then the maximum likelihood estimator of $\bm\rho_i$ is 
\begin{equation}
    \hat{\rho}_{ij}^{\mathrm{MLE}}(\bm s) = \left\{ \begin{matrix}
        \max\{\bar\rho_{ij}, \frac{1}{2}\} & \text{ if } s_j = 1 \\
        \min\{\bar\rho_{ij}, \frac{1}{2}\} & \ \ \ \text{ if } s_j = -1
    \end{matrix} \right..
\end{equation}
From this, we can compute the maximum log-likelihood under the restricted model with sign pattern $\bm s$ as 
\begin{equation}
    \ell_0(\bm s) = \sum_i\sum_j\sum_k \  \hat\rho_{ij}^{(k)} \log \hat{\rho}_{ij}^{\mathrm{MLE}}(\bm s) + \left(1-\hat\rho_{ij}^{(k)}\right) \log \left(1-\hat{\rho}_{ij}^{\mathrm{MLE}}(\bm s)\right).
\end{equation}
On the other hand, under the unrestricted model, the maximum likelihood estimator of $\bm\rho_i$ will simply be the sample mean $\hat{\rho}_{ij}^{\mathrm{MLE}} = \bar{\rho}_{ij}$, and the maximum log-likelihood under the unrestricted model is 
\begin{equation}
    \ell_1 = \sum_i\sum_j\sum_k \ \hat\rho_{ij}^{(k)} \log \bar\rho_{ij} + \left(1-\hat\rho_{ij}^{(k)}\right) \log \left(1-\bar\rho_{ij}\right).
\end{equation}

For each individual null hypothesis, the likelihood ratio statistic is $\lambda_\mathrm{LR}(\bm s) = -2(\ell_0(\bm s) - \ell_1)$. Let $\lambda_\mathrm{LR} = \min_{\bm s} \lambda_\mathrm{LR}(\bm s)$ as the smallest statistic value among all $2^{17}$ values. This is the ``most favorable'' choice of $\bm s$ that gives the highest possible likelihood under the null. To reject the original null hypothesis, it is equivalent to reject all individual null hypotheses, and in turn equivalent to reject this one hypothesis with the smallest statistic value.

The asymptotic distribution of $\lambda_\mathrm{LR}$ is in general a mixture chi-squared distribution and difficult to compute,. However, since the total number of degrees of freedom in the parameters is $31 \times 17 = 527$, this mixture chi-squared can have at most 527 degrees of freedom, which serves as an inflated estimation of the distribution of the test statistic. In other words, using the $1-\alpha$ quantile of $\chi^2_{527}$ will provide a conservative proxy critical value, in the sense that rejection using this proxy critical value will be sufficient for rejection for the test.

With this dataset, we get $\lambda_\mathrm{LR} = 3363.02$, and hence reject the null hypothesis with $p<0.001$. Therefore, the test provides overwhelming evidence against the null hypothesis that all individuals share the same taste, further supporting the Collective SS model.

\end{document}